\documentclass[10pt]{article}

\usepackage[usenames,svgnames,table]{xcolor}
\definecolor{myblue}{rgb}{0.1,0.1,0.5}
\definecolor{mygreen}{rgb}{0.2,0.5,0.1}

\usepackage{xspace}
\usepackage{framed}
\usepackage{xifthen}
\usepackage{tikz}
\usetikzlibrary{calc}
\usepackage{amsmath}
\usepackage{amsthm}
\usepackage{amssymb}
\usepackage{nicefrac}
\usepackage{algorithmic}
\usepackage{algorithm}
\usepackage{multirow}

\usepackage[backref=page]{hyperref}
\hypersetup{
    colorlinks = true,
    linkbordercolor = {white},
     urlcolor     = {myblue},
citecolor    = {mygreen}
}

\usepackage[mathscr]{euscript}

\usepackage[margin=0.85in]{geometry}

 \definecolor{shadecolor}{gray}{0.75}

\usepackage{enumitem}
\usepackage{comment}

\usepackage{bm}
\usepackage{apxproof}

\usepackage[round,sort,semicolon]{natbib}
\usepackage{varioref}

\renewcommand*{\backref}[1]{}
\renewcommand*{\backrefalt}[4]{%
    \ifcase #1 (Not cited.)%
    \or        (Cited on page~#2)%
    \else      (Cited on pages~#2)%
    \fi}
    
\newlength{\RoundedBoxWidth}
\newsavebox{\GrayRoundedBox}
\newenvironment{GrayBox}[1]%
   {\setlength{\RoundedBoxWidth}{.93\textwidth}
    \def\boxheading{#1}
    \begin{lrbox}{\GrayRoundedBox}
       \begin{minipage}{\RoundedBoxWidth}}%
   {   \end{minipage}
    \end{lrbox}
    \begin{center}
    \begin{tikzpicture}%
       \node(Text)[draw=black!20,fill=white,rounded corners,%
             inner sep=2ex,text width=\RoundedBoxWidth]%
             {\usebox{\GrayRoundedBox}};
        \coordinate(x) at (current bounding box.north west);
        \node [draw=white,rectangle,inner sep=3pt,anchor=north west,fill=white] 
        at ($(x)+(6pt,.75em)$) {\boxheading};
    \end{tikzpicture}
    \end{center}}     
\newenvironment{defproblemx}[2][]{\noindent\ignorespaces%
                                \FrameSep=6pt%
                                \parindent=0pt%
                \vspace*{-1.5em}
                \ifthenelse{\isempty{#1}}{%
                  \begin{GrayBox}{\textsc{#2}}%
                }{%
                  \begin{GrayBox}{\textsc{#2}  parameterized by~{#1}}%
                }
                \begin{tabular*}{\textwidth}{@{\hspace{.1em}} >{\itshape} p{1.8cm} p{0.8\textwidth} @{}}%
            }{
                \end{tabular*}%
                \end{GrayBox}%
                \ignorespacesafterend
            }

\newcommand{\NP}{\textsf{NP}}
\newcommand{\Pol}{\textsf{P}}
\newcommand{\FPT}{\textsf{FPT}}
\newcommand{\XP}{\textsf{XP}}
\newcommand{\NPH}{\textsf{NP}-hard}
\newcommand{\NPC}{\textsf{NP}-complete\xspace}

\newcommand{\tw}{{\sf tw}}

\newcommand{\card}[1]{\ensuremath{{\vert {#1} \vert }}}
\newcommand{\ca}[1]{\ensuremath{\mathcal{#1} }}
\newcommand{\set}[1]{\ensuremath{\left\{ {#1} \right\}}}
\newcommand{\fn}[3]{\ensuremath{{{#1} : {#2} \rightarrow {#3}}}}

\newcommand{\cO}{\mathcal{O}}

\newtheorem{lemma}{Lemma}[section]

\newtheorem{theorem}[lemma]{Theorem}

\newtheorem{observation}[lemma]{Observation}

\newtheorem{claim}[lemma]{Claim}

\theoremstyle{definition}
\newtheorem{remark}[lemma]{Remark}

\newtheoremrep{lemmaapx}[lemma]{Lemma}

\newcommand{\diam}{\ensuremath{\text{diam}}}
\newcommand{\dist}{\ensuremath{\text{dist}}}

\newcommand{\ext}{\ensuremath{\mathtt{ext}}}
\newcommand{\add}{\ensuremath{\mathtt{add}}}

\newcommand{\ttsup}[2]{\ensuremath{{#1}^{\texttt{#2}}}}

\newcommand{\blk}{\ensuremath{\mathtt{block}}}
\newcommand{\mms}{\ensuremath{\mathtt{mms}}}

\newcommand{\merge}{\ensuremath{\mathtt{merge}}}
\newcommand{\tsub}[2]{\ensuremath{{#1}_{\text{#2}}}}

\mathchardef\mh="2D

\newcommand{\aco}{\ensuremath{A_{\texttt{co}}}}
\newcommand{\asc}{\ensuremath{A_{\texttt{sco}}}}

\newcommand{\bco}{\ensuremath{A'_{\texttt{co}}}}

\newcommand{\maxval}{\ensuremath{W}}

\newcommand{\spropn}[2]{{\sc Prop-Str-}$({#1}, {#2})$-{\sc Compact-FD}}
\newcommand{\sefn}[2]{{\sc EF-Str-}$({#1}, {#2})$-{\sc Compact-FD}}
\newcommand{\scefn}[2]{{\sc Complete EF-Str-}$({#1}, {#2})$-{\sc Compact-FD}}
\newcommand{\spoefn}[2]{{\sc Complete EF-Str-}$({#1}, {#2})$-{\sc Compact-FD}}

\newcommand{\xcompact}[2]{$[\mathtt{X}]$-$({#1}, {#2})$-{\sc Compact-FD}}
\newcommand{\xanno}{$[\mathtt{X}]$-{\sc Annotated-FD}}

\newcommand{\compact}[2]{$({#1}, {#2})$-{\sc Compact-FD}}

\newcommand{\anno}{{\sc Annotated-FD}}

\newcommand{\mmsn}[2]{{\sc MMS-$({#1}, {#2})$-{\sc Compact-FD}}}
\newcommand{\smmsn}[2]{{\sc MMS-Str-$({#1}, {#2})$-{\sc Compact-FD}}}

\newcommand{\propn}[2]{{\sc Prop-}$({#1}, {#2})$-{\sc Compact-FD}}
\newcommand{\efn}[2]{{\sc EF-}$({#1}, {#2})$-{\sc Compact-FD}}
\newcommand{\cefn}[2]{{\sc Complete-EF-}$({#1}, {#2})$-{\sc Compact-FD}}
\newcommand{\poefn}[2]{{\sc PO-EF-}$({#1}, {#2})$-{\sc Compact-FD}}

\usepackage{balance} 

\begin{document}
\title{Fair Division of a Graph into Compact Bundles}

\author{Jayakrishnan Madathil \\ University of Glasgow, United Kingdom \\ \texttt{jayakrishnan.madathil@glasgow.ac.uk}} 
\date{}

\maketitle
\begin{abstract}
We study the computational complexity of fair division of indivisible items in an enriched model: there is an underlying graph on the set of items. And we have to allocate the items (i.e., the vertices of the graph) to a set of agents in such a way that (a) the allocation is fair (for appropriate notions of fairness) and (b) each agent receives a bundle of items (i.e., a subset of vertices) that induces a subgraph with a specific ``nice structure.'' This model has previously been studied in the literature with the nice structure being a connected subgraph. 
In this paper, we propose an alternative for connectivity in fair division.  
We introduce what we call compact graphs, and look for fair allocations in which each agent receives a compact bundle of items. Through compactness, we try to capture the idea that every agent must receive a set of ``closely related'' items. We prove a host of hardness and tractability results for restricted input settings with respect to fairness concepts such as proportionality, envy-freeness and maximin share guarantee. 
\end{abstract}


`

\section{Introduction}\label{sec:intro}
We study a fair allocation problem: $m$ items are to be allocated among $n$ agents. The items are indivisible, the agents have preferences over the items, and the allocation must be fair (for appropriate notions of fairness). A typical fair allocation problem thus far. But wait, there's more. The items to be allocated are the vertices of a graph, and we have to allocate the items in such a way that the allocation is fair \emph{and} each agent receives a bundle of items (i.e., a subset of vertices) that induces a subgraph with a specific ``nice structure.'' The fairness notions that we study are proportionality, envy-freeness and maximin share guarantee; and the nice structure that we study requires that each agent must receive a bundle of ``closely related'' items. 

While fair division of indivisible items has been studied extensively from both economics and computational perspectives~\citep{DBLP:books/daglib/0017734,DBLP:journals/corr/abs-2202-08713,DBLP:conf/ijcai/AmanatidisBFV22}, fair division of graphs is a relatively recent line of study, starting with the work of~\citet{DBLP:conf/ijcai/BouveretCEIP17}. 
The classic fair division literature assumes no relationship between the items; the items are independent of each other. 
But there are any number of scenarios where the items to be allocated, rather than being an assortment of isolated units, are related to one another. And we can often model such relationships between the items using a graph. The items, for example, may be rooms in a large building, and the graph models  adjacency between pairs of rooms; or the items may be a valuable collection of vintage postage stamps, and the graph models similarities between pairs of stamps on the basis of country, design or typography; or the items may be the topics for a course (to be divided between two professors who are planning to co-teach the course), and the graph models connections between the topics. Fair division of graphs attempts to capture precisely such settings.  

While fair division of graphs has received considerable attention in recent years~\citep{DBLP:conf/ijcai/DeligkasEGHO21,DBLP:journals/geb/BiloCFIMPVZ22,DBLP:conf/aaai/GrecoS20,DBLP:conf/aaai/IgarashiP19,DBLP:journals/siamdm/BeiILS22}, to the best of our knowledge, all work so far has  focused solely on \emph{connected} fair division. That is, each agent should receive a set of vertices that induces a connected subgraph. 
\citet{DBLP:conf/ijcai/BouveretCEIP17}, who introduced this model of fair division of graphs, considered the example of office spaces in a university being allocated to various research groups, 
where each research group should receive a contiguous set of offices. This scenario can easily be modelled as a fair-division-of-a graph problem where each agent must get a connected subgraph---agents represent various research groups, and the vertices and the edges of the graph respectively represent the offices and adjacency between pairs of offices. Notice that while contiguity might be useful, the ``closeness'' of the offices allocated to each group might be just as important. Simply demanding contiguity might leave a research group with a set of offices along a long and narrow corridor, which may not be an attractive proposition to the members of the group. In the absence of other constraints, the connectivity requirement, while desirable, might produce bundles with unwieldy topologies. 
In some cases, demanding connected bundles for every agent may be too stringent a requirement. For example, imagine a scenario with just one agent and two items that are isolated from each other, where  the agent has a utility of $1/2$ for each of the two items and a total utility of $1$ for the two items together. As we can see, demanding connectivity results in only one item being allocated to the agent, costing her half her total utility for the set of items. 

We can also think of the connectivity requirement in the fair division of graphs as a direct adaptation of the contiguity requirement in the fair division of divisible items, i.e., the cake cutting problem where each agent should receive a contiguous piece of the cake~\citep{stromquist1980cut,DBLP:journals/jair/GoldbergHS20}. But contiguity is a singularly important consideration in the cake cutting setting, for otherwise an agent ``who hopes only for a modest interval of cake may be presented instead with a countable union of crumbs''~\citep{stromquist1980cut}. That is not the case with graphs. Moreover, graphs are a rich combinatorial structure capable of modelling a variety of relationships among the items,  and connectivity of the items allocated to each agent may not always be the most important consideration. There is, however, a void when it comes to our understanding of fair division of graphs under constraints other than connectivity. In this paper, we take a first step towards filling this void. 

\subsection{{Our Contribution}}  

We make three contributions. (1) Propose meaningful alternatives for connectivity in the fair division of graphs. (2) Study the computational complexity of finding fair allocations under such alternative constraints, even on restricted inputs. The restrictions on the input are based on the nature of the valuation functions and the structure of the graph on the set of items. (3) Identify and exploit results and techniques from structural and algorithmic graph theory to design efficient algorithms for fair division of graphs. To that end, we propose an alternative for connectivity: We demand that each agent must receive a  \emph{compact} bundle of items. Through compactness, we try to capture the idea  that each agent must receive a set of closely related items.  

To formally define compact graphs, we introduce the following notation. For $n \in \mathbb{N}$, $[n]$  denotes the set $\set{1, 2,\ldots, n}$ and $[n]_0 = [n] \cup \set{0}$. For a graph $G$, $V(G)$ and $E(G)$ respectively denote the set of vertices and edges of $G$. For a path $P$ in $G$, the length of $P$ is the  number of edges in $P$. For vertices $z, z' \in V(G)$, the distance between $z$ and $z'$, denoted by $\dist_G(z, z')$ is the length of a shortest path between $z$ and $z'$. The diameter of $G$, denoted by $\diam(G)$, is the maximum distance between a pair of vertices, i.e., $\diam(G) = \max_{z, z' \in V(G)} \dist_G(z, z')$. For $z \in V(G)$ and a non-negative integer $\beta$, $B_G(z, \beta) = \set{z' \in V(G) ~|~ \dist_G(z, z') \leq \beta}$; we call $B_G(z, \beta)$ the ball of radius $\beta$ centred at $z$. 

\paragraph*{{Compact graphs: Definitions.}} For non-negative integers $\alpha$ and $\beta$, we say that a graph $G$ is \emph{$(\alpha, \beta)$-compact} if there exist vertices $z_1, z_2,\ldots,  z_{\alpha} \in V(G)$ such that for every $z \in V(G)$, $\dist(z, z_i) \leq \beta$ for some $i \in [\alpha]$, i.e., $V(G) = \bigcup_{i = 1}^{\alpha} B_G(z_i, \beta)$.\footnote{In graph theory literature, a set of vertices $\set{z_1, z_2,\ldots, z_{\alpha}} \subseteq V(G)$ such that $V(G) = \bigcup_{i = 1}^{\alpha} B_G(z_i, \beta)$ is often called a distance-$\beta$ dominating set~\citep{haynes2020topics}. So $(\alpha, \beta)$-compact graphs are precisely those graphs that have a distance-$\beta$ dominating set of size at most $\alpha$.} Intuitively, $G$ is $(\alpha, \beta)$-compact if $G$ can be ``covered by $\alpha$ many balls, each of \emph{radius} at most $\beta$.''  
We also define a related class of graphs called strongly $(\alpha, \beta)$-compact graphs. 
We say that a graph $G$ is \emph{strongly $(\alpha, \beta)$-compact} if there exist vertex subsets $V_1, V_2,\ldots, V_{\alpha} \subseteq V(G)$ such that $V(G) = \bigcup_{i = 1}^{\alpha} V_i$ and for every $i \in [\alpha]$ and for every $z, z' \in V_i$, $\dist_G(z, z') \leq \beta$. 
In this paper, we focus primarily on $(\alpha, \beta)$-compact graphs. (We consider the empty graph---the graph $G$ with $V(G) = E(G) = \emptyset$---to be both $(\alpha, \beta)$-compact and strongly $(\alpha, \beta)$-compact for every $\alpha, \beta \geq 0$.) 

\paragraph*{{Compact graphs: Examples.}} Observe that if a graph $G$ is strongly $(\alpha, \beta)$-compact, then it also $(\alpha, \beta)$-compact. But the converse need not hold. For example, a star is $(1, 1)$-compact, but not strongly $(1, 1)$-compact. On the other hand, if $G$ is $(\alpha, \beta)$ compact, then it is strongly $(\alpha, 2\beta)$-compact.  When $\beta = 0$, the two definitions coincide. For example, a graph is $G$ is (strongly) $(\alpha, 0)$-compact if and only if $G$ has at most $\alpha$ many vertices. And $G$ is strongly $(1, 1)$-compact if and only if $G$ is a clique; and $G$ is strongly $(1, \beta)$-compact if and only if $G$ has diameter at most $\beta$. In particular, an $m$-vertex graph $G$ is connected if and only if $G$ is (strongly) $(1, m-1)$-compact. Similarly, a graph $G$ with $\alpha$ connected components is strongly $(\alpha, \beta)$-compact if and only every component has diameter at most $\beta$. In particular, a cluster graph with $\alpha$ components is strongly $(\alpha, 1)$-compact;  a cluster graph is a graph in which every connected component is a clique. 
As these examples show, the definition of (strongly) compact graphs is expansive enough to accommodate several natural restrictions on the structure of $G$, including properties of graphs such as number of vertices, connectivity, bounded diameter etc. 

\subsubsection{{Our Model}} We study fair division problems of the following type. A typical instance of our problem consists of a triplet $(G, N, \ca{V})$, where $G$ is an $m$-vertex graph, $N$ is a set of agents with $\card{N} = n$, and $\ca{V}$ is an $n$-tuple of functions $(v_i)_{i \in N}$, where $\fn{v_i}{2^{V(G)}}{\mathbb{Q}_{\geq 0}}$. For $i \in N$, $v_i$ is called the valuation function or utility function of agent $i$. We assume throughout that the valuation functions are additive (unless otherwise stated). That is, for $i \in N$, $S \subseteq V(G)$, we have $v_i(S) = \sum_{z \in S}v_i(\set{z})$ and $v_i(\emptyset) = 0$. When $S = \set{z}$ is a singleton set, we omit the braces and simply write $v_i(z)$.  We call the set of vertices of $G$ goods or items as well. Throughout the paper, we use $m$ for the number of items and $n$ for the number of agents. And we use $N$ or $[n]$ for the set of agents. 

\begin{table}[]
\centering
\scalebox{0.83}{
\begin{tabular}{|c|cc|c|c|c|c|}
\hline
                   & \multicolumn{2}{c|}{$\beta = 0$}                                                                                                                                       &                                                                                            &                                                                                            &                                                                                                                                                                                    &                                                                                                                                                                               \\ \cline{2-3}
\multirow{-2}{*}{} & \multicolumn{1}{c|}{$\alpha = 1$}                                                       & $\alpha = 2$                                                                 & \multirow{-2}{*}{\begin{tabular}[c]{@{}c@{}}$\alpha \geq 1$\\ $\beta \geq 1$\end{tabular}} & \multirow{-2}{*}{\begin{tabular}[c]{@{}c@{}}$\alpha \geq 3$\\ $\beta \geq 0$\end{tabular}} & \multirow{-2}{*}{\begin{tabular}[c]{@{}c@{}}$\alpha = 1$\\ $\beta \geq 0$\end{tabular}}                                                                                            & \multirow{-2}{*}{\begin{tabular}[c]{@{}c@{}}$\alpha \geq 1$\\ $\beta \geq 0$\end{tabular}}                                                                                    \\ \hline
Prop               & \multicolumn{1}{c|}{\cellcolor[HTML]{C0C0C0}\Pol}                                          & \begin{tabular}[c]{@{}c@{}}\NPH\\ (strong)\\ follows from \\ \cite{DBLP:conf/atal/FerraioliGM14}\end{tabular} & \begin{tabular}[c]{@{}c@{}}\NPH\\ (weak)\\ follows from \\ \cite{DBLP:conf/sigecom/LiptonMMS04}\end{tabular}             & \cellcolor[HTML]{C0C0C0}\begin{tabular}[c]{@{}c@{}}\NPH\\ (strong)\end{tabular}             & \cellcolor[HTML]{C0C0C0}\begin{tabular}[c]{@{}c@{}}(a) \FPT\ w.r.t. $n$\\ when $G$ \\ is a path\\ \\ (b) \XP\ w.r.t. \\ number of \\ agent types \\ when $G$ \\ is a path\end{tabular} & \cellcolor[HTML]{C0C0C0}\begin{tabular}[c]{@{}c@{}}(a) \XP\ w.r.t.\\ $n + \Delta(G)$\\ even for\\ arbitrary \\ valuations\\ \\ (b) pseudo-\XP\ w.r.t.\\ $n + \tw(G)$ \\ for integer \\ valuations \\ {} \end{tabular} \\ \cline{1-3} \cline{6-6}
EF                 & \multicolumn{1}{c|}{\begin{tabular}[c]{@{}c@{}}\Pol${}^\blacktriangle$\\ follows from \\ \cite{DBLP:journals/mss/GanSV19}\end{tabular}} &                                                                             &                                                                                            & \cellcolor[HTML]{C0C0C0}                                                                   &                                                                                                                                                                                   & \cellcolor[HTML]{C0C0C0}\begin{tabular}[c]{@{}c@{}}(c) pseudo-\XP\ w.r.t\\ $n$\\ when $G$\\ is planar, for \end{tabular}                                                              \\ \cline{1-6}
MMS                & \multicolumn{1}{c|}{\cellcolor[HTML]{C0C0C0}\Pol}                                          &                                                                             &                                                                                           &                                                                                            &                                                                                                                                                                                 & \multicolumn{1}{c|}{\cellcolor[HTML]{C0C0C0} integer valuations}  \\ \hline
\end{tabular}
}
\caption{A summary of results for $(\alpha, \beta)$-compact allocations for various choices of $\alpha$ and $\beta$. The results proved in this paper are shaded in grey. Each row corresponds to a fairness concept. The last two columns contain  our main algorithmic results; all the other columns discuss the polynomial time solvability or (weak or strong)-\NPH ness of the problems depending on various choices of $\alpha$ and $\beta$. For a graph $G$, $\Delta(G)$ and $\tw(G)$ denote the maximum degree and treewidth of $G$, respectively. The ``pseudo'' in the last column indicates that the runtime of the algorithm depends on the valuation functions. The $\blacktriangle$ in the (EF, $\alpha = 1, \beta = 0$) cell indicates the additional constraint that each agent be allocated exactly one item.}
\label{table:results}
\end{table}

\paragraph*{Compact allocations.} Consider $(G, N, \ca{V})$. An allocation is a function that assigns pairwise disjoint subsets of $V(G)$ to the agents. That is, an allocation is a function $\fn{\pi}{N}{2^{V(G)}}$ such that $\pi(i) \cap \pi(j) = \emptyset$ for distinct $i, j\in N$. Consider an allocation $\pi$. We call $\pi(i)$ agent $i$'s bundle under the allocation $\pi$. We say that $\pi(i)$ is $(\alpha, \beta)$-compact if $G[\pi(i)]$ is $(\alpha, \beta)$-compact. And we say that an allocation $\pi$ is \emph{$(\alpha, \beta)$-compact} if $\pi(i)$ is $(\alpha, \beta)$-compact for every $i \in N$. The definition of a strongly $(\alpha, \beta)$-compact allocation is analogous. 

\paragraph*{Fairness concepts. } We consider three well-studied fairness concepts---proportionality, envy-freeness and maximin fairness. An allocation $\pi$ is \emph{proportional} if $v_i(\pi(i)) \geq (1/n) \cdot v_i(V(G))$ for every $i \in N$, i.e., every agent receives at least a $1/n$ fraction of her utility for the whole graph. And we say that $\pi$ is \emph{envy-free} if $v_i(\pi(i)) \geq v_i(\pi(j))$ for every pair of distinct agents $i, j \in N$, i.e., every agent prefers her own bundle to that of the other agents. 
We now define maximin fair allocations by suitably adapting the definition introduced by~\cite{budish2011combinatorial}. For $(G, N, \ca{V})$, let $\Pi(G, N, \ca{V})$ be the set of all allocations $\fn{\phi}{N}{2^{V(G)}}$. Let $\Gamma \subseteq \Pi(G, N, \ca{V})$ be any non-empty set of allocations. For every agent $i \in N$, we define the $\Gamma$-maximin share guarantee of $i$, denoted by $\Gamma \mh \mms_i(G, N, \ca{V})$, as follows:
\[
\Gamma \mh \mms_i(G, N, \ca{V}) = \max_{\pi \in \Gamma} \min_{j \in N} v_i(\pi(j)). 
\]
Informally, $\Gamma \mh \mms_i(G, N, \ca{V})$ is the maximum utility agent $i$ could guarantee for herself if agent $i$ were to allocate the items, with the caveat that the allocation be in $\Gamma$, and  allowed to choose only the least valued bundle for herself. We say that an allocation $\fn{\phi}{N}{2^{V(G)}}$ is \emph{$\Gamma$-maximin fair} if $v_i(\pi(i)) \geq \Gamma \mh \mms_i(G, N, \ca{V})$. We are  interested in $\Gamma$-maximin fair allocations for the case when $\Gamma$ is the set of all $(\alpha, \beta)$-compact allocations, which we denote by $\Gamma((\alpha, \beta) \mh \mathtt{com})$. For convenience, when $(G, N, \ca{V})$ and $\Gamma$ are clear from the context, we may simply write $\mms_i$ instead of $\Gamma \mh \mms_i(G, N, \ca{V})$. 

\paragraph*{Economic efficiency constraints.} We say that an allocation $\pi$ is complete if $\pi$ allocates all the vertices of $G$, i.e., $\bigcup_{i \in N} \pi(i) = V(G)$. And we say that $\pi$ is Pareto-optimal if there exists no allocation $\pi'$ such that $v_i(\pi'(i)) > v_i(\pi(i))$ for some $i \in N$ and $v_j(\pi'(j)) \geq v_j(\pi(j))$ for every $j \in N \setminus \set{i}$. 

\paragraph*{Computational questions.} Let $\alpha$ and $\beta$ be fixed integers. We are interested in computational problems that take an instance $(G, N, \ca{V})$ as input, and the question is to decide if $(G, N, \ca{V})$ admits an allocation that is (strongly) $(\alpha, \beta)$-compact and fair? Depending on the fairness concept, we have the following specific problems. 
\begin{itemize} 
\item \propn{\alpha}{\beta}: Does $(G, N, \ca{V})$  admit an allocation that is proportional and $(\alpha, \beta)$-compact? 
\item \efn{\alpha}{\beta}: Does $(G, N, \ca{V})$ admit an allocation that is envy-free and $(\alpha, \beta)$-compact? As the empty allocation (where no item is allocated) is trivially envy-free, we often combine envy-freeness with efficiency constraints such as completeness or Pareto-optimality. We thus have the corresponding computational questions: \cefn{\alpha}{\beta} and \poefn{\alpha}{\beta}, where in addition to envy-freeness and $(\alpha, \beta)$-compactness, the allocation must respectively be complete and Pareto-optimal. 
\item \mmsn{\alpha}{\beta}: Does $(G, N, \ca{V})$ admit an allocation that is $\Gamma((\alpha, \beta) \mh \mathtt{com})$-maximin fair and $(\alpha, \beta)$-compact? 
\end{itemize}
We define the strongly compact variants of the problems analogously. Notice that $\alpha$ and $\beta$ are fixed constants and not part of the input. 

\subsubsection{Our Results}
We prove a host of hardness and algorithmic results for $[\mathtt{X}]$-\compact{\alpha}{\beta}, where we use $[\mathtt{X}]$ as a placeholder for one of the three fairness concepts discussed above. See Table~\ref{table:results} for a quick summary. Recall that $n$ is the number of agents and $m$ is the number of vertices in the input graph. We first discuss the polynomial time solvability versus (weak and strong) \NPH ness of the problems. These results are further divided into cases based on the choice of $\alpha$ and $\beta$. The hardness results hold only for problems corresponding to proportionality and envy-freeness, and they all hold for additive valuations. And then we move on to our algorithmic results, specifically, \FPT\ and \XP\ algorithms w.r.t. a combination of parameters, including the number of agents, maximum degree and treewidth of the graph. We now list our results one by one. First, \NPH ness. 
\begin{enumerate}
\item The well-known reduction from {\sc Partition} that shows the \NPH ness of checking whether a proportional allocation or an envy-free allocation exists (in the classic fair allocation setting)~\citep{DBLP:conf/sigecom/LiptonMMS04} extends to our setting as well. We need only represent the set of items as a clique. Thus \propn{\alpha}{\beta}, \cefn{\alpha}{\beta} and \poefn{\alpha}{\beta} are all \NPH\ for every $\alpha, \beta \geq 1$. These results hold even when there are only two agents with identical additive valuations. Notice, however, that a reduction from {\sc Partition} only implies weak \NPH ness. Also, this result does not cover the case when $\beta = 0$.  
\item The case of $\beta = 0$ warrants special attention. In this case, every agent can receive at most $\alpha$ items as a graph $G$ is $(\alpha, 0)$-compact if and only if $G$ has at most $\alpha$ vertices. So we can ignore the underlying graph, and treat the problem as an instance of the classic fair division problem with the additional constraint that every agent be allocated at most $\alpha$ items. In this case, the complexities of the problems vary depending on the choice of $\alpha$.
\begin{enumerate}
\item If $\alpha = 1$, then every agent can receive at most one item. In this case, for all three fairness notions, the corresponding problems are polynomial time solvable. We can reduce each problem to a matching problem in an agents-items bipartite graph. For envy-freeness, we have the additional constraint that every agent must receive exactly one item; and the polynomial time algorithm is due to \cite{DBLP:journals/mss/GanSV19}. 

\item The case of $\alpha = 2$:  The problem corresponding to proportionality, i.e., \propn{2}{0}, is strongly \NPH. This is implied by a result due to~\cite{DBLP:conf/atal/FerraioliGM14}. The complexity of the problems corresponding to envy-freeness and maximin fairness, however, are still open.  
\item The case when $\alpha \geq 3$ is covered by the next result. 
\end{enumerate}

\item For every $\alpha \geq 3$ and $\beta \geq 0$, we prove that the problems corresponding to (i) proportionality, (ii) envy-freeness + completeness and (iii) envy-freeness + Pareto-optimality are all  strongly \NPH. These results hold even for edgeless graphs. 
\end{enumerate}
Now, the algorithmic results.
\begin{enumerate}
\setcounter{enumi}{3}
\item  If $\alpha = 1$ and the graph $G$ is a path, then for proportionality, the corresponding problem admits an algorithm with runtime $2^{n} \cdot m^{\cO(1)}$, and hence is fixed-parameter tractable (\FPT) with respect to the number of agents. We also extend this result to an algorithm with runtime $n^{p} \cdot m^{\cO(1)}$, where $p$ is the number of \emph{types} of agents. Two agents are of the same type if their valuations are identical. These results hold for arbitrary valuations. Our algorithms are based on a straightforward dynamic programming procedure, which is only a slight adaptation of the algorithm of~\cite[Theorem 3.3]{DBLP:conf/ijcai/BouveretCEIP17} for proportional and \emph{connected} fair division. But the fact that $\alpha = 1$ is crucial for these algorithms to work.  
\item For all three fairness concepts and for all values of $\alpha$ and $\beta$, we design algorithms with runtime $m^{\cO(\alpha n)} \cdot 2^{\alpha n \cdot \Delta^{\cO(\beta)}}$ for the corresponding problems, where $\Delta$ is the maximum degree of $G$. We thus have slicewise polynomial time (\XP) algorithms  parameterized by $n + \Delta$. 
The algorithms follow from the simple observation that for any vertex $z \in V(G)$, the number of vertices that are at distance at most $\beta$ from $z$ is at most $\Delta^{\beta + 1}$. Therefore, we can enumerate all $(\alpha, \beta)$-compact allocations in time $m^{\cO(\alpha n)} \cdot 2^{\alpha n \cdot \Delta^{\cO(\beta)}}$. 
\item For all three fairness fairness concepts, we design pseudo-\XP\ algorithms with respect to the parameter $\tw + n$, where $\tw$ is the treewidth of $G$. We design a single dynamic programming algorithm that works for all three fairness concepts. We assume here that the valuations are integer-valued. Let $\maxval$ be the maximum value any agent has for the whole graph, i.e., $\maxval = \max_{i \in N}v_i(V(G))$. Our algorithms have a runtime of $(\tw + n)^{\cO(\tw + n)}  \beta^{\cO(\tw)} \cdot m^{\cO(\alpha n)} \cdot \maxval^{\cO(n^2)}$. In particular, when the valuations are polynomially bounded, we have \XP\ algorithms. These algorithms yield a number of interesting corollaries, which we discuss below. 
\begin{enumerate}
\item {\bf Welfare maximisation:} We can use our dynamic programming procedure to compute $(\alpha, \beta)$-compact allocations that maximise welfare. For example, we can answer questions such as this: Among all the allocations that maximise the utilitarian social welfare, is there one that is also $(\alpha, \beta)$-compact?  
The utilitarian social welfare of an allocation $\pi$ is the sum of the utilities of the agents have under $\pi$, i.e., $\sum_{i \in N} v_i(\pi(i))$. 
\item {\bf Connected fair division:} As noted earlier, an $m$-vertex graphs is connected if and only if it is $(1, m-1)$ compact. Hence our algorithms work for connected fair division as well, with a runtime of $(\tw + n)^{\cO(\tw + n)} \cdot  m^{\cO(\tw + n)} \cdot \maxval^{\cO(n^2)}$. 
\item {\bf Implications for planar graphs (and more):} On planar graphs, our algorithms run in time $n^{\cO(n)} \cdot m^{\cO(\alpha n)} \cdot \maxval^{\cO(n^2)}$, i.e., a pseudo-\XP\ algorithm parameterized by $n$. To derive this result, we leverage the  well-known fact that the treewidth of a planar graph is $\cO(D)$, where $D$ is the diameter of the graph~\citep{DBLP:journals/algorithmica/Eppstein00}, and the fact that the diameter of every connected component of an $(\alpha, \beta)$-compact graph is $\cO(\alpha \beta)$.\footnote{We note that the classes of bounded degree graphs, bounded treewidth graphs and planar graphs are pairwise-incomparable. That is, no class is contained in either of the two other classes. Hence an algorithm for problems restricted to one of these classes does not automatically imply a similar result for the other classes. In particular, the reliance on bounding treewidth by a function of the diameter is crucial for deriving our algorithm on planar graphs.} In fact, such a runtime is possible not just on planar graphs, but on all graphs whose treewidth is bounded by a function of its diameter. In particular, if $\ca{F}$ is a minor-closed family of graphs and $\ca{F}$ excludes an apex-graph, then on any graph $G \in \ca{G}$, our algorithms run in time $f(n) \cdot m^{\cO(\alpha n)} \maxval^{\cO(n^2)}$. This follows from a deep result in structural graph theory due to~\cite{DBLP:journals/algorithmica/Eppstein00}, who characterised all minor-closed graph classes that have the ``diameter-treewidth property.''\footnote{A graph $G$ is an apex graph if $G - v$ is planar for some $v \in V(G)$. That is, apex graphs are ``planar plus 1 vertex'' graphs. Consider a family of graphs $\ca{F}$. We say that $\ca{F}$ has the diameter-treewidth property if there exists a function $\fn{f}{\mathbb{R}}{\mathbb{R}}$ such that $\tw(G) \leq f(\diam(G))$ for every $G \in \ca{F}$. Notice that the function $f$ depends only on the family $\ca{F}$. \cite{DBLP:journals/algorithmica/Eppstein00} proved that if $\ca{F}$ is a minor closed family of graphs, then $\ca{F}$ has the diameter-treewidth property if and only if $\ca{F}$ excludes an apex graph. For more details, see~\citep{DBLP:journals/algorithmica/Eppstein00} or \citep{DBLP:journals/algorithmica/DemaineH04}.} 
\end{enumerate}
\end{enumerate} 

\paragraph*{Results for strongly compact variants.} All our hardness results mentioned above hold for the strongly compact variants of the respective problems as well. So do (i) our algorithm for proportional allocations on paths when $\alpha = 1$ and (ii) the \XP\ algorithms parameterized by $n + \Delta$. In addition, we show that for proportionality and for envy-freeness + completeness, the strongly compact variants of the corresponding problems are strongly \NPH\ for every $\alpha, \beta \geq 1$. (Notice that our strong \NPH ness results for the compact variants only cover the case of $\alpha \geq 3$.) 

\paragraph*{Runtime of our algorithms.} The runtimes of the algorithmic results discussed above have an exponential dependence on $n$ (the number of agents). But this need not be a disqualifying factor as several fair division settings only have a constant number of agents, and often just two agents. Common examples of 2-agent settings, as noted by~\cite{DBLP:journals/siamcomp/PlautR20}, include divorce settlements, inheritance division and international border disputes. In fact, 2-agent setting is one of the most intensively studied special cases in fair division; see, for example, \citep{DBLP:journals/scw/BramsF00,DBLP:journals/orf/BramsKK22}. We would also like to point out that the exponential dependence of the runtime of our treewidth-based algorithm on both $n$ and $\tw$ is also essential; notice that our \NPH ness results for proportional or envy-free, compact allocations when $\alpha \geq 3$ hold for graphs of treewidth zero (even when the valuations are additive and polynomially bounded).  The runtime of our treewdith-based algorithm depends also on $\maxval$, the maximum valuation of an agent for the whole graph. This may be inevitable as well. 
As shown by~\cite{DBLP:conf/aaai/IgarashiP19} in the context of connected fair division, unless $\Pol = \NP$, there is no polynomial time algorithm that finds a Pareto-optimal allocation even when the valuations are binary and additive and the graph is a tree with bounded pathwidth or bounded diameter. We believe that a similar result should hold for compact allocations as well. Also, as noted above, we can use our algorithms to find welfare-maximising allocations. Our algorithms can thus compute allocations with three properties: fairness, compactness and economic efficiency. We would like to emphasise that for algorithms that find efficient allocations,  ``pseudo'' runtimes are rather common; see, for example~\cite{DBLP:conf/sigecom/BarmanKV18,DBLP:journals/eor/AzizHMS23}. 

\paragraph*{Choice of parameters.} We would like to add a word about the parameterizations that we use in our algorithmic results. As noted above, in many fair allocation problems, the number of agents is small, and therefore, parameterizations by the number of agents or agent types are quite reasonable. In particular, such parameterizations have been used by \cite{DBLP:conf/ijcai/BouveretCEIP17} and \cite{DBLP:conf/ijcai/DeligkasEGHO21} in the context of connected fair division of graphs. Now, about the the graph parameters that we use. It is quite common in algorithmic graph theory to leverage structural parameters of the input graph (maximum degree, degeneracy, treewidth etc.) to design efficient algorithms. The two graph parameters that we use---maximum degree and treewidth---have been leveraged this way in the parameterized algorithms literature. Maximum degree, while being a simple and reasonable parameter, is perhaps less interesting because a number of problems remain \NP-hard on graphs of constant degree. But treewidth is a completely different story. It is considered as a natural measure of ``sparsity'' of the graph, and has been studied extensively.  
Treewidth is related to a number of other graph parameters, and it is the source of several deep theorems in graph theory. See, for example, the Festschrift edited by~\cite{DBLP:conf/birthday/2020bodlaender} for the many applications of treewidth in algorithm design. We must also add that connected fair division of graphs has previously been studied on special graphs such as paths, stars, trees, cycles etc., which are all graphs of treewidth 1 or 2~\citep{DBLP:conf/ijcai/BouveretCEIP17,DBLP:conf/ijcai/LoncT18}. Also, parameters such as treewidth and cliquewidth have been used as structural parameters in connected fair division~\citep{DBLP:conf/ijcai/DeligkasEGHO21}. We discuss some of these works below. So it is very much in line with the existing literature to use treewidth as a parameter to design algorithms. In particular, one of our goals was to apply results and techniques from graph theory literature to problems in fair division. And for this purpose, treewidth is arguably the most suitable choice. 

\subsection{Related Work} 
Modern fair division literature goes back at least to the pioneering work of \cite{steinhaus}. For an overview of the relevant literature, we refer the reader to the surveys by~\cite{DBLP:books/sp/16/LindnerR16} on the divisible goods setting; and by~\cite{DBLP:journals/corr/abs-2202-08713} or \cite{DBLP:conf/ijcai/AmanatidisBFV22} on the indivisible goods setting. 

There is, in particular, a large volume of literature on fair division under connectivity constraints. 
As for the cake cutting setting (i.e., divisible goods), \cite{stromquist1980cut} proved that an envy-free allocation of the cake into connected pieces exists, but there is no finite protocol that can find such an allocation~\citep{DBLP:journals/combinatorics/Stromquist08}.  
Coming to the indivisible goods setting, the work of~\cite{DBLP:conf/ijcai/BouveretCEIP17}, which formally introduced this line of study, contains several algorithmic and hardness results for the connected fair division of a graph. They showed that finding a connected allocation that is envy-free or proportional is \NPH, even when the graph is a path. On the other hand, a maximin allocation always exists and can be computed in polynomial time if the graph is a tree. They also designed several \FPT\ and \XP\ algorithms with respect to parameters such as number of agents and number of agent types, for special cases of the graph. This was followed by the work of~\cite{DBLP:conf/ijcai/LoncT18}, who studied the existence and complexity of maximin fair allocations on cycles. Among other results, they proved that when there are only three agents, there exists an allocation that guarantees each agent a constant fraction of her maximin share. Most related to our work is that of \cite{DBLP:conf/ijcai/DeligkasEGHO21}, who undertook an in-depth study of the parameterized complexity of connected fair division. They established a host of hardness and algorithmic results---under various combinations of parameters such as the number of agents, treewidth, cliquewidth etc.---with respect to fairness concepts such as proportionality, envy-freeness, EF1 and EFX. A number of recent works deal with the special case when the graph is a path, under various fairness notions. In particular, \cite{DBLP:journals/geb/BiloCFIMPVZ22} studied the existence and complexity of EF1 and EF2 allocations; \cite{DBLP:journals/corr/abs-2209-01348} studied EF1 allocations and showed that EF1 allocations always exists for any number of agents with monotone valuations;  \cite{DBLP:journals/dam/Suksompong19} studied approximate guarantees for proportionality, envy-freeness and equitability; and \cite{DBLP:journals/corr/abs-2101-09794} studied the complexity of finding allocations that are equitable up to one item (EQ1) combined with economic efficiency notions such as Pareto-optimality, non-wastefulness etc. 
In addition to these works that mainly deal with the existence or computational questions associated with connected fair division, \cite{DBLP:journals/siamdm/BeiILS22} studied the price of connectivity---the loss incurred by agents when the connectivity constraint is imposed---for maximin fair allocations. Their techniques included the application of several graph theoretic tools and concepts to fair division. 

Much of the literature on fair division under restricted input settings focus on restrictions on valuations, number of agents or relaxations of fairness notions. To the best of our knowledge, there are only a handful of works that consider restrictions on the \emph{structure} of bundles allocated to agents (other than connected fair division in the context of graphs). The typical restriction on bundles involves what can be called cardinality constraints. The canonical example is the house allocation problem, where each agent must receive exactly one item. ~\cite{DBLP:journals/mss/GanSV19} studied envy-freeness and \cite{DBLP:journals/orl/KamiyamaMS21} studied proportionality and equitability in the house allocation setting. \cite{DBLP:conf/atal/FerraioliGM14} et al. studied the problem of maximising egalitarian welfare, under the additional constraint that each agent be allocated exactly $k$ items. \cite{DBLP:conf/ijcai/BiswasB18} introduced a more general variant of cardinality constraints: the items are partitioned into groups and there is a cap on each group's contribution to any agent's bundle. And they showed that a $1/3$-approximate maximin fair allocation can be computed efficiently; \cite{DBLP:conf/eumas/HummelH22} recently improved the approximation guarantee to $1/2$. See the survey by \cite{DBLP:journals/sigecom/Suksompong21} for a comprehensive discussion of various constraints in fair division.

\section{Preliminaries}\label{sec:prelims}
We use this section only to collect all the notation and terminology we use in one place. Some of the terms that we define here will be used only in Section~\ref{sec:treedp}. 

\paragraph*{Sets, tuples and functions.} For $n \in \mathbb{N}$, $[n]$ denotes the set $\set{1, 2,\ldots, n}$, $[n]_0$ denotes the set $\set{0, 1, 2, \ldots, n}$; $\tsub{[n]}{od}$ and $\tsub{[n]}{ev}$ respectively denote the set of odd numbers in $[n]$ and the set of even numbers in $[n]$, i.e., $\tsub{[n]}{od} = [n] \cap \set{1, 3, 5,\ldots}$ and $\tsub{[n]}{ev} = [n] \cap \set{2, 4, 6,\ldots}$. We use the shorthand $(z_i)_{i \in [n]}$ for the tuple $(z_1, z_2,\ldots, z_n)$. For sets $X, Y$, a function $\fn{f}{X}{Y}$ and a subset $Z \subseteq X$, $f \big|_{Z}$ denotes the restriction of $f$ to $Z$. That is, $\fn{f \big|_{Z}}{Z}{Y}$ is the function from $Z$ to $Y$ such that $f \big|_{Z}(z) = f(z)$ for every $z \in Z$. For sets $X, Y, X', Y'$ and a function $\fn{f}{X}{Y}$, $\ext(f, X', Y')$ denotes the set of all functions $\fn{g}{X \cup X'}{Y \cup Y'}$ such that $g(x) = f(x)$ for every $x \in X$. That is, for every $g \in \ext(f, X', Y')$, $g \big|_{X} = f$. 

\paragraph*{Partition.} Consider a non-empty set $S$. A partition of $S$ is a family $\ca{P} \subseteq 2^S$ of non-empty subsets of $S$ such that $\bigcup_{P \in \ca{P}} P = S$ and $P \cap P' = \emptyset$ for every distinct $P, P' \in \ca{P}$. For a partition $\ca{P}$ of $S$, each set $P \in \ca{P}$ is  called a block of $\ca{P}$. And for $z \in S$, we denote the block of $\ca{P}$ that contains $z$ by $\blk_{\ca{P}}(z)$, i.e., $\blk_{\ca{P}}(z) = P$ if $P \in \ca{P}$ and $z \in P$. 
When the partition $\ca{P}$ is clear from the context, we omit the subscript and simply write $\blk(z)$. 
Consider $z \in S$ and let $\ca{P}$ be a partition of $S \setminus \set{z}$. By $\add(z, \ca{P})$, we denote the set of all partitions of $S$ obtained from $\ca{P}$ by adding $z$ to one of the blocks of $\ca{P}$. That is, if $\ca{P} = \set{P_1, P_2,\ldots,P_r}$ is a partition of $S \setminus \set{z}$, then $\add(z, \ca{P}) = \set{\ca{P}^1, \ca{P}^2,\ldots, \ca{P}^r}$, where for each $i \in [r]$, $\ca{P}^i$ is the partition of $S$ defined as $\ca{P}^i = \set{P^i_1, P^i_2,\ldots, P^i_r}$ with $P^i_i = P_i \cup \set{z}$ and $P^i_j = P_j$ for every $j \in [r] \setminus \set{i}$. 
\paragraph*{Rooted partition.} We now define an annotated version of a partition called a rooted partition. A rooted partition $\ca{P}$ is nothing but a partition, but each block $P$ of $\ca{P}$ has a designated element, which we call the root of $P$. Formally, a rooted partition of $S$ is an ordered pair $(\ca{P}, R)$ such that $\ca{P}$ is a partition  of $S$, $R \subseteq S$ with $\card{R} = \card{\ca{P}}$ and for each block $P$ of $\ca{P}$, $\card{P \cap R} = 1$; we call the unique element of $P \cap R$ the root of $P$.

\paragraph*{Graphs.} All graphs is this paper are undirected.  
For a graph $G$, $V(G)$ and $E(G)$ respectively denote the vertex set and edge set of $G$. 
Consider a graph $G$. 
For a path $P$ in $G$, the length of $P$ is the number of edges in $P$. 
For vertices $z, z' \in V(G)$, the distance between $z$ and $z'$ in $G$, denoted by $\dist_G(z, z')$, is the length of a shortest path between $z$ and $z'$. 
If $z = z'$, then $\dist_G(z, z') = 0$. 
For a vertex $z \in V(G)$ and a non-negative integer $\beta$, we use $B(z, \beta)$ to denote the set of all vertices in $G$ that are at distance at most $\beta$ from $z$, i.e., $B_G(z, \beta) = \{z' \in V(G) ~|~ \dist_G(z, z') \leq \beta\}$. 
(We may omit the subscript $G$ when the graph is clear from the context.) 
The diameter of $G$, denoted by $\diam(G)$, is the maximum distance between any pair of vertices in $G$, i.e., $\diam(G) = \max_{z, z' \in V(G)} \dist(z, z')$.  For a graph $G$, a subgraph $H$ of $G$ is called a spanning subgraph if $V(H) = V(G)$. An acyclic spanning subgraph is called a spanning forest.  
A tree is a connected, acyclic graph. A rooted tree $T$ is a tree in which we designate exacly one vertex $x \in V(T)$ as the ``root'' of $T$; and we say that $T$ is rooted at $x$. For a tree rooted at $x \in V(T)$, and vertices $z, z' \in V(T)$, we say that $z$ is a descendant of $z'$ if $z'$ lies on the unique path in $T$ from $x$ to $z$. For a rooted tree $T$ and $z \in V(T)$, by the sub-tree of $T$ rooted at $z$, we mean the sub-tree of $T$ induced by $z$ and all descendants of $z$, and we root the sub-tree at $z$. 
A rooted forest is a graph in which each connected component is a rooted tree. For a rooted forest $G$, by the set of roots of $G$, we mean the set $R \subseteq V(G)$ that consists of the root of each connected component of $G$. That is, for each connected component of $H$ of $G$, $R \cap V(H) = \set{x_H}$, where $x_H$ is the root of $H$. For a connected graph $G$ and a vertex $z \in V(G)$, a BFS tree of $G$ rooted at $z$ is a spanning tree of $G$, produced by a breadth-first search traversal of $G$, starting at $z$. If $H$ is a BFS tree of $G$ rooted at $z$, then for every $z' \in V(G)$, we have $\dist_G(z, z') = \dist_H(z, z')$. 

\paragraph*{Merge of two graphs.} For graphs $G_1$ and $G_2$ with $V(G_1) = V(G_2)$, by the merge of $G_1$ and $G_2$, denoted by $\merge(G_1, G_2)$ we mean the (multi)graph $G$ obtained by the union of $G_1$ and $G_2$. That is, $V(G) = V(G_1) = V(G_2)$ and $E(G) = E(G_1) \cup E(G_2)$ (treated as a multiset). Thus, for $z, z' \in V(G)$, $zz' \in E(G)$ if and only if $zz' \in E(G_1)$ or $zz' \in E(G_2)$. And if  $zz' \in E(G_1)$ \emph{and} $zz' \in E(G_2)$, then $G$ contains two parallel edges between $z$ and $z'$. 

\paragraph*{Partition into the connected components. } For a graph $G$ and a non-empty set $Z \subseteq V(G)$, by the \emph{partition of $Z$ into the connected components of $G$}, we mean the partition $\ca{P}$ of $Z$ such that for $z, z' \in Z$, $z$ and $z'$ are in the same block of $\ca{P}$ if and only if they are in the same connected component of $G$. 

\paragraph*{Graphical representation of a partition.} For a non-empty set $S$ and a partition $\ca{P}$ of $S$, a graphical representation of $\ca{P}$, denoted by $\mathtt{G}(\ca{P})$, is an acyclic graph with vertex set $S$ such that $\ca{P}$ is precisely the partition of $S$ into the connected components of $\mathtt{G}(\ca{P})$. That is, each block of $\ca{P}$ corresponds to a connected component of $\mathtt{G}(\ca{P})$ and vice versa; and since $\mathtt{G}(\ca{P})$ is acyclic, each connected component is a tree. 
Notice that $\mathtt{G}(\ca{P})$ need not be unique. So by $\mathtt{G}(\ca{P})$, we only mean any fixed (arbitrarily chosen) acyclic graph that satisfies the required property. 

\paragraph*{Acyclic join of partitions.} For partitions $\ca{P}', \ca{P}''$ and $\ca{P}$ of $S$, we say that $\ca{P}$ is an acyclic join of $\ca{P}'$ and $\ca{P}''$ if the merge of $\mathtt{G}(\ca{P}')$ and $\mathtt{G}(\ca{P}'')$ is a graphical representation of $\ca{P}$. For rooted partitions $(\ca{P}', R'), (\ca{P}'', R'')$ and $(\ca{P}, R)$ of $S$, we say that $(\ca{P}, R)$ is an acyclic join of $(\ca{P}', R')$ and $(\ca{P}'', R'')$ if $\ca{P}$ is an acyclic join of $\ca{P}'$ and $\ca{P}''$ and $R = R' \cap R''$. We now make the following observation about the complexity of checking if a given partition is an acyclic join. 

\begin{observation}\label{obs:acyclic}
Given a non-empty set $S$ and rooted partitions $(\ca{P}', R'), (\ca{P}'', R'')$ and $(\ca{P}, R)$ of $S$, we can check in time $\card{S}^{\cO(1)}$ if $(\ca{P}, R)$ is an acyclic join of $(\ca{P}', R')$ and $(\ca{P}'', R'')$. This is possible because (i) we can construct $\mathtt{G}(\ca{P}')$ and $\mathtt{G}(\ca{P}')$ in time $\card{S}^{\cO(1)}$, and (ii) we can also construct the merge of $\mathtt{G}(\ca{P}')$ and $\mathtt{G}(\ca{P}')$ and verify whether it is a graphical representation of $\ca{P}$ in time $\card{S}^{\cO(1)}$, and (iii) we can also verify whether $R = R' \cap R''$ in time $\card{S}^{\cO(1)}$. 
\end{observation}

\paragraph*{Parameterized complexity. } We use standard terminology from parameterized complexity, for which we refer the reader to the book by~\cite{DBLP:books/sp/CyganFKLMPPS15}. A typical parameterized problem on graphs consists of an $m$-vertex graph $G$ and a non-negative integer $k$ as input, and we have to decide if $(G, k)$ satisfies a fixed property. The integer $k$ may be given explicitly as part of the input, or it may be a function of the input graph $G$. We call $k$ the parameter, and say that the problem is parameterized by $k$. We say that the problem is fixed-parameter tractable (or belongs to the class \FPT) if it admits an algorithm running in time $f(k) m^{\cO(1)}$, for some computable function $f$. We also call an algorithm with such a runtime an \FPT\ algorithm. And we say that the  problem belongs to the class \XP\ (slicewise polynomial time) if it admits an algorithm running in time $f(k) m^{g(k)}$ for some computable functions $f$ and $g$. 

\subsubsection*{Tree decomposition} Consider a graph $G$. A tree decomposition of \(G\) is a pair \((T, \{X_t ~|~ t \in V(T)\}) \), where $T$ is a tree and every node $t \in V(T)$ has a set $X_t \subseteq V(G)$ associated with it, such that \((T, \{X_t ~|~ t \in V(T)\}) \) satisfies the following conditions. (a) \(\bigcup_{t \in V(T)} X_t = V(G)\). 
(b) For every edge \(uv \in E(G)\), there exists a node \(t \in V(T)\) such that $u, v \in X_t$. (c) For every \(u \in V(G)\), the set \(\{t \in V(T) : u \in X_t\}\) induces a connected subtree of \(T\).
We call each $X_t$ a bag. 
 The {width} of a tree decomposition $(T, \set{X_t ~|~ t \in V(T)})$ is $\max_{t} \card{X_t} - 1$. The {treewidth} of a graph \(G\), denoted by $\tw(G)$ is the minimum width of a tree decomposition \(G\). When $G$ is clear from the context, we simply write $\tw$. 
 When talking about a tree decomposition $(T, \set{X_t ~|~ t \in V(T)})$ of a graph $G$, to avoid any confusion, we follow the common practice of referring to the vertices of $T$ as nodes, and reserve the term vertex exclusively for the vertices of $G$.   

We now define a particular kind of tree decompositions called nice tree decompositions, which come in handy while designing algorithms.  A tree decomposition $(T, \set{X_t ~|~ t \in V(T)})$ of a graph \(G\) is said to be a nice tree decomposition if it has the following properties. (a) $T$ is rooted tree. (b) If $t \in V(T)$ is the root or a leaf, then $X_t = \emptyset$. (c) Every non-leaf node $t$ of $T$ is one of the following four types: (c.1) introduce vertex node: $t$ has exactly one child, say $t'$, and $X_t = X_{t'} \cup \set{z}$ for some $z \notin X_{t'}$; (c.2) forget node: $t$ has exactly one child, say $t'$, and $X_t = X_{t'} \setminus \set{z}$ for some $z \in X_{t'}$; (c.3) introduce edge node: $t$ has exactly one child, say $t'$, and $X_t = X_{t'}$, and for some edge $zz' \in E(G)$ with $z, z' \in X_t$, the node $t$ is labelled with the edge $zz'$. (c.4) join node: $t$ has exactly two children, say $t'$ and $t''$, and $X_t = X_{t'} = X_{t''}$. 

Consider a graph $G$, a nice tree decomposition $(T, \{X_t ~|~ t \in V(T)\})$ of $G$ and two nodes $t, t' \in V(T)$, where $t'$ is the unique child of $t$. If $t$ is an introduce node with $X_t = X_{t'} \cup \set{z}$, then we say that the vertex $z$ is introduced at $t$ or that $t$ introduces the vertex $z$. Similarly, if $t$ is an introduce edge node labelled with an edge $zz' \in E(G)$, then we say that  $zz'$ is introduced at $t$ or that $t$ introduces $zz'$. And if $t$ is a forget node with $X_t = X_{t'} \setminus \set{z}$, then we say that the vertex $z$ is forgotten at $t$ or that $t$ forgets the vertex $z$. We assume throughout without loss of generality that in a nice tree decomposition, every vertex of $G$ is introduced at exactly one node of $T$ and forgotten at exactly one node of $T$; and that every edge of $G$ is introduced at exactly one node of $T$. For a  node $t$, we use $G_t$ to denote the subgraph of $G$ that consists of all the vertices and edges introduced in the subtree rooted at $t$. 
When designing an algorithm for a problem parameterized by treewidth, we assume that a nice tree decomposition of $G$ of width $\tw(G)$ is also given as part of the input. This assumption is harmless because given an $m$-vertex graph $G$, we can compute a nice tree decomposition of $G$ of width $\cO(\tw(G))$ in time $2^{\cO(\tw(G))} m^{\cO(1)}$~\citep{DBLP:journals/siamcomp/BodlaenderDDFLP16}.

\section{Preliminary Results and Observations}
In this section, we record a few straightforward results that more or less follow from the definition of (strongly) $(\alpha, \beta)$-compact graphs. We will use some of these results in Sections~\ref{sec:nph} and \ref{sec:treedp}. Specifically, results from Sections~\ref{sec:diameter}, \ref{sec:annotated} and \ref{sec:recognition} will be used in later sections. But Sections~\ref{sec:weak-NP}, \ref{sec:max-degree} and \ref{sec:one-zero} are standalone sections. These sections respectively discuss the weak \NPH ness of \xcompact{\alpha}{\beta}, the \XP\ algorithm for \xcompact{\alpha}{\beta}\ parameterized by the number of agents and the maximum degree of $G$, and the polynomial time algorithms for \xcompact{1}{0}. 

\subsection{NP-Hardness of Compact Fair Division}
\label{sec:weak-NP}
We show the \NPH ness of deciding whether a proportional or envy-free (strongly) $(\alpha, \beta)$-compact allocation exists. 
We note that the well-known reduction from {\sc Partition}~\citep{DBLP:conf/sigecom/LiptonMMS04} that shows the \NPH ness of {\sc Prop-FD} and {\sc EF-FD} can be extended to our setting as well. We just need to represent the set of items as a clique. For the sake of completeness, we state the result and sketch its proof. 

\begin{theorem}\label{thm:lipton}
For every fixed $\alpha, \beta \geq 1$, \propn{\alpha}{\beta}, \cefn{\alpha}{\beta} and \poefn{\alpha}{\beta} are all \NPH\ when the valuations are encoded in binary, and even when there are only two agents with identical valuations. So are the strongly compact variants of these three problems. 
\end{theorem}

\begin{proof}
To prove \NPH ness, we reduce from the {\sc Partition} problem, in which the input consists of $m$ non-negative integers $x_1, x_2, \ldots, x_m$ such that $\sum_{j \in [m]} x_j = 2W$; and the question is to decide if there exists a subset of indices $J \subseteq [m]$ such that $\sum_{j \in J} x_j = W$. Given an instance $(x_1, x_2,\ldots, x_m)$ of {\sc Partition}, we construct an instance $(G, N, \ca{V})$ of \xcompact{\alpha}{\beta} as follows, (where $[\mathtt{X}]$ is a placeholder for proportionality or envy-freeness and completeness or envy-freeness and Pareto-optimality). We take $N$ = [2]; and for every $j \in [m]$, introduce an item $z_j$. We take $G$ to be the clique with vertex set $\set{z_j ~|~ j \in [m]}$. Finally, for $i \in N$, we define the valuation function $\fn{v_i}{V(G)}{\mathbb{Q}_{\geq 0}}$, where $v_i(z_j) = x_j$. Notice that for any $Z \subseteq V(G)$, the subgraph $G[Z]$ is a clique and hence (strongly) $(\alpha, \beta)$-compact. Now, an allocation $\fn{\pi}{N}{2^{V(G)}}$ is proportional (or complete and envy-free, or Pareto-optimal and envy-free) if and only if $v_i(\pi(i)) = \sum_{z_j \in \pi(i)} x_j  = (1/2) \cdot  \sum_{j = 1}^m x_j = W$ for each $i \in [2]$. 
\end{proof}

\subsection{Compact Allocations of Bounded Degree Graphs}
\label{sec:max-degree}
We now consider the special case of \xcompact{\alpha}{\beta} when the maximum degree of the graph $G$ is bounded, where $[\mathtt{X}]$ is a placeholder for proportionality, envy-freeness or maximin fairness. Consider $(G, N, \ca{V})$, where the degree of each vertex of $G$ is at most $\Delta$. Observe then that for every $z \in V(G)$ and for every $j \geq 0$, the number of vertices of $G$ that are at a distance of $j$ from $z$ is at most $\Delta^j$. We thus have $\card{B_G(z, \beta)} \leq \sum_{j = 0}^{\beta} \Delta^{j} \leq \Delta^{\beta + 1}$. As an immediate consequence of this observation, we derive the following lemma. 

\begin{lemma}\label{lem:maxdegree}
Consider an instance $(G, [n], \ca{V})$ of \compact{\alpha}{\beta}. Then the  number of $(\alpha, \beta)$-comapct allocations of $(G, [n], \ca{V})$ is at most $m^{\alpha n} \cdot  2^{\alpha n \cdot \Delta^{\cO(\beta)}}$, where $\Delta$ is the maximum degree of $G$. Moreover, we can enumerate all $(\alpha, \beta)$-compact allocations in time $m^{\cO(\alpha n)} \cdot 2^{\alpha n \cdot \Delta^{\cO(\beta)}}$. 
\end{lemma}

\begin{proof}
To prove the lemma, we will bound the possible number of choices for an $(\alpha, \beta)$-compact allocation. 
Consider any $(\alpha, \beta)$-compact allocation $\fn{\pi}{N}{2^{V(G)}}$. Then for every $i \in [n]$, there exist vertices $z_i^1, z_i^2,\ldots, z_i^{\alpha} \in V(G)$ such that $\pi(i) = \bigcup_{j = 1}^{\alpha} B_{G[\pi(i)]}(z_i^j, \beta) \subseteq \bigcup_{j = 1}^{\alpha} B_{G}(z_i^j, \beta)$. 
And as we observed earlier, we have $\card{B_G(z, \beta)} \leq \Delta^{\cO(\beta)}$, and thus $\card{\pi(i)} \subseteq \alpha \cdot \Delta^{\cO(\beta)}$. 

Now, the number of choices for $(z_i^1, z_i^2,\ldots, z_i^{\alpha})$ is at most $m^{\alpha}$; and the number of choices for $\pi(i)$ such that $\pi(i) \subseteq \bigcup_{j = 1}^{\alpha} B_{G}(z_i^j, \beta)$ is at most $2^{\alpha \cdot \Delta^{\cO(\beta)}}$. Hence the number of choices for $\pi(i)$ is at most $m^{\alpha} \cdot 2^{\alpha \cdot \Delta^{\cO(\beta)}}$. 

Notice now that any allocation $\pi$ is characterised by the tuple $(\pi(1), \pi(2),\ldots, \pi(n))$. Therefore, the number of choices for $\pi$ is at most $\left(m^{\alpha} \cdot 2^{\alpha \cdot \Delta^{\cO(\beta)}}\right)^n = m^{\alpha n} \cdot 2^{\alpha n \cdot \Delta^{\cO(\beta)}}$. 

To enumerate all $(\alpha, \beta)$-compact allocations, we go over all possible choices for $(z_i^j))_{i \in [n], j \in [\alpha]}$ and consider all possible choices for $(\pi(1), \pi(2),\ldots, \pi(n))$, where $\pi(i) \subseteq \bigcup_{j = 1}^{\alpha} B_{G}(z_i^j, \beta)$; and we check if $(\pi(1), \pi(2),\ldots, \pi(n))$ constitutes an $(\alpha, \beta)$-compact allocation. 
\end{proof}

Notice that once we enumerate all $(\alpha, \beta)$-compact allocations, we can verify if at least one of those allocations is fair, with respect to the fairness concepts proportionality, envy-freeness and maximin-fairness. Notice in particular that once we enumerate all $(\alpha, \beta)$-compact allocations, we can compute $\Gamma((\alpha, \beta)\mh \mathtt{com}) \mh \mms_i(G, [n], \ca{V})$ for every $i \in [n]$. Thus, as an immediate consequence of Lemma~\ref{lem:maxdegree}, we have the following result. 

\begin{theorem}\label{thm:maxdegree}
For every $\alpha, \beta \geq 0$, \propn{\alpha}{\beta}, \efn{\alpha}{\beta} and \mmsn{\alpha}{\beta} admit algorithms that run in time $m^{\cO(\alpha n)} \cdot 2^{\alpha n \cdot \Delta^{\cO(\beta)}}$. 
\end{theorem}

\paragraph*{Analogous result for strongly compact variants.} We can use the same idea as in Lemma~\ref{lem:maxdegree} to solve the strongly $(\alpha, \beta)$-compact variants of the problems as well. If $\fn{\pi}{[n]}{2^{V(G)}}$ is a strongly compact allocation, then for every $i \in [n]$, there exist $V_i^1, V_i^2,\ldots, V_i^{\alpha}$ such that $\pi(i) = V_i^1 \cup V_i^2 \cup \cdots V_i^{\alpha}$ and for every $j \in [\alpha]$ and every $z, z' \in V_i^j$, we have $\dist_{G[\pi(i)]}(z, z') \leq \beta$. So, for any fixed $\hat z_i^j \in V_i^j$, we have $V_i^j \subseteq B_{G[\pi(i)]}(z_i^j, \beta) \subseteq B_G(z_i^j, \beta)$. Hence $\card{V_i^j} \leq \card{B_G(z_i^j, \beta)} \leq \Delta^{\cO(\beta)}$. Thus the number of strongly $(\alpha, \beta)$-compact allocations is also bounded by $m^{\alpha n} \cdot 2^{\alpha n \cdot \Delta^{\cO(\beta)}}$. To enumerate all such allocations, we go over all tuples $(z_i^j)_{i \in [n], j \in [\alpha]}$ such that $z_i^j \in V_i^j$. Moreover, notice that for each tuple $(V_i^1, V_i^2, \ldots, V_i^{\alpha})$, for each $j \in [\alpha]$ and $z, z' \in V_i^{j}$, we can indeed check in polynomial time whether $\dist_{G[\pi(i)]}(z, z') \leq \beta$, where $\pi(i) = V_i^1 \cup V_i^2 \cup \cdots V_i^{\alpha}$. That is, we can check whether $\pi(i)$ is indeed strongly $(\alpha, \beta)$-compact. We thus have the following result. 

\begin{theorem}\label{thm:maxdegree-s}
For every $\alpha, \beta \geq 0$, \spropn{\alpha}{\beta}, \sefn{\alpha}{\beta} and \smmsn{\alpha}{\beta} admit algorithms that run in time $m^{\cO(\alpha n)} \cdot 2^{\alpha n \cdot \Delta^{\cO(\beta)}}$.
\end{theorem}

\subsection{The Case of \texorpdfstring{$\bm{(\alpha, \beta) = (1, 0)}$}{(alpha, beta) = (1, 0)}}
\label{sec:one-zero}
Consider an instance $(G, N, \ca{V})$ of \xcompact{1}{0}. Observe that for any $(1, 0)$-compact allocation $\pi$, we have $\card{\pi(i)} \leq 1$ for every $i \in N$. That is, every agent must receive at most one item. In this case, we can reduce the problems corresponding to proportionality and maximin fairness to matching problems. For proportionality, we construct the agents-items bipartite graph $H$: there is a vertex in $H$ for every agent $i \in N$ and there is a vertex in $H$ for every item $z \in V(G)$. And there is an edge in $H$ between $i \in N$ and $z \in V(G)$ if $v_i(z) \geq (1/n) \cdot v_i(V(G))$. Notice that $(G, N, \ca{V})$ admits an allocation that is proportional and $(1, 0)$-compact if and only if $H$ has a matching that saturates all the agent-vertices, which we can check in polynomial time. 

For maximin fairness, we apply an identical procedure. Recall that $\Gamma((1, 0) \mh \mathtt{com})$ denotes the set of all $(1, 0)$ compact allocations. For convenience, in this section, we simply write $\mms_i$ for $\Gamma((1, 0) \mh \mathtt{com}) \mh \mms_i(G, N, \ca{V})$ for every $i \in N$. As before, we construct the agents-items bipartite graph $H$ but with the following difference. There is an edge in $H$ between $i$ and $z$ if and only if $v_i(z) \geq \mms_i$. 
The rest of the arguments are identical to that of proportionality. To construct $H$, we need to compute $\mms_i$. But this is quite straightforward, as $\mms_i = 0$ if $m < n$ and otherwise, $\mms_i$ is precisely the value agent $i$ has for her $n$th most valued item. We relegate a formal proof of this fact to the appendix. 

\begin{toappendix}
Recall that $\Gamma((1, 0) \mh \mathtt{com})$ denotes the set of all $(1, 0)$ compact allocations. For convenience, in this section, we simply write $\mms_i$ for $\Gamma((1, 0) \mh \mathtt{com}) \mh \mms_i(G, N, \ca{V})$ for every $i \in N$. The following lemma says that $\mms_i = 0$ if $m < n$;  and otherwise, $\mms_i$ is precisely the value agent $i$ has for her $n$th most valued item.
\begin{lemma}
Fix $i \in [n]$ and consider an ordering $\sigma^i = (z^i_1, z^i_2, \ldots, z^i_m)$ of the items such that $v_i(z^i_1) \geq v_i(z^i_{2}) \geq \cdots \geq v_{i}(z^i_m)$. If $m < n$, then $\mms_i = 0$; otherwise $\mms_i = v_i(z^i_n)$. 
\end{lemma}
\begin{proof}
Fix $i \in [n]$. Suppose first that $m < n$. Then, for any $(1, 0)$-compact allocation $\fn{\pi'}{[n]}{2^{V(G)}}$, there exists an agent $j' \in [n]$ such that $\pi'(j') = \emptyset$. Hence $v_i(\pi'(j')) = 0$, and in particular, $\min_{j \in N}v_i(\pi'(j)) = v_i(\pi'(j')) = 0$. Thus, by the definition of $\mms_i$, we have $\mms_i = 0$. Suppose now that $m \geq n$. Consider the ordering $\sigma^i = (z^i_1, z^i_2, \ldots, z^i_m)$. Recall that we have $v_i(z^i_1) \geq v_i(z^i_{2}) \geq \cdots \geq v_{i}(z^i_m)$. We will show that $\mms_i = v_i(z^i_n)$. Let us first see that $\mms_i \geq v_i(z^i_n)$. Consider the allocation $\pi$ defined by $\pi(j) = z^i_j$ for every $j \in [n]$. 
First, the allocation $\pi$ is $(1, 0)$-compact as each agent receives at most one item. Second, we have $v_i(z^i_n) \leq v_i(z^i_j)$ for every $j \in [n]$, and therefore, $\min_{j \in [n]} v_i(\pi(j)) = v_i(\pi(n)) = v_i(z^i_n)$. Thus, by the definition of $\mms_i$, we have $\mms_i \geq \min_{j \in [n]} v_i(\pi(j)) = v_i(z^i_n)$. Let us now see that $\mms_i \leq v_i(z^i_n)$. Let $\phi$ be a $(1, 0)$-compact allocation that maximises $\min_{j \in [n]} v_i(\phi(j))$, and we thus have $\mms_i = \min_{j \in [n]} v_i(\phi(j))$. If there exists an agent, say $j' \in [n]$, such that $j'$ does not receive any item under $\phi$, i.e., $\phi(j') = \emptyset$, then $v_i(\phi(j')) = 0$, which implies that $\mms_i = \min_{j \in [n]} v_i(\phi(j)) = 0 \leq v_i(z^i_n)$. So, assume that $\phi(j) \neq \emptyset$ for every $j \in [n]$. Since there are $n$ agents, and since $\phi$ is $(1, 0)$-compact, $\phi$ allocates exactly $n$ items. In particular, $\phi$ allocates $z^i_r$ for some $r \geq n$, which implies that $\mms_i = \min_{j \in [n]} v_i(\phi(j)) \leq v_i(z^i_r) \leq v_i(z^i_n)$. We have thus shown that $\mms_i = v_i(z^i_n)$. 
\end{proof}
\end{toappendix}

As for envy-freeness, we impose the additional constraint that each agent must receive exactly one item. Then the corresponding problem is  polynomial time solvable, as shown by~\cite{DBLP:journals/mss/GanSV19}. Their algorithm is also a matching algorithm, which works as follows: We construct the graph where each agent is adjacent to her most valued items and look for a matching that saturates all agents or a Hall's violator. If we find a matching, then we have an envy-free allocation; else, we remove a Hall's violator and recurse. 

We thus have the following result.

\begin{theorem}
\xcompact{1}{0} is polynomial time solvable, where $[\mathtt{X}]$ is a placeholder for proportionality, envy-freeness or maximin fairness. 
\end{theorem}

\subsection{Diameter of a Connected \texorpdfstring{$\bm{(\alpha, \beta)}$}{(alpha, beta)}-Compact Graph}
\label{sec:diameter}
Observe that $(\alpha, \beta)$-compactness is not a subgraph-closed property. That is, if $G$ is an $(\alpha, \beta)$-compact graph and $H$ is a subgraph of $G$ then $H$ need not necessarily be $(\alpha, \beta)$-compact. But if $H$ is a connected component of $G$, then $H$ necessarily is $(\alpha, \beta)$-compact. Moreover, the diameter of $H$ is bounded by $\cO(\alpha \beta)$. We record this fact below. 

\begin{lemma}\label{lem:diam}
Let $G$ be an $(\alpha, \beta)$-compact graph. Then for every connected component $H$ of $G$, $H$ is $(\alpha, \beta)$-compact and $\diam(H) = \cO(\alpha \beta)$. 
\end{lemma}
\begin{proof}
Since $G$ is $(\alpha, \beta)$-compact, there exist vertices $z_1, z_2,\ldots, z_{\alpha} \in V(G)$ such that $V(G) = \bigcup_{j \in [\alpha]} B_G(z_j, \beta)$. Consider a connected component $H$ of $G$. Notice that for every $j \in [\alpha]$ such that $H$ contains the vertex $z_j$, we have $B_G(z_j, \beta) = B_H(z_j, \beta)$. Let $J = \set{j \in [\alpha] ~|~ z_j \in V(H)}$. Then $V(H) = \bigcup_{j \in J} B_H(z_j, \beta)$. As $\card{J} \leq \alpha$, $H$ is $(\alpha, \beta)$-compact. 

To see that $\diam(H) = \cO(\alpha \beta)$, observe first that for every $j \in J$ and $z, z' \in B_H(z_j, \beta)$, we have $\dist_H(z, z') \leq \dist_H(z, z_j) + \dist_H(z_j, z') \leq 2\beta$. Now, consider an auxiliary graph $H'$ on $\card{J}$ vertices, defined as follows. For each $j \in J$, $H'$ contains a vertex $y_j$. For $j, j' \in J$, $y_j y_{j'} \in E(H')$ if and only if either $B_H(z_j, \beta) \cap B_H(z_{j'}, \beta) \neq \emptyset$ or there exists $z \in B_H(z_j, \beta)$ and $z' \in B_H(z_{j'}, \beta)$ such that $zz' \in E(H)$. (Informally, $H'$ is obtained from $H$ by ``contracting'' each $B_H(z_j, \beta)$ into a single vertex). Notice that $H'$ has $\card{J} \leq \alpha$ vertices; and since $H$ is connected, $H'$ is connected as well. Now, consider $z, z' \in V(H)$. Let us observe that $H$ contains a $z$-$z'$ path of length $\cO(\alpha \beta)$. Let $j, j' \in J$ be such that $z \in B_H(z_j, \beta)$ and $z' \in B_H(z_{j'}, \beta)$. Let $Q$ be a path in $H'$ from $y_j$ to $y_{j'}$; the length of $Q$ is at most $\card{J} - 1 \leq \alpha - 1$. Observe that we can ``expand out'' each vertex of $Q$ to a path in $H$ of length at most $2\beta$; and thus obtain a $z$-$z'$ path of length $\cO(\alpha \beta)$. Hence $\dist_H(z, z') = \cO(\alpha \beta)$ for every $z, z' \in V(H)$, and the lemma follows.  
\end{proof}

\begin{observation}\label{obs:alpha1-diam}
Let us observe an immediate consequence of Lemma~\ref{lem:alpha1}. Consider a graph $G$ and an $(\alpha, \beta)$-compact allocation $\fn{\pi}{[n]}{2^{V(G)}}$. {\bf Then the subgraph of $\bm{G}$ induced by the set of allocated vertices, i.e., $\bm{G[\pi([n])]}$, is $\bm{(n \alpha, \beta)}$ compact.} To see this, notice that for each $i \in [n]$, there exist $z_{1}^{i}, z_2^i,\ldots, z_{\alpha}^i$ such that $\pi(i) = \bigcup_{j = 1}^{\alpha} B_{G[\pi(i)]}(z_j^i, \beta)$. 
Notice that $B_{G[\pi(i)]}(z_j^i, \beta) \subseteq B_{G[\pi([n])]}(z_j^i, \beta)$. 
And we have $\pi([n]) = \bigcup_{i = 1}^n \pi(i) = \bigcup_{i = 1}^n \bigcup_{j = 1}^{\alpha} B_{G[\pi(i)]}(z_j^i, \beta) \subseteq \bigcup_{i = 1}^n \bigcup_{j = 1}^{\alpha} B_{G[\pi([n])]}(z_j^i, \beta)$. That is, $G[\pi([n])]$ can be covered by at most $n \alpha$ balls, each of radius at most $\beta$, which implies that $G[\pi(n)]$ is $(n \alpha, \beta)$-compact. But then, {\bf Lemma~\ref{lem:diam} implies that $\bm{\diam(H) = \cO(n \alpha \beta)}$ for every connected component $\bm{H}$ of $\bm{G[\pi([n])]}$}. 
\end{observation}

\subsection{Annotated Allocations}
\label{sec:annotated} 
Consider a graph $G$. Suppose that $G$ is a $(\alpha, \beta)$-compact. Then there exist vertices $z_1, z_2,\ldots, z_{\alpha} \in V(G)$ such that $V(G) = \bigcup_{j \in [\alpha]} B_{G}(z_j, \beta)$. Let $G'$ be the graph obtained from $G$ by adding a new vertex $\hat z$ and making $\hat z$ adjacent to $z_1, z_2,\ldots z_{\alpha}$. We then have $V(G') = B_{G'}(\hat z, \beta + 1)$. That is, $G'$ is $(1, \beta + 1)$-compact. To exploit this idea in our algorithms, we introduce an annotated variant of \compact{\alpha}{\beta}. 
For a graph $G$, a vertex $\hat z \in V(G)$ and a non-negative integer $\beta$, we say that $G$ is a $(\hat z; \beta)$-annotated graph if $z \in B_G(\hat z, \beta)$ for every $z \in V(G)$, i.e., every vertex in $G$ is within a distance of at most $\beta$ from $\hat z$. Consider a graph $G$, a set of agents $N = [n]$ and an ordered sequence $(\hat z_i)_{i \in N}$ of $n$ distinct vertices $\hat z_i \in V(G)$. We say that an allocation $\fn{\pi}{N}{2^{V(G)}}$ is $((\hat z_i)_{i \in N}; \beta)$-annotated if for every $i \in [n]$, (i) $\hat z_i \in \pi(i)$ and (ii) $G[\pi(i)]$ is a $(\hat z_i; \beta)$-annotated graph. And we define the associated computational problem \anno\ as follows: Given $(G, N, \ca{V})$, $(\hat z_i)_{i \in N}$ and $\beta$ as input, does $(G, N, \ca{V})$ admit a $((\hat z_i)_{i \in N}; \beta)$-annotated allocation? We can augment the definition of the problem to include any of the fairness concepts as well. With respect to the fairness concept $[\mathtt{X}]$, we thus define \xanno: Given $(G, N, \ca{V})$, $(\hat z_i)_{i \in N}$ and $\beta$ as input, does $(G, N, \ca{V})$ admit an allocation that is $[\mathtt{X}]$-fair and $((\hat z_i)_{i \in N}; \beta)$-annotated?\footnote{To be precise, when it comes to annotated allocations, while the definitions of proportional and envy-free allocations remain unchanged, we must adapt the definition of maximin-fair annotated allocations appropriately. For $(\hat z_i)_{i \in N}$, we define $\Gamma(((\hat z_i)_{i \in N}; \beta) \mh \mathtt{anno})$ to be the set of all $((\hat z_i)_{i \in N}; \beta)$-annotated allocations. And an allocation $\pi$ is  $\Gamma(((\hat z_i)_{i \in N}; \beta) \mh \mathtt{anno})$-maximin fair if $v_i(\pi(i)) \geq  \Gamma(((\hat z_i)_{i \in N}; \beta) \mh \mathtt{anno}) \mh \mms_i(G, N, \ca{V})$ for every $i \in N$.} The following lemma says that to solve an instance of \xcompact{\alpha}{\beta}, it is enough to solve $m^{\cO(\alpha n)}$ many instances of \xanno. 

\begin{lemma}\label{lem:alpha1}
Consider an instance $(G, [n], \ca{V})$ of \xcompact{\alpha}{\beta}. Then there exist $m^{\cO(\alpha n)}$ many instances of \xanno\ such that $(G, [n], \ca{V})$ is a yes-instance of \xcompact{\alpha}{\beta} if and only if at least one of the $m^{\cO(\alpha n)}$ instances of \xanno\ is a yes-instances. Moreover, these $m^{\cO(\alpha n)}$ instances of \xanno\ can be constructed in time $m^{\cO(\alpha n)}$. 
\end{lemma}
\begin{proof}
Let $T$ be the set of all tuples $(C_1, C_2,\ldots, C_n)$ such that for every $i \in [n]$, $C_i \subseteq V(G)$, $\card{C_i} \leq \alpha$,  and $C_i \cap C_j = \emptyset$ for distinct $i, j \in [n]$. Corresponding to each tuple $\tau = (C_1, C_2,\ldots, C_n) \in T$, we construct an instance $(G^{\tau}, [n], \ca{V}^{\tau}; (\hat z_i)_{i \in [n]}, \beta^{\tau})$ of \xanno\ as follows. To construct $G^{\tau}$, we add $n$ new vertices $\hat z_1, \hat z_2,\ldots, \hat z_n$ to $G$; and for each $i \in [n]$, we make $\hat z_i$ adjacent to all the vertices in $C_i$. We set $v^{\tau}_i(\hat z_j) = 0$ and $v^{\tau}_i(z) = v_i(z)$ if $z \in V(G^{\tau}) \setminus \set{\hat z_j ~|~ j \in [n]}$. And we take $\beta^{\tau} = \beta + 1$. 

First, notice that $\card{T} = m^{\cO(\alpha n)}$, as the number of subsets of $V(G)$ of size at most $\alpha$ is $\sum_{j = 0}^{\alpha}\binom{m}{j} = m^{\cO(\alpha)}$; and hence there are $m^{\cO(\alpha n)}$ tuples $(C_1, C_2,\ldots, C_n)$ with $\card{C_i} \leq \alpha$ for every $i \in [n]$. And notice that for each tuple $\tau \in T$, we can construct the corresponding instance $(G^{\tau}, [n], \ca{V}^{\tau}; (\hat z_i)_{i \in [n]}, \beta^{\tau})$ of \anno\ in time $m^{\cO(1)}$. We can also verify that $(G, [n], \ca{V})$ is a yes-instance of \xcompact{\alpha}{\beta} if and only if $(G^{\tau}, [n], \ca{V}^{\tau}; (\hat z_i)_{i \in [n]}, \beta^{\tau})$ is a yes-instance of \xanno\ for some $\tau \in T$. 

Consider an $(\alpha, \beta)$-compact allocation $\fn{\pi}{[n]}{2^{V(G)}}$. Then for each $i \in [n]$, there exist $z^i_1, z^i_2,\ldots, z^i_{\alpha}$ such that $\pi(i) = \bigcup_{j \in [\alpha]} B_{G[\pi(i)]}(z^i_j, \beta)$. Consider the tuple $\tau = (C_i)_{i \in [n]}$, where $C_i = \set{z^i_1, z^i_2,\ldots, z^i_{\alpha}}$. Then $\tau \in T$. Consider the allocation $\fn{\pi^{\tau}}{[n]}{2^{V(G^{\tau})}}$, where $\pi^{\tau}(i) = \pi(i) \cup \set{\hat z_i}$. Notice that $\pi^{\tau}$ is a $((\hat z_i)_{i \in [n]}; \beta + 1)$-annotated allocation. This follows from the fact that for $i \in [n]$, and $z \in \pi(i)$, there exists $j \in [\alpha]$ such that $\dist_{G[\pi(i)]}(z^i_j, z) \leq \beta$, which implies that $\dist_{G^{\tau}[\pi^{\tau}(i)]}(\hat z_i, z) \leq \beta + 1$. 
Notice also that for every $i, j \in [n]$, as $v^{\tau}_i(\hat z_j) = 0$, we have $v_i(\pi(j)) = v^{\tau}_i(\pi^{\tau}(j))$; and hence $\pi$ is $[\mathtt{X}]$-fair if and only if $\pi^{\tau}$ is $[\mathtt{X}]$-fair. 

Conversely, suppose there exists $\tau = (C_i)_{i \in [n]} \in T$ such that $\fn{\phi}{[n]}{2^{V(G^{\tau})}}$ is a $((\hat z_i)_{i \in [n]}; \beta + 1)$-annotated allocation. Then the allocation $\fn{\phi'}{[n]}{2^{V(G)}}$ defined by $\phi'(i) = \phi(i) \setminus \set{\hat z_i}$ for every $i \in [n]$ is an $(\alpha, \beta)$-compact allocation. To see this, fix $i \in [n]$ and consider the set $D_i = \phi(i) \cap N_{G^{\tau}}(\hat z_i)$. Notice that $D_i \subseteq C_i$, and hence $\card{D_i} \leq \alpha$. 
Now consider $z \in \phi(i) \setminus \set{\hat z_i}$,  there exists a $z$-$\hat z_i$ path of length at most $\beta + 1$ in $G^{\tau}[\phi(i)]$. Let $P$ be such a path, and let $y_i$ be the unique neighbor of $\hat z_i$ in $P$. Then $y_i \in D_i$. Let $P'_i$ denote the $z$-$y_i$ subpath of $P$. Notice that $P'_i$ is a path in $G[\phi'(i)]$ and that $P'_i$ has length at most $\beta$. Thus $z \in B_{G[\phi'(i)]}(y_i, \beta)$. Since $z$ was an arbitrary element of $\phi'(i) = \phi(i) \setminus \set{\hat z_i}$, we can conclude that $\phi(i) = \bigcup_{y \in D_i} B_{G[\phi(i)]}(y, \beta)$, which implies that $G[\phi(i)]$ is $(\alpha, \beta)$-compact. Also, as $v^{\tau}_i(\phi(j)) = v_i(\phi'(j))$, $\phi$ is $[\mathtt{X}]$-fair if and only if $\phi'$ is $[\mathtt{X}]$-fair. 
\end{proof}

\begin{remark}
\label{rem:rrule}
Consider a tuple $\tau = (C_1, C_2,\ldots, C_n)$ as defined in the proof of Lemma~\ref{lem:alpha1}. 
And consider our construction of the instance $(G^{\tau}, [n], \ca{V}^{\tau}; (\hat z_i)_{i \in [n]}, \beta^{\tau})$ of \xanno. Recall that $\beta^{\tau} = \beta + 1$. Consider $z \in V(G)$ such that $z \notin \bigcup_{i \in [n]}\bigcup_{z' \in C_i} B_G(z', \beta)$. Notice that $z \notin \bigcup_{i \in [n]} B_{G^{\tau}}(\hat z_i, \beta + 1)$. Hence $z$ remains un-allocated under \emph{every} $((\hat z_i)_{i \in [n]}; \beta + 1)$-annotated allocation. We can therefore safely delete $z$ from $V(G)$. Therefore, we can assume without loss of generality that $V(G) = \bigcup_{i \in [n]}\bigcup_{z' \in C_i} B_G(z', \beta)$, which implies that $G$ is $(\alpha n, \beta)$-compact, as $\card{C_i} \leq \alpha$ for every $i \in [n]$. Then, by Observation~\ref{obs:alpha1-diam}, we have $\diam(H) = \cO(n \alpha \beta)$ for every connected component $H$ of $G$. 
\end{remark}

\subsection{The Complexity of Recognising (Strongly) Compact Graphs}
\label{sec:recognition}
Consider the recognition problem for (strongly) $(\alpha, \beta)$-compact graphs. That is, the problem of testing whether a given a graph $G$ is (strongly) $(\alpha, \beta)$-compact (for fixed $\alpha$ and $\beta$). We can recognise $(\alpha, \beta)$-compact graphs in polynomial time. Given an $m$-vertex graph $G$, for every choice of $\alpha$ vertices $z_1, z_2,\ldots, z_{\alpha}$, we can check if $V(G) = \bigcup_{j \in [\alpha]} B_G(z_j, \beta)$, and we can do this in time $m^{\cO(1)}$. Since the number of choices for $z_1, z_2,\ldots, z_{\alpha}$ is $m^{\cO(\alpha)}$, we conclude that we can check if $G$ is $(\alpha, \beta)$-compact in time $m^{\cO(\alpha)}$. As $\alpha$ is a constant, we can indeed check if $G$ is $(\alpha, \beta)$-compact in polynomial time. For future reference, we note this fact. 
\begin{remark}
\label{rem:recognition}
For fixed $\alpha$ and $\beta$, we can recognise $(\alpha, \beta)$-compact graphs in polynomial time. 
\end{remark}

But the case of strong compactness is quite different. For $\alpha \leq 2$, the problem is polynomial time solvable and for $\alpha \geq 3$, the problem is \NPH. If $\alpha = 1$, then the problem is the same as testing whether $\diam(G) \leq \beta$, which can be done in polynomial time. If $\alpha = 2$, then the problem can be reduced to {\sc 2-SAT}, and hence it can still be solved in polynomial time. For $\alpha \geq 3$, the problem is \NPH\ even when $\beta = 1$. Notice that $G$ is strongly $(\alpha, 1)$-compact if and only if there exist $V_1, V_2\ldots, V_{\alpha} \subseteq V(G)$ such that $G[V_i]$ is a clique for every $i \in [\alpha]$. So checking whether $G$ is $(\alpha, 1)$-compact is the same as checking whether $V(G)$ can be covered by $\alpha$ cliques. This problem is called the $\alpha$-{\sc Clique Cover} problem, which is known to be \NPC~\citep{karp1972reducibility}. In fact, notice that $\alpha$-{\sc Clique Cover} is equivalent to the well known $\alpha$-{\sc Colouring} problem. For a graph $G$, $V(G)$ can be covered using $\alpha$ cliques if and only if $V(\overline G)$ can be covered using $\alpha$ independent sets, where $\overline G$ is the complement of $G$; and $V(\overline G)$ can be covered using $\alpha$ independent sets if and only if the vertices of $\overline G$ can be properly coloured using $\alpha$ colours.

\section{Strong NP-Hardness of (Strongly) Compact Fair Division}\label{sec:nph}

In this section, we establish the strong \NPH ness of (strongly) compact fair division. 
We consider three fairness and efficiency combinations (i) proportionality, (ii) envy-freeness plus  completeness and (iii) envy-freeness plus Pareto-optimality. For each of these three choices,
we show that the corresponding problem of checking whether there exists a fair (strongly) $(\alpha, \beta)$-compact allocation is strongly \NP-hard for every fixed $\alpha \geq 3$ and $\beta \geq 0$.\footnote{Notice that even though Theorem~\ref{thm:lipton} already establishes the \NPH ness of these problems for $\alpha, \beta \geq 1$, Theorem~\ref{thm:lipton} only implies \emph{weak} \NPH ness, as we reduced from {\sc Partition}. Besides Theorem~\ref{thm:lipton} only works for $\beta \geq 1$.}  
These results hold even the graph $G$ is edgeless and thus has treewidth $0$. And notice that if $G$ is an edgeless graph, then for any $S \subseteq V(G)$, $G[S]$ is (strongly) $(\alpha, \beta)$-compact if and only if $\card{S} \leq \alpha$. 
For proportionality and envy-freeness + completeness, we also show that the strongly compact variants of the corresponding problems are \NPH\ for every fixed $\alpha, \beta \geq 1$. All the hardness results hold even when the valuations are additive and polynomially bounded. We prove these results one by one. 

\subsection{Strong NP-Hardness of \texorpdfstring{$\bm{(\alpha, \beta)}$}{(alpha, beta)}-Compact FD for \texorpdfstring{$\bm{\alpha \geq 3}$}{alpha >= 3}}

\begin{theorem}\label{thm:prop30}
For every $\alpha \geq 3$ and $\beta \geq 0$, \propn{\alpha}{\beta}\ is strongly \NPC. 
\end{theorem}

\begin{proof}
Fix $\alpha \geq 3$ and $\beta \geq 0$. Membership of \propn{\alpha}{\beta} in \NP\ is trivial: Given an allocation, we can verify in polynomial time whether it is proportional, and by Remark~\ref{rem:recognition}, we can also verify whether the allocation is $(\alpha, \beta)$-compact. We now show \NPH ness by a reduction from {\sc Exact Cover by ${\alpha}$-Sets (X$\alpha$C)}, which is known to be \NPC for every $\alpha \geq 3$~\citep{karp1972reducibility}.\footnote{Exact Cover by ${3}$-Sets (X$3$C) was, in fact, one of Karp's original 21 \NPC\ problems~\citep{karp1972reducibility}. The original paper of Karp only shows the \NP\ hardness of X3C. But it is straightforward to extend that result to show \NP-hardness of X$\alpha$C for every $\alpha > 3$ as well. For the sake of completeness, we sketch a proof of this fact in the appendix.} In X$\alpha$C, the input consists of a set $X$ such that $\card{X} = \alpha s$ and a family $\ca{F} \subseteq 2^X$ of subsets of $X$ such that $\card{F} = \alpha$ for every $F \in \ca{F}$; and the question is to decide if $(X, \ca{F})$ has an exact cover, i.e., a sub-family $\ca{F}' \subseteq \ca{F}$ such that $\card{\ca{F}'} = s$ and $\bigcup_{F \in \ca{F}'} F = X$. Consider an instance $(X, \ca{F})$ of X$\alpha$C. 
Let $ \card{F} = r$ and $\ca{F} = \set{F_1, F_2,\ldots, F_r}$.  Notice first that if $r = \card{\ca{F}} < s$, then $(X, \ca{F})$ is a no-instance of X$\alpha$C, and in this case, we return a trivial no-instance of \propn{\alpha}{\beta}. If $r = s$, then $(X, \ca{F})$ is a yes-instance if and only if $\bigcup_{F \in \ca{F}} = X$, and we can verify this in polynomial time. In this case, we return a trivial yes-instance or no-instance of \propn{\alpha}{\beta}\ accordingly. So from now on, we assume without loss of generality that $r > s$. Given $(X, \ca{F})$, we now construct an instance $(G, N, \ca{V})$ of \propn{\alpha}{\beta} as follows. Corresponding to each $F_i \in \ca{F}$, we introduce an agent $i$; in addition we introduce a ``dummy'' agent $r + 1$. Thus $N = [r + 1]$. Corresponding to each $x \in X$, we introduce an item $w_x$; in addition, we introduce $r - s$ ``auxiliary'' items $y_1, y_2,\ldots, y_{r - s}$ and a ``special'' item $y_*$. And we define the graph $G$ to be the edge-less graph with vertex set $V(G) = \set{y_*} \cup \set{y_j ~|~ j \in [r -s]} \cup \set{w_x ~|~ x \in X}$. Finally, we define the valuation functions as follows. Consider $i \in [r]$. For each $x \in X$, we set
\[
v_i(w_x) = \begin{cases} 1/(\alpha r + \alpha), & \text{ if } x \in F_i \\ 0, & \text{ if } x \notin F_i; \end{cases}
\] 
$v_i(y_j) = 1/(r + 1)$ for every $j \in [r - s]$; and $v_i(y_*) = s/(r + 1)$. Finally, $v_{r + 1}(w_x) = 0$ for every $x \in X$; $v_{r + 1}(y_j) = 0$ for every $j \in [r - s]$; and $v_{r + 1}(y_*) = 1$. 
Notice that we have $v_i(V(G)) = \sum_{z \in V(G)} v_i(z) = 1$ for every $i \in N$. And notice also that since $G$ is edge-less, for every non-empty subset $S \subseteq V(G)$, the subgraph $G[S]$ is $(\alpha, \beta)$-compact if and only if $\card{S} \leq \alpha$. 

Assume that $(X, \ca{F})$ is a yes-instance, and let $\ca{F}' \subseteq \ca{F}$ be an exact cover for $(X, \ca{F})$. Then $\card{\ca{F}} = s$ and $\card{\ca{F} \setminus \ca{F}'} = r - s$. Assume without loss of generality that $\ca{F} \setminus \ca{F}' = \set{F_1, F_2,\ldots, F_{r - s}}$.  
We define an allocation $\fn{\pi}{N}{2^{V(G)}}$ as follows. Consider $i \in [r]$. If $F_i \in \ca{F}'$, then we set $\pi(i) = \set{w_x ~|~ x \in F_i}$ so that $v_i(\pi(i)) = \sum_{x \in F_i} v_i(w_x) = \card{F_i}(1/(\alpha r + \alpha)) = 1/(r + 1)$; and if $F_i \in \ca{F} \setminus \ca{F}'$, then $i \in [r - s]$, and we set $\pi(i) = \set{y_{i}}$ so that $v_i(\pi(i)) = 1/(r + 1)$. And we set $\pi(r + 1) = \set{y_*}$ so that $v_{r + 1}(\pi(r + 1)) = 1$. Notice that $\pi$ is a complete allocation; and as $v_i(\pi(i)) \geq 1/(r +1)$ for every $i \in N$, $\pi$ is proportional. As $\card{\pi(i)} \leq \alpha$ for every $i \in N$, $\pi$ is $(\alpha, \beta)$-compact as well. 

Conversely, assume that $(G, N, \ca{V})$ admits an $(\alpha, \beta)$-compact, proportional allocation, say $\fn{\phi}{N}{2^{V(G)}}$. First, as $\phi$ is proportional, we have $v_i(\phi(i)) \geq (1/\card{N}) \cdot v_i(V(G)) = 1/(r + 1)$ for every $i \in N$. In particular, $v_{r + 1}(\phi(r + 1)) \geq 1/(r + 1)$, which implies that $y_* \in \phi(r + 1)$. Thus $y_* \notin \phi(i)$ for every $i \in [r]$. Let $R = \{i \in [r] ~|~ \phi(i) \subseteq \{w_x ~|~ x \in X \} \}$. Observe that $\card{R} \geq s$. To see this, notice that there exist at most $r -s$ agents $i \in [r]$ such that $y_j \in \phi(i)$ for some $j \in [r - s]$. Therefore there exist at least $s$ agents $i \in [r]$ such that $\phi(i) \subseteq \set{w_x ~|~ x \in X}$. Now, for every $i \in R$, we must have $w_x \in \phi(i)$ for every $x \in F_i$, for otherwise $v_i(\phi(i)) < 1/(r + 1)$. As $\phi$ is $(\alpha, \beta)$-compact and $G$ is edge-less, we have $\card{\phi(i)} \leq \alpha$, and hence $\phi(i) = \set{w_x ~|~ x \in F_i}$ for every $i \in R$. Also, since $\phi$ is an allocation, for distinct $i, i' \in R$, we have $\phi(i) \cap \phi(i') = \emptyset$ and hence $F_i \cap F_{i'} = \emptyset$. Then $\card{R} \geq s$ implies that $\card{\bigcup_{i \in R} F_i} = \alpha \card{R} \geq \alpha s$. Since $\card{X} = \alpha s$, we can conclude that $\card{R} = s$ and $\bigcup_{i \in R} F_i = X$, and hence $\set{F_i ~|~ i \in R}$ is an exact cover for $(X, \ca{F})$. 
\end{proof}

By slightly modifying the reduction in the proof of Theorem~\ref{thm:prop30}, we can show that \cefn{\alpha}{\beta} is \NPC\ as well. We only sketch the proof below as the arguments are identical to the ones used to prove Theorem~\ref{thm:prop30}. 

\begin{theorem}\label{thm:cef30}
For every $\alpha \geq 3$ and $\beta \geq 0$, \cefn{\alpha}{\beta} is strongly \NPC.
\end{theorem}

\begin{proof}[Proof Sketch]
Membership in \NP\ is again trivial. To show hardness, we reduce from (X$\alpha$C). Given an instance $(X, \ca{F})$, we proceed exactly as in the proof of Theorem~\ref{thm:prop30} except for the following three differences. (1) We no longer have the agent $r + 1$. That is, $N = [r]$ and $i \in N$ corresponds to the set $F_i \in \ca{F}$.  (2) We no longer have the item $y_*$. Thus the graph $G$ is the edge-less graph with vertex set $V(G) = \set{y_j ~|~ j \in [r -s]} \cup \set{w_x ~|~ x \in X}$. (3) We modify the valuation functions to ensure that (a) $v_i(\{w_ x ~|~ x \in X \}) = v_i(\{y_j \})$ for every $i \in [r]$ and $j \in [r - s]$ and (b) $\sum_{z \in V(G)} v_i(z) = 1$ for every $i \in N$. Specifically, we define the valuation functions as follows. Consider $i \in [r]$.  For $x \in X$, we have  
\[
v_i(w_x) = \begin{cases} 1/(\alpha r - \alpha s + \alpha), & \text{ if } x \in F_i \\ 0, & \text{ if } x \notin F_i; \end{cases}
\] 
and $v_i(y_j) = 1/(r - s + 1)$ for every $j \in [r - s]$. 

Now, if $(X, \ca{F})$ is a yes-instance, then we have an exact cover $\ca{F}' \subseteq \ca{F}$ for $(X, \ca{F})$; and we define an allocation $\fn{\pi}{N}{2^{V(G)}}$ exactly as in the proof of Theorem~\ref{thm:prop30}. We thus have $\card{\pi(i)} \leq \alpha$ for every $i \in [r]$, and hence $\pi$ is $(\alpha, \beta)$-compact.  
And for every pair of distinct $i, i' \in [r]$, we can verify that $v_i(\pi(i)) = 1/(r - s + 1)$ and $v_i(\pi(i')) \leq 1/(r - s + 1)$. Thus every agent is envy-free under $\pi$. 

Conversely, consider a complete, envy-free, $(\alpha, \beta)$-compact allocation $\fn{\phi}{N}{2^{V(G)}}$. Since $\phi$ is complete, every item is allocated; in particular, $y_j$ is allocated for every $j \in [r - s]$. Fix $ j \in [r - s]$ and let $i' \in [r]$ be such that $y_j \in \phi(i')$. Then for every agent $i \in [r]$, we have $v_{i}(\phi(i')) \geq 1/(r - s + 1)$. But then, since $\phi$ is envy-free, we must have $v_i(\phi(i)) \geq 1/(r - s + 1)$ as well. Now we consider the set $R = \{i \in [r] ~|~ \phi(i) \subseteq \{w_x ~|~ x \in X \} \}$. In particular, we must have $v_i(\phi(i)) \geq 1/(r - s + 1)$ for every $i \in R$. And using exactly the same arguments that we used in Theorem~\ref{thm:prop30}, we can prove that $\card{R} = s$ and $\set{F_i ~|~ i \in R}$ is an exact cover for $(X, \ca{F})$. 
\end{proof}

In Theorem~\ref{thm:cef30}, we can replace completeness with Pareto-optimality, and the reduction would still hold. We thus have the following result. 
\begin{theorem}\label{thm:poef30}
For every $\alpha \geq 3$ and $\beta \geq 0$, \poefn{\alpha}{\beta} is strongly \NPH. 
\end{theorem}

\begin{proof}[Proof Sketch]
We use the same reduction as in the proof of Theorem~\ref{thm:cef30}, and construct the instance $(G, N, \ca{V})$ of \poefn{\alpha}{\beta} from a given instance $(X, \ca{F})$ of X$\alpha$C. 

Notice that the allocation $\pi$ that we defined in the forward direction of the proof of Theorem~\ref{thm:cef30} is indeed Pareto-optimal. To prove this, it is sufficient to prove that $\pi$ achieves the maximum utilitarian social welfare, i.e., $\sum_{i \in N} v_i(\pi(i)) = \max_{\pi'} \sum_{i \in N}{v_i(\pi'(i))}$, where the maximum is over all allocations $\fn{\pi'}{N}{2^{V(G)}}$. Notice that for each $z \in V(G)$, if $\pi$ allocates $z$ to agent $i$, then for every agent $i'$, we have $v_i(z) \geq v_{i'}(z)$. This fact, along with the completeness of $\pi$ and the additivity of $v_i$s, implies that $\pi$ achieves the maximum utilitarian social welfare. To prove this formally, we can first verify that $\sum_{i \in N} v_i(\pi(i)) = (\alpha s)/(\alpha r - \alpha s + 1) + (r - s )/(r - s + 1) = r/(r - s + 1)$. 

Now, for any allocation $\fn{\pi'}{N}{2^{V(G)}}$, we have, 
\begin{align*}
\sum_{i \in N} v_i(\pi'(i)) &\leq \sum_{z \in V(G)} \max_{i \in N} v_i(z)  \\
&=\sum_{x \in X} \max_{i \in N} v_i(w_x) + \sum_{j \in [r - s]} \max_{i \in N} v_i(y_j) \\
&= \sum_{x \in X} \frac{1}{\alpha r - \alpha s + \alpha} +  \sum_{j \in [r - s]} \frac{1}{r - s + 1} \\
&= \frac{\alpha s}{\alpha r - \alpha s + \alpha}  + \frac{r - s}{r - s + 1}  \\
&= \frac{r}{r - s + 1} \\
&= \sum_{i \in N} v_i(\pi(i)),  
\end{align*} 
which shows that $\pi$ is Pareto-optimal. 

For the backward direction, consider a Pareto-optimal, envy-free, $(\alpha, \beta)$-compact allocation $\fn{\phi}{N}{2^{V(G)}}$. (Notice that in the proof of Theorem~\ref{thm:cef30}, we used completeness only to ensure that $y_j$ is allocated (for at least one $j \in [r -s]$) under the allocation $\phi$.  And for this, Pareto-optimality is sufficient.) Every Pareto-optimal allocation must necessarily be complete, because for every $z \in V(G)$, we have $v_i(z) > 0$ for some $i \in [r]$.\footnote{We can, in fact, argue that every Pareto-optimal, $(\alpha, \beta)$-compact allocation $\phi$ must necessarily allocate $y_j$ for some $j \in [r]$. Suppose $\phi$ does not allocate $y_j$ for every $j \in [r - s]$. Then $\phi([r]) \subseteq \{w_x ~|~ x \in X \} \}$, which implies that $\card{\phi([r])} \leq \alpha s$. Recall that we are under the assumption that $r > s$. Since $\card{\phi(i)} \leq \alpha$ for every $i \in [r]$, there exists an agent $p \in [r]$ such that $\card{\phi(p)} < \alpha$. We can thus allocate $y_1$ to $p$, (while still ensuring $(\alpha, \beta)$-compactness), which will increase $p$'s utility by $v_p(y_1) = 1/(r - s + 1)$. And this contradicts the assumption that $\phi$ is Pareto-optimal.}
The rest of the arguments are identical to that in the proof of Theorem~\ref{thm:cef30}. 
\end{proof}

Observe that in each of the reductions in the proofs of Theorems~\ref{thm:prop30}, \ref{thm:cef30} and \ref{thm:poef30}, the graph $G$ that we constructed is edge-less. Thus for $S \subseteq V(G)$, the subgraph $G[S]$ is $(\alpha, \beta)$-compact if and only if $G[S]$ is \emph{strongly} $(\alpha, \beta)$-compact. Therefore, even if we replace compactness with strong compactness in Theorems ~\ref{thm:prop30}, \ref{thm:cef30} and \ref{thm:poef30}, the hardness results would still hold. We thus have the following result.

\begin{theorem}
\label{thm:NPH-strong-compact}
For every $\alpha \geq 3$ and $\beta \geq 0$, \spropn{\alpha}{\beta}, \scefn{\alpha}{\beta} and \spoefn{\alpha}{\beta} are all strongly \NPH. 
\end{theorem}

Theorems~\ref{thm:prop30}-\ref{thm:NPH-strong-compact}, however, do not cover the case when $\alpha \leq 2$. We provide a partial answer for this case in Section~\ref{sec:strong-NP-strong-compact}. In particular, for the strongly compact variants, we show that the problems of proportional as well as envy-free + complete allocations are still \NPH for every $\alpha \geq 1$ and $\beta \geq 1$. 

\begin{toappendix}
\begin{lemma}
For every $\alpha \geq 3$, {\sc Exact Cover by $\alpha$-Sets (X$\alpha$C)} is \NPC. 
\end{lemma}
\begin{proof}
Membership in \NP\ is trivial. When $\alpha = 3$, the problem is known to be \NPC~\citep{karp1972reducibility}. Fix $\alpha > 3$. 
To prove \NP-hardness, we reduce from X3C. Consider an instance $(X, \ca{F})$ of X3C, where $X = \set{x_1, x_2,\ldots, x_{3s}}$ and $\ca{F} = \set{F_1, F_2,\ldots, F_r}$. Recall that $\card{F_i} = 3$ for every $i \in [r]$. We construct an instance $(U, \ca{A})$ of X$\alpha$C by (i) adding a carefully chosen number of new elements to $X$, (b) adding $\alpha - 3$ of the newly introduced elements to each $F_i \in \ca{F}$ and (c) adding all possible $\alpha$-sized subsets of the newly introduced elements to $\ca{F}$. To construct $U$, we add $3sr\alpha(\alpha - 3) - 3s$ new elements to $X$. Formally, let $X' = \set{y_1, y_2,\ldots, y_{3sr\alpha(\alpha - 3)}}$ such that $X' \cap X = \emptyset$; and we take $U = X \cup X'$. Thus $\card{U} = 3s + 3sr\alpha(\alpha - 3) - 3s = \alpha(\alpha - 3)(3sr)$. Let $\ca{F}'$ be the family of sets obtained from $\ca{F}$ as follows. For each $i \in [r]$, $Y_i = \{y_{i(\alpha - 3)}, y_{i(\alpha - 3) - 1}, y_{i(\alpha - 3) - (\alpha - 4)}\}$. That is, $Y_1 = \set{y_1, y_2,\ldots, y_{\alpha - 3}}$, $Y_2 = \set{y_{(\alpha - 3) + 1}, y_{(\alpha - 3) + 2},\ldots, y_{2(\alpha - 3)}}$ and so on. Now, let $F'_i = F_i \cup Y_i$ for every $i \in [r]$. Notice that $F'_i \subseteq U$ and $\card{F'_i} = \alpha$ for every $i \in [r]$ and $F'_i \cap F'_j = F_i \cap F_j$ for every distinct $i, j \in [r]$. Also, let $\ca{F}'$ be the family of all $\alpha$-sized subsets of $X'$, i.e., $\ca{F}' = \set{Y \subseteq X' ~|~ \card{Y} = \alpha}$. Notice that $\card{\ca{F}'} = \binom{\card{X'}}{\alpha} = \binom{3sr\alpha(\alpha - 3) - 3s}{\alpha} = (\alpha sr)^{\cO(\alpha)}$. And we take $\ca{A} = \set{F'_i ~|~ i \in [r]} \cup \ca{F}'$. Notice that $(U, \ca{A})$ can be constructed in polynomial time. We claim that $(X, \ca{F})$ is a yes-instance of X$3$C if and only if $(U, \ca{A})$ is a yes-instance of X$\alpha$C. 

Let $\hat{\ca{F}} \subseteq \ca{F}$ be an exact cover for $(X, \ca{F})$.  
Let $\hat{\ca{F}_1} = \{F'_i ~|~ F_i \in \hat{\ca{F}}\}$. Notice that $\card{\hat{\ca{F}_1}} = \card{\hat{\ca{F}}} = s$ and that $F'_i \cap F'_j = F_i \cap F_j = \emptyset $ for every distinct $F_i, F_j \in \hat{\ca{F}}$.  
Now, let $Y' = \bigcup_{F_i \in \hat{\ca{F}}} Y_i$. Then $\card{Y'} = \card{\hat{\ca{F}}} \cdot (\alpha - 3) = s(\alpha - 3)$. And consider the set $X' \setminus Y'$. Notice that $\card{X' \setminus Y'} = 3sr\alpha(\alpha - 3) - 3s - s(\alpha - 3) = s(3r\alpha(\alpha - 3) - 3 -\alpha + 3) = \alpha s(3r(\alpha - 3) - 1)$. Fix a partition $\hat{\ca{F}_2}$ of $X' \setminus Y'$ such that every set in $\hat{\ca{F}_2}$ has size exactly $\alpha$. Such a partition $\hat{\ca{F}_2}$ has exists as $\card{X' \setminus Y'}$ is a multiple of $\alpha$. Notice that $\card{\hat{\ca{F}_2}} = s(3r(\alpha - 3) - 1)$ and that that the sets in $\hat{\ca{F}_2}$ are pairwise-disjoint. We now define $\hat{A} = \hat{\ca{F}_1} \cup \hat{\ca{F}_2}$. And we have $\card{\hat{\ca{A}}} = \card{\hat{\ca{F}_1}} + \card{\hat{\ca{F}_2}} = s + s(3r(\alpha - 3) - 1) = (\alpha - 3)(3sr) = (1/{\alpha})\card{U}$, and $\bigcup_{F \in \hat{\ca{A}}} F = \left( \bigcup_{F \in \hat{\ca{F}_1}} F \right) \cup \left( \bigcup_{F \in \hat{\ca{F}_2}} F \right) = (X \cup Y') \cup (X' \setminus Y') = U$. Thus $\hat{\ca{A}}$ is an exact cover for $(U, \ca{A})$.  

Conversely, let $\ca{A}' \subseteq \ca{A}$ be an exact cover for $(U, \ca{A})$. Notice that for every set $F \in \ca{A}'$, either $\card{F \cap X} = 0$ or $\card{F \cap X} = 3$; and if $\card{F \cap X} = 3$. Since $\ca{A}'$ is an exact cover and since $\card{X} = 3s$, there exist exactly $s$ sets $F \in \ca{A}'$ such that $\card{F \cap X} = 3$. Also, notice that if $\card{F \cap X} = 3$, then $F = F'_i = F_i \cup Y_i$ for some $i \in [r]$. Let $\ca{F}_{\ca{A}'} = \set{F_i \in \ca{F} ~|~ F'_i \in \ca{A}'}$. Then $\ca{F}_{\ca{A}'}$ is an exact cover for $(X, \ca{F})$. 
\end{proof}
\end{toappendix}

\subsection{Strong NP-Hardness of Strongly \texorpdfstring{$\bm{(\alpha, \beta)}$}{(alpha, beta)}-Compact FD for \texorpdfstring{$\bm{\alpha, \beta \geq 1}$}{alpha, beta >= 1}}
\label{sec:strong-NP-strong-compact}

\begin{theorem}\label{thm:prop11}
For every fixed $\alpha, \beta \geq 1$, \spropn{\alpha}{\beta} is strongly \NPH.
\end{theorem}

\begin{proof}
Fix $\alpha, \beta \geq 1$.  
We show \NPH ness by a reduction from the {\sc $\beta$-Club} problem, in which the input consists of a graph $H$ and a non-negative integer $k$, and the question is to decide if $H$ has a $\beta$-club of size exactly $k$. A $\beta$-club in $H$ is a set of vertices $S \subseteq V(H)$ such that $\diam(H[S]) \leq \beta$. So the {\sc $\beta$-Club} problem asks if $H$ contains a set of exactly $k$ vertices that induces a subgraph of diameter at most $k$. This problem is \NPC~\citep{DBLP:journals/eor/BourjollyLP02}.\footnote{We must note that there are two variants of the {\sc $\beta$-Club} problem: the ``exact'' and the ``at least'' variants. In the former, the task is to check if $H$ contains a $\beta$-club of size exactly $k$, whereas in the latter variant, the task is to check if $H$ contains a $\beta$-club of size at least $k$. Notice that these two problems are not equivalent. A graph $H$ may contain a $\beta$-club of size $k + 1$, yet no $\beta$-club of size exactly $k$. For example, take $H$ to be a cycle on $2\beta + 1$ vertices and  $k = 2\beta$. \cite{DBLP:journals/eor/BourjollyLP02} shows that the ``at least'' variant is \NPH. But their reduction, in fact, proves the \NPH ness of the ``exact'' variant as well.}

Consider an instance $(H, k)$ of {\sc $\beta$-Club}. Let $V(H) = \set{z_1, z_2, \ldots, z_s}$. We assume without loss of generality that $1 < k < s$, and construct an instance $(G, N, \ca{V})$ of \spropn{\alpha}{\beta} as follows. First, we take $G$ to be the disjoint union  of $H$ and $(\alpha - 1) + (s + k - 1)$ isolated vertices $x_1, x_2,\ldots, x_{\alpha - 1}$ and $y_1, y_2,\ldots, y_{ s +  k - 1}$. (Just to be clear, if $\alpha = 1$, then we do not have the vertices $x_1, x_2,\ldots, x_{\alpha - 1}$.) That is, $V(G) = V(H) \cup \set{x_j ~|~ j \in [\alpha - 1]} \cup \set{y_j ~|~ j \in [s + k -1]}$, and $E(G) = E(H)$. 
And we introduce $2s$ agents: a ``special'' agent $a^*$, $s - k$  ``regular'' agents $a^{1}, a^{2},\ldots, a^{s - k}$ and $s + k - 1$ ``dummy'' agents $b^1, b^2, \ldots, b^{s + k - 1}$. Thus $N = \set{a^*} \cup \ttsup{N}{reg} \cup \ttsup{N}{dum}$, where $\ttsup{N}{reg} = \set{a^{j} ~|~ j \in [s - k]}$ and $\ttsup{N}{dum} = \set{b^j ~|~ j \in [s + k - 1]}$. For every agent $c \in N$, we now define the valuations $\fn{v_c}{V(G)}{\mathbb{Q}_{\geq 0}}$ as follows. For every $i \in [s]$, $j \in [s + k - 1]$ and $j' \in [\alpha - 1]$, we have
\[
\begin{matrix}
{v_c(z_{i}) = \begin{cases} \frac{1}{2\alpha ks},  \text{ if } c = a^* \\ \frac{1}{2s}, \text{    if } c \in \ttsup{N}{reg} \\ 0, \text{   if } c \in \ttsup{N}{dum};\end{cases}} & {v_c(y_j) = \begin{cases} \frac{1}{s + k - 1} - \frac{s + \alpha k - k}{2\alpha k s (s + k - 1)},  \text{ if } c = a^* \\ \frac{1}{2(s + k - 1)},  \text{ if } c \in \ttsup{N}{reg} \\ 1,  \text{   if } c = b^j \\ 0,  \text{   otherwise;} \end{cases}} & {v_c(x_{j'}) = \begin{cases} \frac{1}{2 \alpha s}, \text{ if } c = a^* \\ 0, \text{ if } c \in N \setminus \set{a^*}. \end{cases}} 
\end{matrix}
\]

We can verify that $v_c(V(G)) = \sum_{u \in V(G)} v_c(u) = 1$ for every agent $c \in N$. 

Assume that $(H, k)$ is a yes-instance of {\sc $\beta$-Club}. Let $Q \subseteq V(H)$ be such that $\card{Q} = k$ and $Q$ is a $\beta$-club. Thus $\diam(H[Q]) \leq \beta$. Let us assume without loss of generality that $Q = \set{z_{s - k + 1}, z_{s - k + 2}, \ldots, z_s}$. And we define an allocation $\fn{\pi}{N}{2^{V(G)}}$ as follows: $\pi(a^*) = Q \cup \set{x_j ~|~ j \in [\alpha - 1]}$, $\pi(a^i) = z_i$ for every $i \in [s - k]$ and $\pi(b^j) = y_j$ for every $j \in [s + k -1]$. Observe that $\card{\pi(c)} = 1$ for each agent $c \in \ttsup{N}{reg} \cup \ttsup{N}{dum}$, and hence $\pi(c)$ is strongly $(\alpha, \beta)$-compact. Now, as $\pi(a^*) = Q \cup \set{x_j ~|~ j \in [\alpha - 1]}$, and $\diam(G[Q]) = \diam(H[Q]) \leq \beta$ and $\diam(G[\set{x_j}]) = 0 \leq \beta$ for every $j \in [\alpha - 1]$, we can conclude that $\pi(a^*)$ is strongly $(\alpha, \beta)$-compact as well. Finally, we have $v_{a^*}(\pi(a^*)) = \sum_{j = 1}^k v_{a^*}(z_{s - k + j}) + \sum_{j' = 1}^{\alpha - 1} v_{a^*}(x_{j'}) = k \cdot (1/(2\alpha ks)) + (\alpha - 1) \cdot (1/ (2\alpha s))= 1/2s$. Also, for every $i \in [s - k]$, $v_{a^i}(\pi(a^i)) = v_{a^i}(z_i) = 1/2s$; and for every $j \in [s - k + 1]$, $v_{b^j}(\pi(b^j)) = v_{b^j}(\pi(y_j)) = 1$. Since $v_c(\pi(c)) \geq 1/2s = 1/\card{N}$ for every $c \in N$, $\pi$ is proportional. 

Conversely, consider a proportional strongly $(\alpha, \beta)$-compact allocation $\fn{\phi}{N}{2^{V(G)}}$. Then for every agent $c \in N$, we have $v_c(\phi(c)) \geq 1/(2s)$. Consider $b^j \in \ttsup{N}{dum}$. Since $y_j$ is the only item for which $b^j$ has a non-zero valuation, we can conclude that $y_j \in \phi(b^j)$, which implies that $\phi(\set{a^*} \cup \set{a^i ~|~ i \in [s - k]}) \subseteq \set{z_1, z_2,\ldots, z_s} \cup \set{x_1, x_2,\ldots, x_{\alpha - 1}}$. We can assume without loss of generality that $x_{j'} \notin \phi(a^i)$ for every $i \in [s - k]$ and $j' \in [\alpha - 1]$; if $x_{j'} \in \phi(a^i)$ for some $i, j'$, then we can simply remove $x_{j'}$ from $\phi(a^i)$. Since $v_{a^i}(x_{j'}) = 0$, the resulting allocation would still be proportional; and since $x_{j'}$ is an isolated vertex, the resulting allocation would still be $(\alpha, \beta)$-compact.  Thus $\phi(a^i) \subseteq \set{z_1, z_2,\ldots, z_s}$ for every $i \in [s - k]$. Therefore, since $\phi$ is proportional and hence $\phi(a^i) \neq \emptyset$ for every $i \in [s - k]$, we can conclude that every agent in $\set{a^i ~|~ i \in [s - k]}$ receives at least one item from $\set{z_1, z_2,\ldots, z_s}$. In other words, at least $s - k$ items from $\set{z_1, z_2,\ldots, z_s}$ are allocated to agents from $\set{a^i ~|~ i \in [s - k]}$. Let $J = \set{j \in [s] ~|~ z_j \in \phi(a^i) \text{ for some } i \in [s - k]}$ and let $J^* = [s] \setminus J$. Then $\card{J} \geq s - k$, and hence $\card{J^*} \leq k$. 

We claim that $\card{J^*} = k$. Suppose not. 
Notice that $\phi(a^*) \subseteq {z_j ~|~ j \in J^*} \cup \set{x_{j'} ~|~ j' \in [\alpha - 1]}$.  
Hence $v_{a^*}(\phi(a^*)) \leq \card{J^*} \cdot (1/(2 \alpha k s)) + (\alpha - 1) \cdot (1/(2 \alpha s)) < k/(2\alpha k s) + (\alpha - 1)/(2 \alpha s) = 1/(2s)$. That is, $v_{a^*}(\phi(a^*)) < 1/(2s) = 1/\card{N}$, a contradiction to the fact that $\pi$ is proportional. Notice also that $\set{x_{j'} ~|~ j' \in [\alpha - 1]} \subseteq \phi(a^*)$, for otherwise, we would still have $\phi(a^*) < 1/(2s)$. Thus $\phi(a^*) = \set{z_j ~|~ j \in J^*} \cup \bigcup_{j' = 1}^{\alpha - 1} \set{x_{j'}}$. Since $\phi$ is $(\alpha, \beta)$-compact and since each $x_{j'}$ is an isolated vertex, we must have $\dist_{G[\set{z_j ~|~ j \in J^*}]}(z, z') \leq \beta$, for every $z, z' \in \set{z_j ~|~ j \in J^*}$. Thus $\diam(G[\set{z_j ~|~ j \in J^*}]) \leq \beta$, which implies that $\diam(H[\set{z_j ~|~ j \in J^*}]) \leq \beta$. That is, $\set{z_j ~|~ j \in J^*}$ is a $\beta$-club in $H$, and since $\card{J^*} = k$, we can conclude that $(H, k)$ is a yes-instance of {\sc $\beta$-Club}.

\end{proof}

A few modifications to the reduction in the proof of Theorem~\ref{thm:prop11} yields the following result. 

\begin{theorem}\label{thm:cef11}
For every fixed $\alpha, \beta \geq 1$, \scefn{\alpha}{\beta} is strongly \NPH.
\end{theorem}

\begin{proof}
We again reduce from {\sc $\beta$-Club}. Given an instance $(H, k)$ of {\sc $\beta$-Club}, where $V(H) = \set{z_1, z_2, \ldots, z_s}$, we proceed as in the proof of Theorem~\ref{thm:prop11}, but with the following differences. (1) We assume without loss of generality that $s + \alpha - 1$ is a multiple of $k + \alpha - 1$, i.e., $s + \alpha - 1 = p(k + \alpha - 1)$ for some integer $p$. (Notice that we can always add isolated vertices to $H$ to ensure this property.) (2) We take $G$ to be the union of $H$ and $p + k + \alpha - 1$ isolated vertices $y_1, y_2,\ldots, y_{p}$, $y'_1, y'_2,\ldots, y'_{k}$ and $x_1, x_2,\ldots, x_{\alpha - 1}$. (3) We have one special agent $a^*$ and a set of $s$ regular agents $\ttsup{N}{reg} = \{a^1, a^2,\ldots, a^s \}$ and a set of $p$ dummy agents $\ttsup{N}{dum} = \{b^1, b^2,\ldots, b^p \}$.  (4) Finally, we define the valuation functions. For $c \in N$, $i \in [s]$, $j' \in [\alpha - 1]$, we have 
\[
\begin{matrix}
{v_c(z_i) = \begin{cases} 1/(2(s + \alpha - 1)), & \text{ if } c = a^* \\ 1/(k + 1), & \text{ if } c = a^i  \\ 0, & \text{ otherwise; } \end{cases}} & 
{v_c(x_{j'}) = \begin{cases} 1/(2(s + \alpha - 1)), & \text{ if } c = a^* \\ 0, & \text{ otherwise. } \end{cases}}
\end{matrix}
\]

For $c \in N$, $j \in [p]$ and $\ell \in [k]$, we have 
\[
\begin{matrix}
{v_c(y_j) = \begin{cases} 1/(2p),  & \text{ if } c = a^* \\ 1, & \text{   if } c = b^j  \\ 0, & \text{    otherwise;} \end{cases}} & {v_c(y'_{\ell}) = \begin{cases} {1}/{(k + 1)},  & \text{ if } c \in \ttsup{N}{reg} \\0,  & \text{   otherwise.} \end{cases}} 
\end{matrix}
\]

We can verify that  $\sum_{x \in V(G)} v_c(x) = 1$ for every agent $c \in N$. (Recall that $s + \alpha - 1 = p(k + \alpha - 1)$.)

Now, suppose that $H$ contains a $\beta$-club $Q \subseteq V(H)$ of size exactly $k$. Assume without loss of generality that $Q = \set{z_{1}, z_{2}, \ldots, z_k}$. We define an allocation $\fn{\pi}{N}{2^{V(G)}}$, where $\pi(a^*) = Q \cup \set{x_{j'} ~|~ j' \in [\alpha - 1]}$; for $i \in [s]$, $\pi(a^i) = z_i$ if $i > k$ and $\pi(a^i) = y'_{i}$ if $i \leq k$; and for $j \in [p]$, $\pi(b^j) = y_j$. 
Since $Q$ is a $\beta$-club in $H$ and hence in $G$ as $G[Q] = H[Q]$, we have, $\diam(G[Q]) \leq \beta$; and since $\pi(a^*)$ is the disjoint of $Q$ and $\alpha - 1$ vertices, $G[\pi(a^*)]$ is strongly $(\alpha, \beta)$-compact. Also, since $\card{\pi(c)} = 1$ for every agent $c \in N \setminus \set{a^*}$, $G[\pi(c)]$ is strongly $(\alpha, \beta)$-compact as well. 
To see that $\pi$ is an envy-free allocation, observe the following facts. 
\begin{enumerate}
\item The agent $a^*$ is envy-free. 
\begin{enumerate} 
\item We have $v_{a^*}(\pi(a^*)) = \sum_{i = 1}^{k} v_{a^*}(z_i) + \sum_{j' = 1}^{\alpha - 1} = k \cdot (1/(2(s + \alpha - 1))) + (\alpha - 1) \cdot (1/(2(s + \alpha - 1)))= (k + \alpha - 1)/(2(s + \alpha - 1)) = 1/(2p)$. The last equality follows from the fact that $s + \alpha - 1  = p(k + \alpha - 1)$. 
\item For $i \in [s]$ with $i > k$, we have $v_{a^*}(\pi(a^i)) = v_{a^*}(z_i) = 1/(2(s + \alpha - 1)) < 1/(2p)$. 
\item And for $i \in [s]$ with $i \leq k$, we have $v_{a^*}(\pi(a^i)) = v_{a^*}(y'_i) = 0$ if $i \leq k$.  
\item For $j \in [p]$, $v_{a^*}(\pi(b^j)) = v_{a^*}(y_j) = 1/(2p)$. 
\end{enumerate}
\item For every $j \in [p]$, $b^j$ is envy-free as we have $v_{b^j}(\pi(b^j)) = v_{b^j}(y_j) = 1$ and $v_{b^j}(\pi(c)) \leq 1$ for every agent $c \in N$.  \item For every $i \in [s]$, $a^i$ is envy-free. 
\begin{enumerate} 
\item First, $v_{a^i}(\pi(a^i)) = v_{a^i}(\pi(z_i)) = 1/(k + 1)$ if $i > k$ and $v_{a^i}(\pi(a^i)) = v_{a^i}(\pi(y'_i)) = 1/(k + 1)$ if $i \leq k$. 
\item And $v_{a^i}(\pi(a^*)) = v_{a^i}(Q \cup \set{x_{j'} ~|~ j' \in [\alpha - 1]}) = \sum_{j = 1}^k v_{a^i}(z_j) + \sum_{j' = 1}^{\alpha - 1} v_{a^i}(x_{j'})$ Now, $v_{a^i}(x_{j'}) = 0$ for every $j' \in [\alpha - 1]$. Notice also that  $\sum_{j = 1}^k v_{a^i}(z_j) = 1/(k + 1)$ if $i \leq [k]$ (as $v_{a^i}(z_i) = 1/(k + 1)$) and $ \sum_{j = 1}^k v_{a^i}(z_j) = 0$ if $i > k$. In either case, $v_{a^i}(\pi(a^*)) \leq 1/(k + 1) = v_{a^i}(\pi(a^i))$. 
\item For distinct $i, j \in [s]$, $v_{a^i}(\pi(a^j)) \leq 1/(k + 1)$. 
\item For $i \in [s]$ and $j \in [p]$, $v_{a^i}(\pi(b^j)) = v_{a^i}(y_j) = 0$. 
\end{enumerate}
\end{enumerate}

Conversely, consider a complete, envy-free and strongly $(\alpha, \beta)$-compact allocation $\fn{\phi}{N}{2^{V(G)}}$. 
Since $\phi$ is complete, for every $\ell \in [k]$, the vertex $y'_{\ell}$ is allocated. But notice that $v_{a^i}(y'_{\ell}) = 1/(k + 1)$ for every $i \in [s]$. Then, since $\phi$ is envy-free, we must have $v_{a^i}(\phi(a^i)) \geq 1/(k + 1)$ for every $i \in [s]$. But notice that the agent $a^i$ has non-zero valuations only for the vertex $z_i$ and the $k$ vertices $y'_1,\ldots,y'_{k}$. Therefore, for every $i \in [s]$, either $z_i \in \phi(a^i)$ or $y'_{\ell} \in \phi(a^i)$ for some $\ell \in [k]$. 
Let $S = \set{z_i \in V(G) ~|~ z_i \in \phi(a^i)}$. Notice that $\card{S} \geq s - k$ as there exist at most $k$ agents $a^i$ with $y'_{\ell} \in \phi(a^i)$ for some $\ell \in [k]$. 
Now, since $\phi$ is complete and envy-free and since $v_{b^j}(y_j) = 1$ for every $j \in [p]$, we must have $y_j \in \phi(b^j)$. But notice that $v_{a^*}(\phi(b^j)) \geq v_{a^*}(y_j) = 1/(2p)$. Since $\phi$ is envy-free, we must have $v_{a^*}(\phi(a^*)) \geq v_{a^*}(\phi(b^j)) \geq 1/(2p)$. 
Let $R = \set{z_i |~ i \in [s] \text{ and } z_i \in \phi(a^*)}$. First, notice that $R \subseteq \set{z_1, z_2,\ldots, z_s} \setminus S$. Then, since $\card{S} \geq s - k$, we have $\card{R} \leq k$. 
Let us first observe that $\card{R} = k$. Suppose not. Then, $v_{a^*}(\phi(a^*)) \leq v_{a^*}(R)  + v_{a^*}(\set{x_{j'} ~|~ j' \in [\alpha - 1]}) = \card{R} \cdot (1/(2(s + \alpha - 1))) + (\alpha - 1) \cdot (1/(s + \alpha - 1)) < (k + \alpha - 1)/(2(s + \alpha - 1)) = 1/(2p)$. That is, $v_{a^*}(\phi(a^*)) < 1/(2p)$, which is not possible as the agent $a^*$ is envy-free. Thus $\card{R} = k$. Now, notice also that $x_{j'} \in \phi(a^*)$ for every $j' \in [\alpha - 1]$, for otherwise, we would still have $v_{a^*}(\phi(a^*)) < 1/(2p)$. Thus $\phi(a^*)$ contains $R$ and $\alpha -1$ isolated vertices $x_1, x_2,\ldots, x_{\alpha - 1}$. Then, since $\phi$ is $(\alpha, \beta)$-compact, we can conclude that $\dist_{G[R]}(z, z') \leq \beta$ for every $z, z' \in R$. Therefore, $\diam(G[R]) \leq \beta$, which implies that $\diam(H[R]) \leq \beta$. Thus $R$ is a $\beta$-club of size $k$ in $H$. 
\end{proof}

\section{Proportional (Strongly) \texorpdfstring{$(1, \beta)$}{(1, beta)}-Compact Allocations on Paths}
In this section, we consider we consider \propn{1}{\beta} and \spropn{1}{\beta} in a highly restricted input setting: The graph $G$ on the set of items is a path. 
We show that these problems admit \FPT\ algorithms when parameterized by the number of agents, and \XP\ algorithms when parameterized by the number of types of agents. We first observe the following fact. 
\begin{observation}\label{obs:paths}
Fix $\beta \geq 0$. 
A path $P$ is $(1, \beta)$-compact if and only if $P$ has at most $2\beta + 1$ vertices. And $P$ is strongly $(\alpha, \beta)$-compact if and only if $P$ has at most $\beta + 1$ vertices. 
\end{observation}

We now design a dynamic programming algorithm for \propn{1}{\beta}. The algorithm for \spropn{1}{\beta} only requires a slight adaptation. 
In the remaining part of this section, we assume that we are given a instance $(G, N, \ca{V})$ of \propn{1}{\beta}, where $G$ is a path  with $V(G) = \set{z_1, z_2,\ldots, z_m}$ and $E(G) = \{z_i z_{i + 1} ~|~ i \in [m - 1]\}$. For $p, q \in [m]$ with $p \leq q$, let $V_{p, q} = \set{z_p, z_{p + 1}, \ldots, z_{q}}$ and $G_{p, q} = G[V_{p, q}]$. 

We now describe our DP. For each $i, j \in [m]$ with $j \leq i$ and for each non-empty $N' \subseteq N$, we define $\aco[i, j, N'] = 1$ if there exists a $(1, \beta)$-compact, proportional allocation $\fn{\pi}{N'}{2^{V(G_{1, i})}}$ such that no item that appears after $z_j$ in the ordering $z_1, z_2,\ldots, z_i$  is allocated under $\pi$; and $\aco[i, j, N'] = 0$, otherwise. 
That is, $\aco[i, j, N'] = 1$ if each agent in $N'$ receives a proportional, $(1, \beta)$-compact bundle and no $z_{j'}$, where $j' > j$, gets allocated. Notice then that $(G, N, \ca{V})$ is a yes-instance of \propn{1}{\beta} if and only if $\aco[m, j, N] = 1$ for some $j \in [m]$.  
To compute the entries, we use the recurrence
\[
\aco[i, j, N'] = \bigvee_{\substack{{(i', j', a)} \\ {i' \leq j} \\ {j' \leq i'}\\ {a \in N'} \\ {\card{V_{j' + 1, i'}} \leq 2\beta + 1} \\ {v_a(V_{j' + 1, i'}) \geq (1/ {n}) \cdot v_a(V(G))}}} \aco[i', j', N'\setminus \set{a}]. \tag{\textsl{Eq. Agents-Path-$(1, \beta)$-Com}}\label{eq:com}
\]

Notice that the recurrence \ref{eq:com} simply checks if it is feasible to allocate all the vertices from $z_{j' + 1}$ to $z_{i'}$ to some agent $a \in N'$, while ensuring compactness and proportionality for every agent in $N'$. The condition $\card{V_{j' + 1, i'}} \leq 2\beta + 1$ ensures that the allocation remains $(1, \beta)$-compact, and the condition $v_a(V_{j' + 1, i'}) \geq 1/ {n} \cdot v_a(V(G))$ ensures proportionality. 
The base case of the recurrence is when $i = j = \card{N'} = 1$. And for this case, notice that we have $A[1, 1, \set{a}] = 1$ if and only if the allocation $\pi$ that allocates the vertex $z_1$ to agent $a$ is proportional, i.e., if and only if $v_a(z_1) \geq (1/n) \cdot v_a(V(G))$. So for every $a \in N$, we set  $A[1, 1, \set{a}] = 1$ if $v_a(z_1) \geq 1/n$, and set $A[1, 1, \set{a}] = 0$ otherwise. Finally, notice that the number of choices for $(i, j, N')$ is $m^2 \cdot 2^n$, and that the time taken to compute each $A[i, j, N]$ is bounded by the number of choices for $(i', j', a)$, which is at most $m^2 \cdot n$. Thus the total time taken by the algorithm is $\cO(2^n \cdot n \cdot m^4)$. 

\begin{remark}
Observe the following two facts. (1) The DP does not require that the valuations be additive. (2) By slightly modifying the definition of $\aco$, we can solve \spropn{1}{\beta} as well. We define $\asc[i, j, N'] = 1$ if there exists a strongly $(1, \beta)$-compact, proportional allocation $\fn{\pi}{N'}{2^{V(G_{1, i})}}$ such that no item that appears after $z_j$ in the ordering $z_1, z_2,\ldots, z_i$  is allocated under $\pi$; and $\asc[i, j, N'] = 0$, otherwise. And we can compute $\asc[i, j, N']$ using a similar recurrence relation. In light of Observation~\ref{obs:paths}, we  just need to replace the condition $\card{V_{j' + 1, i'}} \leq 2\beta + 1$ in \ref{eq:com} with $\card{V_{j' + 1, i'}} \leq \beta + 1$. 
\end{remark}

We summarise the discussion so far in this section in the following theorem. 

\begin{theorem}\label{thm:paths}
\propn{1}{\beta} and \spropn{1}{\beta} are FPT when parameterized by the number of agents. These results hold even when the valuations are not additive. 
\end{theorem}

Now, notice that if the agents' valuations are identical, then we need not keep track of the set $N'$ of agents who have already received a proportional and compact bundle of items. We simply need to remember \emph{how many} agents have received bundles so far. That is, instead of going over all tuples $(i, j, N')$, we only need to consider tuples $(i, j, n')$, where $0 \leq n' \leq n$. And we can define and compute $\aco[i, j, n']$ and $\asc[i, j, n']$ accordingly. Thus instead of exponential dependence on the number of agents, the algorithm now will only have a polynomial dependence on the number of agents. 

We can generalise the above idea further to design an \XP\ algorithm parameterized by the number of types of agents. We say that two agents $a, a' \in N$ are of the same type if their valuations are identical. Let $p$ be the number of distinct types of agents, and let the types be numbered from $1$ to $p$. Also for each $q \in [p]$, let $n_q$ be the number of agents of type $q$ and let $v_q$ denote the valuation function corresponding to an agent of type $j$. To design our algorithm, as before, we process the graph from $z_1$ to $z_m$, and at every stage and for each type $q \in [p]$, we keep track of the number of agents of type $q$ who received bundles. Formally, for each $i,j \in [m]$, and each tuple $(\ell_1, \ell_2,\ldots,\ell_{p})$, where $0 \leq \ell_q \leq n_q$ for every $q \in [p]$, we define $\bco[i, j, (\ell_1, \ell_2,\ldots,\ell{p})] = 1$ if there exists a $(1, \beta)$-compact, proportional allocation $\fn{\pi}{N}{2^{V(G_{1, i})}}$ such that no item that appears after $z_j$ in the ordering $z_1, z_2,\ldots, z_i$  is allocated under $\pi$, and for each $q \in [p]$, exactly $\ell_q$ agents of type $q$ have received bundles under $\pi$; and $\bco[i, j, (\ell_1, \ell_2,\ldots,\ell{p})] = 0$, otherwise. Notice that $(G, N, \ca{V})$ admits a proportional, $(1, \beta)$-compact allocation if and only if $\bco[m, j, (n_1, n_2,\ldots, n_p)] = 1$ for some $j \in [m]$. 

For each $p$-tuple $(\ell_1, \ell_2,\ldots,\ell{p})$ and for each $q \in [p]$ with $\ell_q > 0$, let $\ell^q_r = \ell_r$ if $r \in [p] \setminus \set{q}$ and $\ell^q_q = \ell_q - 1$. Now, to compute $\bco[i, j, (\ell_1, \ell_2,\ldots,\ell_{p})]$, we have the recurrence
\[
\bco[i, j, (\ell_1, \ell_2,\ldots,\ell_{p})] = \bigvee_{\substack{{(i', j', q)} \\ {i' \leq j} \\ {j' \leq i'}\\ {q \in [p]} \\ \ell_q > 0 \\ {\card{V_{j' + 1, i'}} \leq 2\beta + 1} \\ {v_q(V_{j' + 1, i'}) \geq (1/ {n}) \cdot v_q(V(G))}}} \bco[i, j, (\ell^q_1, \ell^q_2,\ldots,\ell^q_{p})]. \tag{\textsl{Eq. Types-Path-$(1, \beta)$-Com}}\label{eq:com-xp}
\]
For the base case of the recurrence, for every $i, j \in [m]$, we set $\bco[i, j, (0, 0,\ldots,0)] = 1$; notice that in this case, the empty allocation trivially satisfies all the required conditions. Also, for each $q \in [p]$, let $\overline 1_q$ denote the $p$-tuple $(\ell_1, \ell_2,\ldots,\ell_p)$, where $\ell_r = 0$ if $r \in [p] \setminus \set{q}$ and $\ell_q = 1$. And we set $\bco[1,1, \overline 1_q] = 1$ if $v_q(z_1) \geq (1/n) \cdot v_q(V(G))$. 

Notice that the number of choices for the tuple $(i, j, (\ell_1, \ell_2,\ldots, \ell_p))$ is at most $m^2 \cdot n^p$. And notice that to  compute each $\bco[i, j, (\ell_1, \ell_2,\ldots,\ell_{p})]$, we only need to go over the  entries $\bco[i', j', (\ell^q_1, \ell^q_2,\ldots,\ell^q_{p})]$ for all triplets $(i', j', q)$, and the number of choices for such triplets is at most $m^2 \cdot p$. Hence the total runtime of the algorithm is bounded by $\cO(p \cdot n^p \cdot m^4)$. Again, we can solve the strongly compact variant by slightly modifying the DP. We thus have the following result. 

\begin{theorem}\label{thm:pathsid}
\propn{1}{\beta} and \spropn{1}{\beta} admit \XP\ algorithms parameterized by the number of types of agents. These results hold even when the valuations are not additive. 
\end{theorem}

\begin{remark}
The fact that $\alpha = 1$ is crucial for the results in this section. We cannot generalise our strategy here to $\alpha > 1$. Suppose $\alpha > 1$. It might be tempting to try and remember which subset of agents $N' \subseteq N$ have received bundles so far, and how many connected pieces each of them received. So, along with $i, j$, we would keep track of $(N_0, N_1, N_2,\ldots, N_{\alpha})$, where for each $\ell \in [\alpha]$, $N_{\ell}$ is the set of agents who received exactly $\ell$ many connected pieces of the path. But this information alone is inadequate. We would also need to remember the valuation of each agent for the set of items she has received so far. 
\end{remark}

\section{Fair Division of Graphs of Bounded Treewidth}\label{sec:treedp}
In this section, we show that for each of the three fairness concepts---proportionality, envy-freeness and maximin fairness---the corresponding problem of checking if there exists an allocation that is fair and $(\alpha, \beta)$ compact admits a pseudo-\XP\ algorithm, when parameterized by the treewidth of the item graph and the number of agents. We design a single dynamic programming procedure that works for all three fairness concepts. {\bf We assume in this section that the valuations are \emph{integer-valued}.} Throughout this section, we will use terminology and results from Sections~\ref{sec:prelims}, \ref{sec:diameter} and \ref{sec:annotated}. We encourage the reader to have a quick look at these sections before reading any further. 
The following theorem is the main contribution of this section. 

\begin{theorem}\label{thm:treedp}
For every $\alpha$ and $\beta$, \propn{\alpha}{\beta}, \efn{\alpha}{\beta} and \mmsn{\alpha}{\beta} admit algorithms that run in time $(\tw + n)^{\cO(\tw + n)} \cdot \beta^{\tw} \cdot m^{\cO(\alpha n)} \cdot \maxval^{\cO(n^2)}$, where $\tw$ is the treewidth of the input graph $G$ and $\maxval$ is the maximum valuation of an agent for $G$, i.e., $\maxval = \max_{i \in N} v_i(V(G))$. 
\end{theorem}

As an immediate corollary of Theorem~\ref{thm:treedp}, we derive the following result. 
\begin{theorem}\label{thm:planar}
For every $\alpha$ and $\beta$, \propn{\alpha}{\beta}, \efn{\alpha}{\beta} and \mmsn{\alpha}{\beta} admit algorithms that run in time $(\alpha \beta n)^{\cO(\alpha \beta n)} \cdot \beta^{\alpha \beta n} \cdot m^{\cO(\alpha n)} \cdot \maxval^{\cO(n^2)}$ when the item graph $G$ is planar.  
\end{theorem}
\begin{proof}
To prove the theorem, we use the following facts. (1) For a planar graph $G$, $\tw(G) = \cO(D)$, where $D$ is diameter of $G$~(see, for example, \cite{DBLP:journals/algorithmica/Eppstein00}). (2) For a graph $G$, $\tw(G) = \max_{H}\tw(H)$, where the maximum is over all connected components $H$ of $G$. (3) Given an instance $(G, N, \ca{V})$ of \xcompact{\alpha}{\beta}, by Remark~\ref{rem:rrule}, we can assume without loss of generality that $\diam(H) = \cO(\alpha \beta n)$ for every connected component $H$ of $G$. Facts (1), (2) and (3) together imply that for an instance $(G, N, \ca{V})$ of \xcompact{\alpha}{\beta}, we have $\tw(G) = \cO(n \alpha \beta)$ if $G$ is planar. The theorem then follows immediately from Theorem~\ref{thm:treedp}.  
\end{proof}

We now prove Theorem~\ref{thm:treedp}. In light of Lemma~\ref{lem:alpha1}, to solve \xcompact{\alpha}{\beta}, it is enough to solve \xanno, for which we design an algorithm. 
More precisely, we prove the following lemma. 
\begin{lemma}\label{lem:treedp}
There is an algorithm that, given an instance $(G, [n], \ca{V}, \beta, (\hat z_i)_{i \in [n]})$ of \anno, an $n^2$ tuple $(\hat w_{ij})_{i, j \in [n]}$, where $0 \leq w_{ij} \leq \maxval$ for every $i, j \in [n]$ and a nice tree decomposition $(T, \{X_t ~|~ t \in V(T)\})$ of $G$ as input, runs in time $(\tw + n)^{\cO(\tw + n)} \cdot \beta^{\cO(\tw + n)} \cdot m^{\cO(1)} \cdot \maxval^{\cO(n^2)}$, and correctly decides if $(G, [n], \ca{V}, \beta, (\hat z_i)_{i \in [n]}))$ admits a $((\hat z_i)_{i \in [n]}; \beta)$-annotated allocation $\fn{\pi}{[n]}{2^{V(G)}}$ such that $v_i(\pi(j)) = \hat w_{ij}$ for each $i, j \in [n]$.
\end{lemma}

The rest of this section is dedicated to proving   Lemma~\ref{lem:treedp}. And for that, we design a dynamic programming algorithm over a (nice) tree decomposition of $G$. 
From now on, we assume that we are given $(G, [n], \ca{V}, \beta, (\hat z_i)_{i \in [n]})$, an $n^2$ tuple $(\hat w_{ij})_{i, j \in [n]}$, where $0 \leq w_{ij} \leq \maxval$ for every $i, j \in [n]$ and a nice tree decomposition $(T, \{X_t ~|~ t \in V(T) \})$ of $G$ of width $\tw$. {\bf We first add $\bm{\hat z_i}$ for every $\bm{i \in [n]}$, to every bag $\bm{X_t}$} (and denote the resulting tree decomposition by $(T, \{X_t ~|~ t \in V(T) \})$ as well). Notice that this increases the width of the tree decomposition by $n$; and  $(T, \{X_t ~|~ t \in V(T) \})$ still remains a nice tree decomposition but for the fact that the bags corresponding to the root node and the leaf nodes of $T$ are non-empty. In particular, for $t \in V(T)$, where $t$ is a leaf node or $t$ is the root of $T$, we have $X_t = \set{\hat z_1, \hat z_2,\ldots, \hat z_n}$. Following standard convention, for $t \in V(T)$, we use $G_t$ to denote the subgraph of $G$ that consists of all the vertices and edges introduced in the subtree rooted at $t$.

\paragraph*{Outline of our algorithm.} Suppose that $\pi$ is the $((\hat z_i)_{i \in [n]}; \beta)$-annotated allocation that we are looking for. Then for each $i \in [n]$ and for each $z \in \pi(i)$, $\dist_{G[\pi(i)]}(\hat z_i, z) \leq \beta$. Let $H_i$ be a BFS tree of $G[\pi(i)]$, rooted at $\hat z_i$. Since $H_i$ is a BFS tree rooted at $\hat z_i$, we have $\dist_{G[\pi(i)]}(\hat z_i, z) = \dist_{H_i}(\hat z_i, z)$ for every $z \in \pi(i)$. For a node $t \in V(T)$, we guess how $H_i$ intersects with the graph $G_t$. Notice that $H_i$ is an acyclic graph and it could possibly split into multiple connected components when intersected with $G_t$. With this in mind, we guess the following features of the intersection of $H_i$ with $G_t$: (a) the intersection of $V(H_i) (= \pi(i))$ with $X_t$, (b) how $H_i$ partitions the vertices of $X_t \cap \pi(i)$ into connected components (c) for each $z \in X_t \cap \pi(i)$, the distance in $H_i$ between $z$ and $\hat z_i$ and (4) agent $i$'s valuations for the subsets of goods in $G_t$ that have been allocated to herself as well as the other agents, i.e., $v_i(\pi(j) \cap V(G_t))$ for every $i, j \in [n]$. We argue that these four pieces of information are sufficient to design a dynamic programming algorithm that constructs $\pi$ in a bottom up fashion over $T$. 

\begin{sloppypar}
\paragraph*{Designing the DP: Necessary Ingredients.} To design our DP, for each node $t \in V(T)$, we first define a set $\ca{C}_t$, (which formalises our guesses). Formally, for each $t \in V(T)$, let $\ca{C}_t$ be the set of all tuples $((S_i, f_i, (\ca{P}_i, R_i))_{i \in [n]}, (w_{ij})_{i, j \in [n]})$ with the following properties. \begin{itemize}
\item For each $i \in [n]$,  $S_i$ is a vertex subset such that $\hat z_i \in S_i$ and $S_i \subseteq X_t \cap B_G(\hat z_i, \beta)$ and $S_i \cap S_j = \emptyset$ for every $j \in [n] \setminus \set{i}$. 
\item For each $i \in [n]$, $(\ca{P}_i, R_i)$ is a rooted partition of $S_i$ such that $\hat z_i \in R_i$. 
\item For each $i \in [n]$, $\fn{f_i}{S_i}{[\beta]_0}$ is a function such that $f_i(\hat z_i) = 0$ and $f_i(z) \neq 0$ for every $z \in S_i \setminus \set{\hat z_i}$. 
\item The tuple $(w_{ij})_{i, j \in [n]}$ is an $n^2$-tuple of non-negative integers such that $0 \leq w_{ij} \leq \maxval$ for every $i, j \in [n]$. 
\end{itemize}

For example, when $n = 2$, the set $\ca{C}_t$ consists of all tuples of the form $((S_1, f_1, (\ca{P}_1, R_1)), (S_2, f_2, (\ca{P}_2, R_2)), (w_{11}, w_{12}, w_{21}, w_{22}))$. We now bound $\card{\ca{C}_t}$. 
\end{sloppypar}
\begin{observation}\label{obs:Ctsize}
For each $t \in V(T)$, we have $\card{\ca{C}}_t \leq (\tw + n)^{\cO(\tw + n)} \cdot \beta^{\cO(\tw)} \cdot \maxval^{\cO(n^2)}$. 
To see this, observe the following facts that follow from the definition of $\ca{C}_t$. (i) As the sets $S_1, S_2,\ldots, S_n \subseteq X_t$ are pairwise disjoint, each element of $X_t$ belongs to at most one $S_i$, and therefore, the number of choices for $(S_1, S_2,\ldots, S_n)$ is $(n + 1)^{\card{X_t}}$. Identical reasoning applies to the sets $R_1, R_2,\ldots, R_n$ as well, and therefore the number of choices for $(R_1, R_2,\ldots, R_n)$ is at most $(n + 1)^{\card{X_t}}$. 
(ii) For each choice of $(S_1, S_2,\ldots, S_n)$, notice that the family $\bigcup_{i = 1}^{n} \ca{P}_i$ is a partition of $\bigcup_{i = 1}^n S_i \subseteq X_t$. Therefore, corresponding to each $(S_1, S_2,\ldots, S_n)$, the number of choices for $(\ca{P}_1, \ca{P}_2,\ldots, \ca{P}_n)$ is at most $\card{X_t}^{\card{X_t}}$. 
(iii) To bound the number of choices for $(f_1, f_2,\ldots, f_n)$, notice again that as the sets $S_1, S_2,\ldots, S_n$ are pairwise disjoint, we may think of the $n$-tuple $(f_1, f_2,\ldots, f_n)$ of functions as a single function from $\bigcup_{i = 1}^{n}S_i$ to $[\beta]_0$. Therefore, the number of choices for $(f_1, f_2,\ldots, f_n)$ is at most $(\beta + 1)^{\card{X_t}}$. (iv) We have $\card{X_t} \leq \tw + 1 + n$;  the ``$+n$'' accounts for the fact that we added $\hat z_i$ for every $i \in [n]$ to every bag. (v) Putting these together, the number of choices for $(S_i, f_i, (\ca{P}_i, R_i))_{i \in [n]}$ is at most $(n + 1)^{\card{X_t}} \cdot (n + 1)^{\card{X_t}} \cdot \card{X_t}^{\card{X_t}} \cdot (n + 1)^{\card{X_t}} \cdot (\beta + 1)^{\card{X_t}} \leq (\tw + n)^{\cO(\tw + n)} \cdot \beta^{\cO(\tw + n)}$. (vi) The number of choices for the tuple $(w_{ij})_{i,j \in [n]}$ is at most $(\maxval + 1)^{n^2}$. (vii) From points (v) and (vi), we can conclude that the number of choices for $((S_i, f_i, (\ca{P}_i, R_i))_{i \in [n]}, (w_{ij})_{i, j \in [n]})$ is at most $(\tw + n)^{\cO(\tw + n)} \cdot \beta^{\cO(\tw + n)} \cdot (\maxval + 1)^{n^2}$, which bounds $\card{\ca{C}}_t$. 
\end{observation}

\paragraph*{Valid allocations.} We now define what we call valid allocations. Informally, these are allocations that agree with our guesses (i.e., the tuples in $\ca{C}_t$). For a node $t \in V(T)$, recall that $G_t$ is the subgraph of $G$ made up of all the vertices and edges introduced in the subtree rooted at $t$. 
Consider a node $t \in V(T)$ and an allocation $\fn{\pi}{[n]}{2^{V(G_t)}}$. For a tuple $\eta \in \ca{C}_t$, where $\eta = ((S_i, f_i, (\ca{P}_i, R_i))_{i \in [n]}, (w_{ij})_{i, j \in [n]})$, we say that $\pi$ is \emph{$(t, \eta)$-valid} if 
\begin{enumerate}[nosep,label={(\textsl{VA}.\arabic*)}]
\item\label{inter} for each $i \in [n]$, $\pi(i) \cap X_t = S_i$; 
\item\label{span} for each $i \in [n]$, $G_t[\pi(i)]$ has a spanning forest $H_i$ such that 
\begin{enumerate}[label=(\alph*)]
\item\label{comp}  for every $z, z' \in S_i$, $z$ and $z'$ are in the same connected component of $H_i$ if and only if $z$ and $z'$ are in the same block of $\ca{P}_i$; 
\item\label{dist} for every $z \in S_i$, $f_i(z)= f_i(x) + \dist_{H_i}(x, z)$, where $x$ is the root of $\blk_{\ca{P}_i}(z)$;
\item\label{beta}  for every $z \in \pi(i) \setminus X_t$, there exists $y \in R_i$ such that $\dist_{H_i}(y, z) \leq \beta - f_i(y)$; and
\end{enumerate}
\item\label{value} for each $i, j \in [n]$, $v_i(\pi(j))= w_{ij}$.
\end{enumerate}

For $t \in V(T), \eta \in \ca{C}_{t}$, a $(t, \eta)$-valid allocation $\fn{\pi}{[n]}{2^{V(G_t)}}$ and $i \in [n]$, we call a spanning forest $H_i$ of $G_t[\pi(i)]$ that satisfies conditions \ref{span}\ref{comp}--\ref{span}\ref{beta} a \emph{$(t, \eta, i)$-witness for $\pi$}. 

The correctness of our DP crucially relies on the following lemma, which says that every connected component of a $(t, \eta, i)$-witness intersects $X_t$. 
\begin{lemma}\label{lem:root}
Consider $i \in [n]$. For $t \in V(T)$ and $\eta = ((S_i, f_i, (\ca{P}_i, R_i))_{i \in [n]}, (w_{ij})_{i, j \in [n]}) \in \ca{C}_{t}$, let $\fn{\pi}{[n]}{2^{V(G_t)}}$ be a $(t, \eta)$-valid allocation and $H_i$ a $(t, \eta, i)$-witness for $\pi$. Then for each connected component $H$ of $H_i$, there exists a unique vertex $x_H \in V(H) \cap R_i$. 
\end{lemma}

\begin{proof}
Let $H$ be a connected component of $H_i$. 
Suppose first that $V(H) \cap S_i \neq \emptyset$. Then, as $H$ is a connected component of $H_i$, by condition~\ref{span}\ref{comp}, $V(H) \cap S_i$ is a block of $\ca{P}_i$. And by the definition of a rooted partition, we have $\card{V(H) \cap S_i \cap R_i} = 1$. That is, there exists a unique element $x_H \in R_i$ such that $\set{x_H} = V(H) \cap S_i \cap R_i$. The lemma thus holds. Notice now that it is always the case that $V(H) \cap S_i \neq \emptyset$. Suppose for a contradiction that $V(H) \cap S_i = \emptyset$. Let $z \in V(H)$. Then $V(H) \cap X_t = \emptyset$, and by condition~\ref{span}\ref{beta}, there exists $y \in R_i$ such that $\dist_{H_i}(y, z) \leq \beta - f_i(y)$. Since $H$ is the connected component component of $H_i$ that contains $z$, we have $\dist_H(y, z) = \dist_{H_i}(y, z) \leq \beta - f_i(y)$. Thus $y \in V(H)$, and in particular, $y \in V(H) \cap S_i$, which is a contradiction.  
\end{proof}

\begin{observation}[Structure of a $\bm{(t, \eta, i)}$-witness]
\label{obs:structure}
For $i \in [n]$, we can think of a $(t, \eta, i)$-witness $H_i$ as a rooted forest---each connected component is a rooted tree; $R_i$ is precisely the set of roots of the components of $H_i$, each connected component of $H_i$ intersects $X_t$ and the partition $\mathcal{P}_i$ is such that for $z, z' \in S_i$, $z$ and $z'$ are in the same block of $\ca{P}_i$ if and only if $z$ and $z'$ are in the same connected component of $H_i$. 
\end{observation}

Before proceeding further with our DP, let us first see how valid  allocations help us prove Lemma~\ref{lem:treedp}. 
\begin{lemma}\label{lem:compatible}
Consider $(G, [n], \ca{V}, \beta, (\hat z_i)_{i \in [n]})$, $(\hat w_{ij})_{i, j \in [n]}$, and $(T, \{X_t ~|~ t \in V(T)\})$ of $G$ as defined in Lemma~\ref{lem:treedp}. And consider an allocation $\fn{\pi}{[n]}{2^{V(G)}}$. Then $\pi$ is a $((\hat z_i)_{i \in [n]}; \beta)$-annotated  allocation with $v_i(\pi(j)) = \hat w_{ij}$ for every $i, j \in [n]$ if and only if $\pi$ is $(\hat t, \hat \eta)$-valid, where $\hat t$ is the root of the tree $T$ and $\hat \eta = ((S_i, f_i, (\ca{P}_i, R_i))_{i \in [n]}, (\hat w_{ij})_{i, j \in [n]})\in \ca{C}_{\hat t}$  with $S_i = \set{\hat z_i}$, $\fn{f_i}{S_i}{[\beta]_0}$ is the function that maps $\hat z_i$ to $0$, $P_i = \set{\set{\hat z_i}}$ and $R_i = \set{\hat z_i}$ for each $i \in [n]$.  
\end{lemma}

\begin{proof}
Assume first that $\pi$ is a $((\hat z_i)_{i \in [n]}; \beta)$-annotated  allocation and $v_i(\pi(j)) = \hat w_{ij}$ for every $i, j \in [n]$. We will show that $\pi$ satisfies each of the conditions in the definition of $(\hat t, \hat \eta)$-valid allocation. Recall that $X_{\hat t} = \set{\hat z_1, \hat z_2,\ldots, \hat z_n}$ and $G_{\hat t} = G$. Consider $i \in [n]$. 
First, $\pi(i) \cap X_{\hat t} = \pi(i) \cap \set{\hat z_1, \hat z_2,\ldots, \hat z_n} = \set{\hat z_i} = S_i$. Thus $\pi$ satisfies condition~\ref{inter}.  
Since $\pi$ is $((\hat z_i)_{i \in [n]}; \beta)$-annotated, $\dist_{G[\pi(i)]}(\hat z_i, z) \leq \beta$ for every $z \in \pi(i)$. 
Now, fix a BFS tree $H_i$ of $G[\pi(i)]$, rooted at $\hat z_i$. 
Since $\pi(i) \cap X_t = \set{\hat z_i}$ and $\ca{P}_i = \set{\set{\hat z_i}}$, $H_i$ trivially satisfies condition~\ref{span}\ref{comp}. Again, since $S_i = \set{\hat z_i}$, $f_i(\hat z_i) = 0$ and $\dist_{H_i}(\hat z_i, \hat z_i) = 0$, $H_i$ trivially satisfies condition~\ref{span}\ref{dist} as well. Observe now that by the definition of a BFS tree, for every $z \in \pi(i)$, we have $\dist_{H_i}(\hat z_i, z) = \dist_{G[\pi(i)]}(\hat z_i, z) \leq \beta$, which along with the fact that $f_i(\hat z_i) = 0$, implies that $\dist_{H_i}(\hat z_i, z) \leq \beta - f_i(\hat z_i)$. Thus $H_i$ satisfies condition~\ref{span}\ref{beta}. Finally, as $v_{i}(\pi(j)) = \hat w_{ij}$, $\pi$ satisfies condition~\ref{value}. We have thus shown that $\pi$ is $(\hat t, \hat \eta)$-valid. 

Conversely, assume that $\pi$ is $(\hat t, \hat \eta)$-valid. Consider $i \in [n]$. Recall that $S_i = \set{\hat z_i}$ and $f_i(\hat z_i) = 0$. By condition~\ref{inter}, we have $\hat z_i \in \pi(i)$. By the definition of a $(\hat t, \hat \eta)$-valid allocation, $\pi$ has a $(t, \eta, i)$-witness, say $H'_i$. By condition~\ref{span}\ref{beta}, for every $z \in \pi(i) \setminus \set{\hat z_i}$, we have $\dist_{H'_i}(\hat z_i, z) \leq \beta$, which implies that $\dist_{G[\pi(i)]}(\hat z_i, z) \leq \beta$. Thus $((\hat z_i)_{i \in [n]}; \beta)$-annotated.  
\end{proof}

\paragraph*{Definition of the states of the DP.}  Lemma~\ref{lem:compatible} tells us that to check if $(G, N, \ca{V})$ admits a $((\hat z_i)_{i \in [n]}; \beta)$-annotated  allocation, it is sufficient to check if $(G, N, \ca{V})$ admits a $(\hat t, \hat \eta)$-valid allocation. In light of this, we now define the states of our DP as follows. {\bf For each $\bm{t \in V(T)}$ and each tuple $\bm{\eta \in \ca{C}_t}$, we define 
$\bm{A[t, \eta] = 1}$  if there exists a $\bm{(t, \eta)}$-valid allocation  and $\bm{A[t, \eta] = 0}$ otherwise.} 

\subsubsection*{Implications of the DP}

Assuming the correctness of Lemma~\ref{lem:treedp} and assuming that we can correctly compute $A[t, \eta]$ for every $t \in V(T)$ and $\eta \in \ca{C}_t$, let us now complete the proof of Theorem~\ref{thm:treedp}. 
\begin{proof}[Proof of Theorem~\ref{thm:treedp}]
We show how we can use Lemma~\ref{lem:treedp} as well as the definition of $A[t, \eta]$ to solve  \xanno.  Consider an instance $(G, [n], \ca{V}, (\hat z_i)_{i \in [n]}, \beta)$ of \xanno. 

First, proportionality. For $i \in [n]$, let $\maxval_i =  v_i(V(G))$. For every $n^2$-tuple $(\hat w_{ij})_{i, j \in [n]}$ such that $\hat w_{ii} \geq (1/ n) \cdot \maxval_i$, we invoke the algorithm of Lemma~\ref{lem:treedp} on the instance $(G, [n], \ca{V}, (\hat z_i)_{i \in [n]}, \beta, (\hat w_{ij})_{i, j \in [n^2]})$. If the algorithm returns yes for at least one tuple $(\hat w_{ij})_{i, j \in [n]}$, then we return that $(G, [n], \ca{V}, (\hat z_i)_{i \in [n]}, \beta)$ admits a proportional, $(\hat z_i)_{i \in [n]}$-annotated allocation. Otherwise, we return no. 

Now, envy-freeness. We go over all possible tuples $(w_{ij})_{i, j \in [n]}$ where for every $i \in [n]$, $\hat w_{ii} \geq \hat w_{ij}$ for every $j \in [n] \setminus \set{i}$. For each such tuple $(w_{ij})_{i, j \in [n]}$, we invoke the algorithm of Lemma~\ref{lem:treedp} and return yes or no accordingly. 

Notice that for both proportionality and envy-freeness, the number of tuples $(\hat w_{ij})_{i, j \in [n]}$ that we considered is always upper bounded by $\maxval^{\cO(n^2)}$. 

To solve the maximin fair variant of the problem, let us first see how we can compute the maximin share guarantee for each agent. Recall that $\Gamma(((\hat z_i)_{i \in [n]}; \beta) \mh \mathtt{anno})$ is the class of all $((\hat z_i)_{i \in [n]}; \beta)$-annotated allocations; for convince, let us denote this set of allocations simply by $\Gamma$, and for each agent $p \in [n]$, let us denote the maximin share guarantee of $p$ by $\mms_p$.  Recall that $\mms_p = \max_{\pi \in \Gamma} \min_{q \in [n]} v_p(\pi(q))$. But by Lemma~\ref{lem:compatible}, for any allocation $\pi$, we have $\pi \in \Gamma$ if and only if $\pi$ is $(\hat t, \eta)$-valid for an appropriately defined $\eta \in \ca{C}_{\hat t}$. We thus have
\[
\mms_p = \max_{(w_{ij})_{i, j \in [n]}} \min_{q \in [n]} \{w_{pq} ~|~ \text{there exists } \eta \in \ca{C}_{\hat t} \text{ with }\eta = ((S_i, f_i, (\ca{P}_i, R_i))_{i \in [n]}, (w_{ij})_{i, j \in [n]}) \text{ and } A[\hat t, \eta] = 1\}, 
\] 
where for each $i \in [n]$, $S_i = \set{\hat z_i}$, $\fn{f_i}{S_i}{[\beta]_0}$ is the function that maps $\hat z_i$ to $0$, $P_i = \set{\set{\hat z_i}}$ and $R_i = \set{\hat z_i}$. We can thus compute $\mms_p$. Now to check if $(G, [n], \ca{V}, (\hat z_i)_{i \in [n]}, \beta)$ admits a maximin fair allocation, we go over all possible tuples $(\hat w_{ij})_{i,j \in [n]}$ such that $\hat w_{pp} \geq \mms_p$ for every $p \in [n]$; and invoke the algorithm of Lemma~\ref{lem:treedp}. Again, the number of tuples $(\hat w_{ij})_{i, j \in [n]}$ that we have to consider is always upper bounded by $\maxval^{\cO(n^2)}$. Therefore, we can indeed solve each of the three problems in the claimed runtime, assuming we can compute $A[t, \eta]$ in such time for every $t$ and $\eta$.  
\end{proof}

\paragraph*{Welfare-maximisation.} We can use a similar approach to decide if there is a welfare-maximising $((\hat z_i)_{i \in [n]}; \beta)$-annotated allocation as well. Let $\phi$ be a complete allocation such that for every item $z \in V(G)$, if $z \in \phi(i)$ for some agent $i$, then $v_i(z) \geq v_{i'}(z)$ for every $i' \in [n]$. Notice that since the valuations are additive, $\phi$ achieves the maximum maximum utilitarian social welfare. Now, to decide if there exists a welfare-maximising $((\hat z_i)_{i \in [n]}; \beta)$-annotated allocation, we go over all tuples $(\hat w_{ij})_{i, j \in [n]}$ such that $\sum_{i = 1}^{n} w_{ii} \geq \sum_{i = 1}^n \phi(i)$; and for each such tuple we invoke the algorithm of Lemma~\ref{lem:treedp}. 
 
\subsubsection*{Computing $\bm{A[t, \eta]}$}

We now show how to compute $A[t, \eta]$ recursively. For each $t \in V(T)$, we compute $A[t, \eta]$, assuming that we have computed $A[t', \eta']$ correctly for all descendants $t'$ of $t$ and all $\eta' \in \ca{C}_{t'}$. The base case of the recursion corresponds to the case when $t$ is a leaf node. In all the cases below, $\eta = ((S_i, f_i, (\ca{P}_i, R_i))_{i \in [n]}, (w_{ij})_{i, j \in [n]})$. 

\subsubsection*{{Case 1: leaf node.}} Let $t$ be a leaf node. Then $X_t = \set{\hat z_1, \hat z_2, \ldots, \hat z_n}$ and $G_t$ is the graph containing the $n$ isolated vertices, $\hat z_1, \hat z_2, \ldots, \hat z_n$. By the definition of $\ca{C}_t$, for $i \in [n]$, we have $S_i = \set{\hat z_i}$, $\ca{P}_i = \set{\{\hat z_i\}}$, and $R_i = \set{z_i}$. Notice that the allocation $\pi$ that allocates $\hat z_i$ to $i$ for each $i \in [n]$ trivially satisfies~\ref{inter} and \ref{span}; and $\pi$ satisfies~\ref{value} if and only if $v_i(\hat z_j)= w_{ij}$ for $i, j\in [n]$. Accordingly, we set 
\[
A[t, \eta] = \begin{cases} 1, & \text{ if }v_{i}(\hat z_j) = w_{ij} \text{ for every } i, j \in [n], \\ 0, & \text{ otherwise.} \end{cases} \tag{\textsl{Eq. Leaf}}\label{eq:leaf}
\]

\subsubsection*{Case 2: forget node.} Let $t$ be a forget node and $t'$ the unique child of $t$. 
Let $X_t = X_{t'} \setminus \set{z}$ for some $z \in X_{t'}$. 
Then $z \notin \set{\hat z_1, \hat z_2,\ldots, \hat z_n}$ and $G_t = G_{t'}$. 

Suppose that we have a $(t, \eta)$-valid allocation, say $\pi$. 
Notice that $\pi$ may or may not allocate the vertex $z$. 
If $\pi$ does not allocate $z$, then $\pi$ is $(t', \eta)$-valid as well (and we can simply ignore $z$). 
On the other hand, if $\pi$ does allocate $z$ to agent $i$ for some $i \in [n]$, then notice that $\pi(i) \cap X_{t'} = S_i \cup \set{z}$, and $z$ must be present in one of the blocks of $\ca{P}'_i$, where $\ca{P}'_i$ is a partition obtained from $\ca{P}_i$ by adding $z$ to one of the blocks of $\mathcal{P}_i$. 
We formalise this intuition as follows. 

\begin{sloppypar}
For $i \in [n]$, let $\ca{C}_{t'}(i, z, \eta) \subseteq \ca{C}_{t'}$ be set of all tuples $((S'_p, f'_p, (\ca{P}'_p, R'_p))_{p \in [n]}, (w'_{pq})_{p, q \in [n]}) \in \ca{C}_{t'}$ such that (i) $S'_i = S_i \cup \set{z}$, (ii) $\ca{P}'_i \in \add(z, \ca{P}_i)$ (iii) $R'_i = R_i$, (iv) $f'_i \in \ext(f, \set{z}, [\beta])$, (v) for $j \in [n] \setminus \set{i}$, $(S'_j, f'_j, (\ca{P}'_j, R'_j)) = (S_j, f_j, (\ca{P}_j, R_j))$ and (vi) $w'_{pq} = w_{pq}$ for every $p, q \in [n]$. And we set
\begin{equation}
A[t, \eta]  = A[t', \eta] \lor \bigvee_{\substack{(i, \eta') \\ i \in [n] \\ \eta' \in \ca{C}_{t'}(i, z, \eta)}} A[t', \eta']. \tag{\textsl{Eq. Forget}}\label{eq:forget}
\end{equation}
\end{sloppypar}

\subsubsection*{Case 3: introduce vertex node.} Let $t$ be an introduce vertex node and $t'$ the unique child of $t$. 
Let $X_t = X_{t'} \cup \set{z}$ for some $z \notin X_{t'}$. 
Then $z \notin \set{\hat z_1, \hat z_2, \ldots, \hat z_n}$ and $G_t$ is the disjoint union of $G_{t'}$ and the vertex $z$. That is, $z$ is an isolated vertex in $G_t$. 

Suppose that we have a $(t, \eta)$-valid allocation, say $\pi$. 
Again, $\pi$ may or may not allocate the vertex $z$. 
If $\pi$ does not allocate $z$, then $\pi$ is $(t', \eta)$-valid as well (and we can simply ignore $z$). 
On the other hand, if $\pi$ does allocate $z$ to agent $i$ for some $i \in [n]$, then we must have $z \in S_i = \pi(i) \cap X_t$ and hence $\pi(i) \cap X_{t'} = S_i \setminus \set{z}$. 
Also, as $z$ is an isolated vertex in the graph $G_t$, $\set{z}$ must be a block of $\ca{P}_i$ (and consequently $z \in R_i$). 
In addition, because the valuations are additive, the contribution of $z$ to $v_j(\pi(i))$ is exactly $v_j(z)$ for $j \in [n]$. With this discussion in mind, we now formally compute $A[t, \eta]$ as follows. 
\begin{sloppypar}
For $i \in [n]$ such that $z \in S_i$ and $\set{\set{z_i}} \in \ca{P}_i$, let $\eta_i = ((S_{i, p}, f_{i, p}, (\ca{P}_{i, p}, R_{i, p}))_{p \in [n]}, (w_{i, pq})_{p, q \in [n]}) \in \ca{C}_{t'}$, where (i) $S_{i, i} = S_i \setminus \set{z}$, (ii) $\ca{P}_{i, i} = \ca{P}_i \setminus \set{\set{z}}$ (iii) $R_{i, i} = R_i \setminus \set{z}$, (iv) $f_{i, i} = f_i \big|_{S_i \setminus \set{z}}$, (v) for $j \in [n] \setminus \set{i}$, $(S_{i, j}, f_{i, j}, (\ca{P}_{i, j}, R_{i, j})) = (S_j, f_j, (\ca{P}_j, R_j))$ and (vi) for $p, q \in [n]$, $w_{i, pq} = w_{pq} - v_p(z)$ if $q = i$ and $w_{i, pq} = w_{pq}$ if $q \neq i$. 
And we set
\[
A[t, \eta]  = \begin{cases}A[t', \eta],  & \text{ if    } z \notin \bigcup_{p \in [n]}S_p, \\ A[t', \eta_i], & \text{ if } z \in S_i \text{ and } \set{z} \in \ca{P}_i, \text{ for some } i \in [n], \\ 0, & \text{ otherwise.} \end{cases} \tag{\textsl{Eq. Intro-v}}\label{eq:intver}
\]  
\end{sloppypar}

\subsubsection*{Case 4: introduce edge node.} Let $t$ be an introduce edge node and $t'$ the unique child of $t$. 
Then $X_t = X_{t'}$. Let $t$ introduce the edge $z z' \in E(G)$ for some $z, z' \in X_t$. That is, $G_t$ is the graph $G_{t'}$ with the edge $zz'$ added. 

Suppose that we have a $(t, \eta)$-valid allocation, say $\pi$. 
As in the previous cases, $\pi$ may or may not allocate either of the vertices $z$ and $z'$. 
If $\pi$ does not allocate either $z$ or $z'$, then $\pi$ is $(t', \eta)$-valid. 
If $\pi$ allocates $z$ to agent $i \in [n]$ and $z'$ to agent $j \in [n] \setminus \set{i}$, then again, $\pi$ is $(t', \eta)$-valid. 
So suppose that $\pi$ allocates both $z$ and $z'$ to agent $i$, and consider a spanning forest $H_i$ of $G_t[\pi(i)]$. 
If $H_i$ that does not contain the edge $zz'$, then $\pi$ is still $(t', \eta)$-valid and in particular, $H_i$ is a $(t', \eta, i)$-witness for $\pi$. 
The only other possible scenario is when $H_i$ does contain the edge $zz'$. 
In this case, notice that $z$ and $z'$ must be in the same block of $\ca{P}_i$ and $\dist_{H_i}(z, z') = 1$. Then, by condition~\ref{span}\ref{dist}, we must have $| f_i(z) -f_i(z') | = 1$. 
And deleting the edge $zz'$ from $H_i$ would split the connected component of $H_i$ that contains $z$ and $z'$ into exactly two connected components---one containing $z$ and the other containing $z'$, with exactly one of these two components containing the root of $\blk_{\ca{P}_i}(z)$ as well; and the connected components of $H_i$ that do not contain $z$ (and hence $z'$) remain unchanged. 
Intuitively, we try to guess all the possible ways in which the connected component of $H_i$ that contains $z$ and $z'$ be split. 
\begin{sloppypar}
For $i \in [n]$, we say that $\eta$ is $(i, zz')$-consitent if $z, z' \in S_i$, $\blk_{\ca{P}_i}(z) = \blk_{\ca{P}_i}(z')$ and $| f_i(z) - f_i(z') | = 1$. Suppose $\eta$ is $(i, zz')$-consistent for some $i \in [n]$. Let $P \in \ca{P}_i$ be such that $z, z' \in P$ and let $x \in R_i \cap P$. That is, $P = \blk_{\ca{P}_i}(z) = \blk_{\ca{P}_i}(z')$ and $x$ is the root of $\blk_{\ca{P}_i}(z)$. Let $u, u' \in \set{z, z'}$ be such that $f_i(u') = 1 + f_i(u)$. Now, let $\ca{C}_{t'}(i, zz', \eta) \subseteq C_{t'}$ be the set of all tuples $((S'_p, f'_p, (\ca{P}'_p, R'_p))_{p \in [n]}, (w'_{pq})_{p, q \in [n]}) \in \ca{C}_{t'}$ such that (i) $S'_i = S_i$, (ii) $\ca{P}'_i = (\ca{P}_i \setminus \set{P}) \cup \set{P_u, P_{u'}}$ for some partition $\set{P_u, P_{u'}}$ of $P$ such that $x, u \in P_u$ and $u' \in  P'$ (for example, if $u = z$ and $u' = z'$, then we have $x, z \in P_u$ and $z' \in P_{u'}$), (iii) $R'_i = R_i \cup \set{u'}$ (for example, if $u = z$, then we have $R'_i = R_i \cup \set{z'}$), (iv) $f'_i = f_i$, (v) for $j \in [n] \setminus \set{i}$, $(S'_j, f'_j, (\ca{P}'_j, R'_j)) = (S_j, f_j, (\ca{P}_j, R_j))$ and (vi) for $p, q \in [n]$, $w'_{pq} = w_{pq}$. And we set
\[
A[t, \eta]  = \begin{cases} A[t', \eta] \lor \bigvee_{\eta' \in \ca{C}_{t'}(i, zz', \eta)} A[t', \eta'],  & \text{ if  } \eta \text{ is } (i, zz')\text{-consistent for } i \in [n], \\  A[t', \eta],  & \text{ otherwise.} \end{cases} \tag{\textsl{Eq. Intro-e}}\label{eq:inted}
\]  
\end{sloppypar}

\subsubsection*{Case 5: join node.} Let $t$ be a join node with children $t'$ and $t''$. Then $X_t = X_{t'} = X_{t''}$ and the graph $G_t$ is the ``union'' of the two graphs $G_{t'}$ and $G_{t''}$ with $V(G_t) = V(G_{t'}) \cup V(G_{t''})$, $E(G_t) = E(G_{t'}) \cup E(G_{t''})$, $V(G_{t'}) \cap V(G_{t''}) = X_t$ and $E(G_{t'}) \cap E(G_{t''}) = \emptyset$. 

Suppose that we have a $(t, \eta)$-valid allocation, say $\pi$. Then $\pi$ is essentially a ``union'' of two allocations, $\pi'$ and $\pi''$, where $\pi'$ is $(t', \eta')$-valid and $\pi''$ is $(t'', \eta'')$-valid for some $\eta' \in \ca{C}_{t'}$ and $\eta'' \in \ca{C}_{t''}$.  
We guess all possible choices for $\eta'$ and $\eta''$. Notice that while combining $\pi'$ and $\pi''$ to form $\pi$ , we must also ensure that we count agent $i$'s valuation for $\pi(j) \cap X_t = S_j$ (i.e., $v_i(S_j)$) only once. We formalise this intuition as follows. 

For $\eta' = ((S'_i, f'_i, (\ca{P}'_i, R'_i))_{i \in [n]}, (w'_{ij})_{i, j \in [n]}) \in \ca{C}_{t'}$ and $\eta'' = ((S''_i, f''_i, (\ca{P}''_i, R''_i))_{i \in [n]}, (w''_{ij})_{i, j \in [n]}) \in \ca{C}_{t''}$, we say that $\eta$ is $(\eta', \eta'')$-consistent if for each $i \in [n]$, $S'_i = S''_i = S_i$, $f'_i = f''_i = f_i$ and $(\ca{P}_i, R_i)$ is an acyclic join of $(\ca{P}'_i, R'_i)$ and $(\ca{P}''_i, R''_i)$, and for $i, j \in [n]$, $w_{ij} = w'_{ij} + w''_{ij} - v_i(S_j)$. Let $\ca{C}_{(t, t'')}(\eta) \subseteq \ca{C}_{t'} \times \ca{C}_{t''}$ be the set of all $(\eta', \eta'') \in  \ca{C}_{t'} \times \ca{C}_{t''}$ such that $\eta$ is $(\eta', \eta'')$-consistent. And we set
 \[
A[t, \eta] = \bigvee_{(\eta', \eta'') \in \ca{C}_{(t, t'')}(\eta)} (A[t', \eta'] \land  A[t'', \eta'']). \tag{\textsl{Eq. Join}}\label{eq:join}
\]
This completes the description of the DP. We relegate the formal proof of correctness of our computation to the appendix (see Lemma~\ref{lem:treedp-correctness}). To complete the proof of Lemma~\ref{lem:treedp}, let us analyse the runtime of our DP. By Observation~\ref{obs:Ctsize}, we have $\card{\ca{C}_t} \leq (\tw + n)^{\cO(\tw + n)} \cdot \beta^{\cO(\tw)} \cdot \maxval^{\cO(n^2)}$. And since we work with a nice tree decomposition, $\card{V(T)} = m^{\cO(1)}$. Hence the number of states of the DP is at most $(\tw + n)^{\cO(\tw + n)} \cdot \beta^{\cO(\tw)} \cdot \maxval^{\cO(n^2)} \cdot m^{\cO(1)}$. And notice that to compute each of the entries, we only need to look up $(\tw + n)^{\cO(\tw + n)} \cdot \beta^{\cO(\tw)} \cdot \maxval^{\cO(n^2)}$ many previously computed entries. Hence the runtime follows. This completes the proof of Lemma~\ref{lem:treedp}.

\begin{toappendix}
\begin{lemma}\label{lem:treedp-correctness}
For each $t \in V(T)$ and $\eta \in \ca{C}_t$, $A[t, \eta]$ is computed correctly. 
\end{lemma}

\begin{proof}
\begin{sloppypar}
Consider $t \in V(T)$ and $\eta \in \ca{C}_t$. Let $\eta = (((S_i, f_i, (\ca{P}_i, R_i))_{i \in [n]}, w_{ij})_{i, j \in [n]})$. 

We prove the lemma by induction on $t$. That is, assuming that $A[t', \eta']$ is computed correctly for all descendants $t'$ of $t$ and $\eta' \in \ca{C}_{t}'$, we prove that $A[t, \eta]$ is computed correctly as well. The base case of the induction is when $t$ is a leaf node. 

\paragraph*{Base case.} Suppose that $t$ is a leaf node. Then $X_t = \set{\hat z_1, \hat z_2, \ldots, \hat z_n}$ and $G_t$ is the edgeless graph containing the $n$ vertices $\hat z_1, \hat z_2, \ldots, \hat z_n$. Then by the definition of $C_{t}$, for $i \in [n]$, the only choices for $S_i, f_i$ and $(\ca{P}_i, R_i)$ are  as follows: $S_i = \set{\hat z_i}$, $f_i(\hat z_i) = 0$, $\ca{P}_i = \set{\{\hat z_i\}}$, and $R_i = \set{z_i}$. Notice then that an allocation $\fn{\pi}{[n]}{2^{V(G_t)}}$ is $(t, \eta)$-valid if and only if $\pi(i) = \set{\hat z_i}$ for each $i \in [n]$ and $v_i(\hat z_j)= w_{ij}$ for $i, j\in [n]$. That is, $A[t, \eta] = 1$ if and only if $v_i(\hat z_j)= w_{ij}$ for $i, j\in [n]$, and this is precisely how \ref{eq:leaf} computes $A[t, \eta]$. 

\paragraph*{Induction Hypothesis.} Assume now that $t$ is a non-leaf node and that we have computed $A[t', \eta']$ correctly for all descendants $t'$ of $t$ and $\eta' \in \ca{C}_{t}'$. 

Now, to prove that $A[t, \eta]$ is computed correctly, we must prove that $A[t, \eta] = 1$ if and only if the right hand side of the recurrence formula for $A[t, \eta]$ (i.e., \ref{eq:forget}, \ref{eq:intver}, \ref{eq:inted} or \ref{eq:join}, depending on which kind of a node $t$ is) evaluates to $1$. In each of the following cases, whenever we assume that $A[t, \eta] = 1$, we use $\fn{\pi}{[n]}{2^{V(G_t)}}$ to denote a $(t, \eta)$-valid allocation, and for each $i \in [n]$, we use $H_i$ to denote a spanning forest of $G_t[\pi(i)]$ that satisfies conditions~\ref{span}\ref{comp}--\ref{span}\ref{beta}, i.e., a $(t, \eta, i)$-witness for $\pi$. We now split the proof into cases depending on what kind of a node $t$ is. 

\paragraph*{Forget node.} Suppose that $t$ is a forget node and $t'$ the unique child of $t$. Let $t$ forget the vertex $z$. That is, $X_t = X_{t'} \setminus \set{z}$ for some $z \in X_{t'}$. Then $z \notin \set{\hat z_1, \hat z_2, \ldots, \hat z_n}$ and $G_t = G_{t'}$.

Recall that for $i \in [n]$, we defined $\ca{C}_{t'}(i, z, \eta) \subseteq \ca{C}_{t'}$ to be the set of all tuples $((S'_p, f'_p, (\ca{P}'_p, R'_p))_{p \in [n]}, (w'_{pq})_{p, q \in [n]}) \in \ca{C}_{t'}$ such that (i) $S'_i = S_i \cup \set{z}$, (ii) $\ca{P}'_i \in \add(z, \ca{P}_i)$ (iii) $R'_i = R_i$, (iv) $f'_i \in \ext(f, \set{z}, [\beta])$, (v) for $j \in [n] \setminus \set{i}$, $(S'_j, f'_j, (\ca{P}'_j, R'_j)) = (S_j, f_j, (\ca{P}_j, R_j))$ and (vi) $w'_{pq} = w_{pq}$ for every $p, q \in [n]$. And we have set
\begin{equation}
A[t, \eta]  = A[t', \eta] \lor \bigvee_{\substack{(i, \eta') \\ i \in [n] \\ \eta' \in \ca{C}_{t'}(i, z, \eta)}} A[t', \eta']. \tag{\textsl{Eq. Forget}}\label{eq:forget}
\end{equation}

\paragraph*{Forget node-Forward direction.} Assume that $A[t, \eta] = 1$.  
Since $X_t \subseteq X_{t'}$, notice that $\eta \in C_{t}'$ as well. There are two possible cases: either $\pi$ does not allocate $z$ or $\pi$ does allocate $z$. 

Suppose that $\pi$ does not allocate $z$. That is, $z \notin \bigcup_{i \in [n]}\pi(i)$. Then $\pi$ is $(t', \eta)$-valid as well. To see this, observe that for each $i \in [n]$,  $\pi(i) \cap X_{t'} = \pi(i) \cap X_{t} = S_i$ and $G_{t'}[\pi(i)] = G_{t}[\pi(i)]$. Hence $\pi$ and $H_i$ satisfy the all the properties \ref{inter}-\ref{value}. Since $\pi$ is $(t', \eta)$-valid, by the definition of $A[t', \eta']$, we have $A[t', \eta] = 1$. And by induction hypothesis, we computed $A[t', \eta]$ correctly, which implies that the right hand side of \ref{eq:forget} evaluates to $1$. 

Suppose now that $\pi$ allocates $z$. Assume without loss of generality that $\pi$ allocates $z$ to agent $1$. That is, $z \in \pi(1)$.  (The other cases are symmetric.) We define a tuple $\eta' = ((S'_p, f'_p, (\ca{P}'_p, R'_p))_{p \in [n]}, (w'_{pq})_{p, q \in [n]}) \in \ca{C}_{t'}$ as follows. (i) We define $S'_1 = \pi(1) \cap X_{t'} = \pi(1) \cap (X_t \cup \set{z}) = S_1 \cup \set{z}$. (ii) We define $\ca{P}'_1$ to be the partition of $S'_1$ into connected components of $H_1$. Observe that the set $\set{z}$ is not a block of $\ca{P}'_1$. To see this, notice that since $\pi$ is $(t, \eta)$-valid and since $z \in \pi(1) \setminus X_t$, by \ref{span}\ref{beta}, there exists $y \in R_1$ such that $\dist_{H_1}(y, z) \leq \beta - f_1 (y)$, which implies that $y$ and $z$ are in the same connected component of $H_i$. And since $y \in X_t$, we have $y \in X_{t'}$ as well. Thus $\blk_{\ca{P}'_1} (y) = \blk_{\ca{P}'_1} (z)$, and in particular, $\ca{P}'_1 \in \add(z, \ca{P}_1)$. (iii) We define $R'_1 = R_1$ (and hence $z \notin R'_1$). (iv) The function $\fn{f'_1}{S'_1}{[\beta]_0}$ is defined such that $f'_1 \big|_{S_1} = f_1$ and $f'_1(z) = f'_1(y) + \dist_{H_1}(y, z)$, (where $y \in R'_1$ is the root of $\blk_{\ca{P}'_1}$). Notice  that as $y \in S_1$, $f'_1 (y) = f_1 (y)$ and since $\dist_{H_1}(y, z) \leq \beta - f_1 (y)$, we indeed have $f'_1 (z) \in [\beta]$. In particular, $f'_1 \in \ext(f_1, \set{z}, [\beta]_0)$. (vi) We define $(S'_j, f'_j, (\ca{P}'_j, R'_j)) = (S_j, f_j, (\ca{P}_j, R_j)$ for every $j \in [n] \setminus \set{1}$. And finally, we define $w'_{pq} = w_{pq}$ for every $p, q \in [n]$. Observe that $\eta' \in C_{t'}(1, z, \eta)$. Observe also that our choice of $\eta'$ is such that $\pi$ is $(t', \eta')$-valid, which implies that $A[t', \eta'] = 1$. Now, by induction hypothesis, we computed $A[t', \eta']$ correctly, which, along with the fact that $\eta' \in C_{t'}(1, z, \eta)$, implies that the right hand side of \ref{eq:forget} evaluates to $1$. 

\paragraph*{Forget node-Backward direction.} Assume that the right hand side of ~\ref{eq:forget} evaluates to $1$. Then either $A[t', \eta] = 1$ or $A[t', \eta'] = 1$ for some $i \in [n]$ and $\eta' \in \ca{C}_{t'}(i, z, \eta)$. In either case, we will prove that $A[t, \eta] = 1$, which will imply that ~\ref{eq:forget} computes $A[t, \eta]$ correctly. And to prove that $A[t, \eta] = 1$, we will prove that there exists a $(t, \eta)$-valid allocation. 

Suppose first that $A[t', \eta] = 1$. By the induction hypothesis, we computed $A[t', \eta]$ correctly, and hence there exists a $(t', \eta)$-valid allocation $\fn{\phi}{[n]}{2^{V(G_{t'})}}$. We claim that $\phi$ is a $(t, \eta)$-valid allocation as well. 
To see this, observe that $\phi$ satisfies all the conditions in the definition of a $(t, \eta)$-valid allocation. In particular, since $\phi$ is $(t', \eta)$-valid, for $i \in [n]$, we have $\phi(i) \cap X_{t'} = S_i$, which implies that $z \notin \phi(i)$. And the facts that $z \notin \phi(i)$ and $X_t = X_{t'} \setminus \set{z}$ together imply that $\phi(i) \cap X_{t} = S_i$. Thus condition~\ref{inter} is satisfied. Now, for $i \in [n]$, let $H_{\phi, i}$ be a $(t', \eta, i)$-witness for $\phi$. Then $H_{\phi, i}$ is a $(t, \eta, i)$-witness for $\phi$ as well, as it trivially satisfies the required conditions~\ref{span}\ref{comp}--\ref{span}\ref{beta}. In particular, as $z \notin \phi(i)$ and $X_t = X_{t'} \setminus \set{z}$, for every $z' \in V(G_t)$, we have $z' \in \phi(i) \setminus \set X_{t}$ if and only if $z' \in \phi(i) \setminus X_{t'}$; and  hence $H_{\phi, i}$ satisfies condition \ref{span}\ref{beta}. Finally, condition~\ref{value} is trivially satisfied as well.  

Suppose now that $A[t', \eta'] = 1$ for some $i \in [n]$ and $\eta' \in \ca{C}_{t'}(i, z, \eta)$. Again, by the induction hypothesis, we computed $A[t', \eta']$ correctly, and hence there exists a $(t', \eta')$-valid allocation $\fn{\pi'}{[n]}{2^{V(G_{t'})}}$. For $i \in [n]$, let $H'_i$ be a $(t', \eta', i)$-witness for $\pi'$. First, assume without loss of generality that $\eta' \in \ca{C}_{t}(1, z, \eta)$. Let $\eta' = ((S'_p, f'_p, (\ca{P}'_p, R'_p))_{p \in [n]}, (w'_{pq})_{p, q \in [n]})$. Then, by the definition of $\ca{C}_{t}(1, z, \eta)$, we have (i) $S'_1 = S_1 \cup \set{z}$, (ii) $\ca{P}'_1 \in \add(z, \ca{P}_1)$ (iii) $R'_1 = R_1$, (iv) $f'_1 \in \ext(f, \set{z}, [\beta])$, (v) $(S'_p, f'_p, (\ca{P}'_p, R'_p)) = (S_p, f_p, (\ca{P}_p, R_p))$ for every $p \in [n] \setminus \set{1}$ and (vi) $w'_{pq} = w_{pq}$ for every $p, q \in [n]$.  We claim that for each $i \in [n]$, $H'_i$ is a $(t, \eta, i)$-witness for $\pi'$, and that $\pi'$ is a $(t, \eta)$-valid allocation. And to prove this, we only need to verify that (a) $\pi(1) \cap X_t = S_1$ and (b) there exists $y \in R_1$ such that $\dist_{H'_1}(y, z) \leq \beta - f_1 (y)$, as the rest of the conditions in the definition of a $(t, \eta)$-valid allocation follow directly from the fact that $\pi'$ is $(t', \eta')$-valid. Now, since $\pi'$ is $(t', \eta')$-valid, we have $\pi'(1) \cap X_{t'} = S'_1 = S_1 \cup \set{z}$. In particular, we have $z \in \pi'(1) \setminus X_t$.  Since $X_t = X_{t'} \setminus \set{z}$, we have$\pi'(1) \cap X_t = S'_1 \setminus \set{z} = S_1$. 
Observe now that since $\ca{P}'_1 \in \add(z, \ca{P}_1)$, the set $\set{z}$ is not a block of $\ca{P}'_1$. And since $R'_1 = R_1 \subseteq X_t$, we have $z \notin R'_1$. Let $y \in R_1$ be the root of $\blk_{\ca{P}'_1}(z)$. Then $y \in X_t$ and hence $f'_1 (y) = f_1 (y)$. Therefore, by condition~\ref{span}\ref{beta}, we have $\dist_{H'_1}(y, z) \leq \beta - f'_1 (y) = \beta - f_1 (y)$. Thus $\pi'$ is a $(t, \eta)$-valid allocation. 

\paragraph*{Introduce vertex node.} Suppose that $t$ is an introduce vertex node and $t'$ the unique child of $t$. Let $t$ introduce the vertex $z$. That is, $X_t = X_{t'} \cup \set{z}$ for some $z \notin X_{t'}$. Then $z \notin \set{\hat z_1, \hat z_2, \ldots, \hat z_n}$ and the graph $G_t$ is the union of the graph $G_{t'}$ and the isolated vertex $z$. 

\begin{sloppypar}
Recall that for $i \in [n]$ such that $z \in S_i$ and $\set{\set{z_i}} \in \ca{P}_i$, we defined $\eta_i = ((S_{i, p}, f_{i, p}, (\ca{P}_{i, p}, R_{i, p}))_{p \in [n]}, (w_{i, pq})_{p, q \in [n]}) \in \ca{C}_{t'}$, where (i) $S_{i, i} = S_i \setminus \set{z}$, (ii) $\ca{P}_{i, i} = \ca{P}_i \setminus \set{\set{z}}$ (iii) $R_{i, i} = R_i \setminus \set{z}$, (iv) $f_{i, i} = f_i \big|_{S_i \setminus \set{z}}$, (v) for $j \in [n] \setminus \set{i}$, $(S_{i, j}, f_{i, j}, (\ca{P}_{i, j}, R_{i, j})) = (S_j, f_j, (\ca{P}_j, R_j))$ and (vi) for $p, q \in [n]$, $w_{i, pq} = w_{pq} - v_p(z)$ if $q = i$ and $w_{i, pq} = w_{pq}$ if $q \neq i$. 
And we have set
\[
A[t, \eta]  = \begin{cases}A[t', \eta],  & \text{ if    } z \notin \bigcup_{p \in [n]}S_p, \\ A[t', \eta_i], & \text{ if } z \in S_i \text{ and } \set{z} \in \ca{P}_i, \text{ for some } i \in [n], \\ 0, & \text{ otherwise.} \end{cases} \tag{\textsl{Eq. Intro-v}}\label{eq:intver}
\]  
\end{sloppypar}

\paragraph*{Introduce vertex node-Forward direction.} Assume that $A[t, \eta] = 1$. There are two possible cases: either $\pi$ does not allocate $z$ or $\pi$ does allocate $z$. 

Suppose that $\pi$ does not allocate $z$. That is, $z \notin \bigcup_{i = 1}^{n}\pi(i)$. Then, since $z \in X_t$ and since $\pi$ is $(t, \eta)$-valid, we have $z \notin \bigcup_{i = 1}^{n} S_i$. In this case, \ref{eq:intver} sets $A[t, \eta]$ to be $A[t', \eta]$. First, since $z \notin \bigcup_{i = 1}^{n} S_i$, we have $S_1, S_2,\ldots, S_n \subseteq X_{t'}$ and hence $\eta \in C_{t'}$. Since $z \notin \bigcup_{i = 1}^{n} \pi(i)$, and since the graphs $G_t$ and $G_{t'}$ differ only in the isolated vertex $z$, observe that $\pi$ is $(t', \eta)$-valid as well. Hence $A[t', \eta] = 1$. By the induction hypothesis, $A[t', \eta]$ is computed correctly, and we can thus conclude that \ref{eq:intver} correctly assigns the value $1$ to $A[t, \eta]$ in this case. 

Suppose now that $\pi$ does allocate $z$. Assume without loss of generality that $\pi$ allocates $z$ to agent $1$. 
We thus have $z \in \pi(1) \cap X_t = S_1$. Since $z$ is an isolated vertex in $G_t$ and hence in $H_1$, the set $\set{z}$ is a block of $\ca{P}_1$. Therefore, \ref{eq:intver} sets $A[t, \eta]$ to be $A[t', \eta_1]$ in this case. Now, consider the allocation $\fn{\pi_1}{[n]}{V(G_{t'})}$, where $\pi_1(1) = \pi(1) \setminus \set{z}$ and $\pi_1 (j) = \pi(j)$ for $j \in [n] \setminus \set{1}$. We claim that $\pi_1$ is $(t', \eta_1)$-valid. 
To see this, observe the following facts. (1) Since $\pi$ is $(t, \eta)$-valid, $\pi(1) \cap X_t = S_1$, which implies that $\pi_1(1) \cap X_{t'} = S_1 \setminus \set{z} = S^1_1$. 
(2) Since $z$ is an isolated vertex in $H_1$, we can verify that $H_1 - z$ is a $(t', \eta_1, 1)$-witness for $G_{t'}[\pi_1]$. 
(3) For $p \in [n]$, we have $v_p(\pi_1 (1)) = v_p(\pi(1) \setminus \set{z}) = v_p(\pi(1)) - v_p(z) = w_{p1} - v_p(z) = w_{1, p1}$; and $v_p(\pi_1 (q)) = v_p(\pi(q)) = w_{pq} = w_{1, pq}$. 
Thus $\pi_1$ is $(t', \eta_1)$-valid and hence $A[t', \eta_1] = 1$. By the induction hypothesis, we computed $A[t', \eta_1]$ correctly, and hence \ref{eq:intver} correctly assigns the value $1$ to $A[t, \eta]$. 

\paragraph*{Introduce vertex node-Backward direction.} Assume that the right hand side of \ref{eq:intver} evaluates to $1$. Notice that there are only two possible ways for this to happen: either (i) $z \notin \bigcup_{i = 1}^{n} S_i $ and $A[t', \eta] = 1$ or $z \in S_i$, $\set{z} \in \ca{P}_i$ and $A[t', \eta_i] = 1$ for some $i \in [n]$.  

Suppose that $z \notin \bigcup_{i = 1}^{n} S_i$ and $A[t', \eta = 1]$. Then, by the induction hypothesis, there exists a $(t', \eta)$-valid allocation $\fn{\pi'}{[n]}{2^{V(G_{t'})}}$. And we can verify that $\pi'$ is $(t, \eta)$-valid as well, which implies that $A[t, \eta] = 1$. Thus \ref{eq:intver} correctly assigns the value $1$ to $A[t, \eta]$ in this case. 

Suppose that $z \in S_i$, $\set{z} \in \ca{P}_i$ and $A[t', \eta_i] = 1$ for some $i \in [n]$. Assume without loss of generality that $i = 1$. That is, $A[t', \eta_1] = 1$. And by the induction hypothesis, there exists a $(t', \eta_1)$-valid allocation $\fn{\phi_1}{[n]}{2^{V(G_{t'})}}$. 
Consider the allocation $\fn{\phi}{[n]}{2^{V(G_{t})}}$ that extends $\phi_1$ as follows: $\phi(1) = \phi_1 (1) \cup \set{z}$ and $\phi(p) = \phi_1 (p)$ for every $p \in [n] \setminus \set{1}$. 
We can verify that $\phi$ is $(t, \eta)$-valid, which implies that $A[t, \eta] = 1$, and therefore, we can conclude that \ref{eq:intver} correctly assigns the value $1$ to $A[t, \eta]$ in this case. 
To see that $\phi$ is $(t, \eta)$-valid, observe the following facts. (a) We have $\phi(1) \cap X_t = \phi(1) \cap (X_{t'} \cup \set{z}) = (\phi_1(1) \cap X_{t'}) \cup \set{z} = S_{1, 1} \cup \set{z} = S_1$ and $\phi(p) \cap X_t = \phi_1(p) \cap X_{t'} = S_{1, p} = S_2$. 
(b) Consider a $(t', \eta, 1)$-witness $H'_1$ for $\phi_1$. Then the graph obtained from $H'_1$ by adding the isolated vertex $z$ is a $(t, \eta, 1)$-witness for $\phi$. 
And a $(t', \eta, p)$-witness for $\phi_1$ is a $(t, \eta, p)$-witness for $\phi$ as well.  
(c) Finally, for $p \in [n]$, we have $v_p (\phi(1)) = v_p(\phi_1 (1) \cup \set{z}) = v_p (\phi_1 (1)) + v_p (z) = w_{1, p1} + v_p (z) = w_{p1}$ and $v_p (\phi(q)) = v_p(\phi_1 (q)) = w_{1, pq} = w_{pq}$. 

\paragraph*{Introduce edge node.} Suppose that $t$ is an introduce vertex node and $t'$ the unique child of $t$. Let $t$ introduce the edge $zz' \in E(G)$. That is, $X_t = X_{t'}$ and $z, z' \in X_{t}$. for some $z \notin X_{t'}$. Thus the only difference between the graphs $G_{t}$ and $G_{t'}$ is that while $G_{t}$ contains the edge $zz'$, $G_{t'}$ does not. 

Let us recall \ref{eq:inted}. For $i \in [n]$, we say that $\eta$ is $(i, zz')$-consistent if $z, z' \in S_i$, $\blk_{\ca{P}_i}(z) = \blk_{\ca{P}_i}(z')$ and $| f_i(z) - f_i(z') | = 1$. Suppose $\eta$ is $(i, zz)$-consistent for some $i \in [n]$. Let $P \in \ca{P}_i$ be the block of $\ca{P}_i$ that contains $z$ (and $z'$) and $x \in R_i$ the root of $P$.  
Let $u, u' \in \set{z, z'}$ be such that $f_i(u') = 1 + f_i(u)$. We have defined $\ca{C}_{t'}(i, zz', \eta) \subseteq C_{t'}$ to be the set of all tuples $((S'_p, f'_p, (\ca{P}'_p, R'_p))_{p \in [n]}, (w'_{pq})_{p, q \in [n]}) \in \ca{C}_{t'}$ such that (i) $S'_i = S_i$, (ii) $\ca{P}'_i = (\ca{P}_i \setminus \set{P}) \cup \set{P_u, P_{u'}}$ for some partition $\set{P_u, P_{u'}}$ of $P$ such that $x, u \in P_u$ and $u' \in  P'$ (for example, if $u = z$ and $u' = z'$, then we have $x, z \in P_u$ and $z' \in P_{u'}$), (iii) $R'_i = R_i \cup \set{u'}$ (for example, if $u = z$, then we have $R'_i = R_i \cup \set{z'}$), (iv) $f'_i = f_i$, (v) for $j \in [n] \setminus \set{i}$, $(S'_j, f'_j, (\ca{P}'_j, R'_j)) = (S_j, f_j, (\ca{P}_j, R_j))$ and (vi) for $p, q \in [n]$, $w'_{pq} = w_{pq}$. And we have set
\[
A[t, \eta]  = \begin{cases} A[t', \eta] \lor \bigvee_{\eta' \in \ca{C}_{t'}(i, zz', \eta)} A[t', \eta'],  & \text{ if  } \eta \text{ is } (i, zz')\text{-consistent for } i \in [n], \\  A[t', \eta],  & \text{ otherwise.} \end{cases} \tag{\textsl{Eq. Intro-e}}\label{eq:inted}
\]  

\paragraph*{Introduce edge node-Forward direction.} Assume that $A[t, \eta] = 1$. There are two possible cases: either $zz' \in E(H_i)$ for some $i \in [n]$, or (2) $zz' \notin E(H_i)$ for any $i \in [n]$. 

Suppose first that $zz' \in E(H_i)$ for some $i \in [n]$. Without loss of generality, assume that $i = 1$. We will show that $\eta$ is $(1, zz')$-consistent and that $A[t', \eta'] = 1$ for some $\eta' \in C_{t'}(1, zz', \eta)$. 
The fact that $zz' \in E(H_1)$ implies that $z, z' \in \pi(1) \cap X_t = S_1$. It also implies that $z$ and $z'$ are in the same block of $\ca{P}_1$. Let $P \in \ca{P}_1$ be the block of $\ca{P}_1$ that contains $z$ and $z'$, and let $x \in R_i$ be the root of $P$. 
Notice that as $zz' \in E(H_1)$ and $H_1$ is acyclic, either the unique path in $H_1$ from $x$ to $z$ contains $z'$ or the unique path in $H_1$ from $x$ to $z'$ contains $z$; for otherwise $H_1$ would contain a cycle.  Assume without loss of generality that the unique path in $H_1$ from $x$ to $z'$ contains $z$ (and hence the edge $zz'$). Thus, $\dist_{H_1}(x, z') = \dist_{H_1}(x, z) + 1$. And by condition~\ref{span}\ref{dist}, $f_1 (z') = f_1(x) + \dist_{H_1}(x, z') = f_1(x) + \dist_{H_1}(x, z) + 1 = f_1 (z) + 1$. That is, $|f_1(z') - f_1(z)| = 1$. We have thus shown that $\eta$ is $(1, zz')$-consistent. 

Now, consider the graph $H_1 - zz'$.  Notice that as $H_1$ is acyclic, $z$ and $z'$ are in different components of $H_1 - zz'$. Let $H_1^{z}$ be the component of $H_1 - zz'$ that contains $z$, and $H_1^{z'}$ the component that contains $z'$. Notice that $x$ belongs to the component $H_1^z$.  
We now define a tuple $((S'_p, f'_p, (\ca{P}'_p, R'_p))_{p \in [n]}, (w'_{pq})_{p, q \in [n]})$ as follows: 
$S'_1 = S_1, f'_1 = f_1, R'_1 = R_1 \cup \set{z'}$, $\ca{P}'_1 = (\ca{P} \setminus \set{P}) \cup \set{P_{z}, P_{z'}}$, where $\set{P_z, P_{z'}}$ is the partition of $P$ into the connected components of $H_1 - zz'$; $(S'_j, f'_j, (\ca{P}'_j, R'_j) = (S_j, f_j, (\ca{P}_j, R_j)$ for every $j \in [n] \setminus \set{1}$ and $w'_{pq} = w_{pq}$ for every $p, q \in [n]$. Notice that $\eta' \in \ca{C}_{t'}(1, zz;, \eta)$. We can verify that $\pi$ is indeed $(t', \eta')$-valid. Notice that conditions \ref{inter} and \ref{value} trivially hold; and that $H_p$ is a $(t', \eta', p)$- witness for $\pi$ for every $p \in [n] \setminus \set{1}$. So we only need to argue that $H_1 - zz'$ is a $(t', \eta', 1)$-witness for $\pi$. And for this, notice first that condition~\ref{span}\ref{comp} trivially holds as well, as $\ca{P}'_1$, by definition, is precisely the partition of $S'_1$  into the connected components of $H_1 - zz'$. 

To see that condition~\ref{span}\ref{dist} holds, consider $z'' \in S_i$ and let $x'' \in R_i$ be the root of $\blk_{\ca{P}'_1}(x'')$. If the unique path in $H_1$ from $x''$ to $z''$ does not contain the edge $zz'$, then \ref{span}\ref{dist} trivially holds. 
So suppose that the unique path in $H_1$ from $x''$ to $z''$ does contain the edge $zz'$ (and hence $x'' = x$). Notice then that $z'$ and $z''$ are in the same connected component of $H_1 - zz'$, (and because we have $z' \in R'_1$, $z'$ is the root of $\blk_{\ca{P}'_1}(z'')$). 
We also have $f'_1(z'') = f_1 (z'') = \dist_{H_1}(x'', z'') = \dist_{H_1}(x'', z') + \dist_{H_1}(z', z'') = f_1(z') +\dist_{H_1}(z', z'') = f'_1 (z') + \dist_{H_1 - zz'}(z', z'')$. 
Thus condition~\ref{span}\ref{dist} holds as well. 
Finally, consider $z''' \in \pi(1) \setminus X_{t'}$. Since $z''' \in \pi(1) \setminus X_t$, there exists $y \in R_i$ such that $\dist_{H_1}(y, z''') \leq \beta - f_1(y)$. 
If the unique path in $H_1$ from $y$ to $z'''$ does not contain the edge $zz'$, then again, condition~\ref{span}\ref{beta} holds. 
Suppose the path does contain the edge $zz'$. Then $\beta \geq \dist_{H_1}(y, z''') = \dist_{H_1}(y, z') + \dist_{H_1}(z', z''') = f_1 (z') + \dist_{H_1}(z', z''')$. 
But notice that $z'$ and $z'''$ are in the same connected component of $H_1 - zz'$ and therefore $\dist_{H_1 - zz'}(z', z''') = \dist_{H_1} (z', z''')$. And recall that $f'_1 = f_1$. Therefore, $\beta \geq f_1 (z') + \dist_{H_1} (z', z''')$ implies that $\beta \geq f'_1 (z') + \dist_{H_1 - zz'}(z', z''')$, which implies that $\dist_{H_1 - zz'}(z', z''') \leq \beta - f'_1 (z')$. Finally, since $z' \in R'_1$, condition~\ref{span}\ref{beta} holds as well. 
We have thus shown that $\pi$ is $(t', \eta')$-valid and hence $A[t', \eta'] = 1$. By the induction hypothesis, we computed $[t', \eta']$ correctly, and hence \ref{eq:inted} correctly assigns the value $1$ to $A[t, \eta]$ in this case. 

Suppose now that $zz' \notin E(H_i)$ for every $i \in [n]$. Then $H_i$ is a $(t', \eta, i)$-witness for $\pi$ for each $i \in [n]$, and $\pi$ is $(t', \eta)$-valid. Thus $A[t', \eta] = 1$. Again, from the induction hypothesis, we can conclude that \ref{eq:inted} correctly assigns the value $1$ to $A[t, \eta]$. 

\paragraph*{Introduce edge node-Backward direction.} Assume that the right hand side of \ref{eq:inted} evaluates to $1$. Notice that there are only two possible ways for this to happen:  either $A[t', \eta] = 1$ or for some $i \in [n]$, $\eta$ is $(i, zz')$-consistent and $A[t', \eta'] = 1$ for some $\eta' \in C_{t'}(i, zz', \eta)$. 

If $A[t', \eta] = 1$, we then have a $(t', \eta)$-valid allocation $\fn{\phi'}{[n]}{2^{V(G_{t'})}}$. And we can verify that $\phi'$ is indeed a $(t, \eta)$-valid allocation. Thus $A[t, \eta] = 1$, and \ref{eq:inted} is correct in this case. 

Assume that for some $i \in [n]$, $\eta$ is $(i, zz')$-consistent and $A[t', \eta'] = 1$ for some $\eta' \in C_{t'}(i, zz', \eta)$, and assume without loss of generality that $i = 1$. 
Since $\eta$ is $(1, zz')$ consistent, we have $z, z' \in S_1$, $\blk_{\ca{P}_1}(z) = \blk_{\ca{P}_1}(z')$ and $|f_1 (z) - f_1 (z')| = 1$. Assume without loss of generality that $f_1 (z') = 1 + f_1 (z)$. 
Recall that $P$ is the block of $\ca{P}_1$ that contains $z$ and $z'$, and $x \in P \cap R_1$ is the root of $P$.  
Let $\eta' =  ((S'_p, f'_p, (\ca{P}'_p, R'_p))_{p \in [n]}, (w'_{pq})_{p, q \in [n]})$. By the definition of $C_{t'}(1, zz', \eta)$, we have $S' = S_1, f'_1 = f_1, R'_1 = R_1 \cup \set{z'}$, $\ca{P}'_1 = (\ca{P}_1 \setminus P) \cup \set{P_z, P_{z'}}$, where $\set{P_z, P_{z'}}$ is some partition of $P$ such that $x, z \in P_z$ and $z' \in P_{z'}$. 
Since $A[t', \eta'] = 1$, there exists a $(t', \eta')$-valid allocation, say $\fn{\phi}{[n]}{2^{V(G_t)}}$. For $j \in [n]$, let $F'_j$ be a $(t',  \eta', j)$-witness for $\phi(j)$. Let $F_1$ be the graph obtained from $F'_1$ by adding  the edge $zz'$. Notice that $F_1$ is acyclic as $z$ and $z'$ are in different connected components of $F'_1$. Also, $\ca{P}_1$ is indeed the partition of $S_1$ into the connected components of $F_1$. To see that $F_1$ is a $(t, \eta, 1)$-witness for $\phi$, consider $z'' \in P_{z'}$. Notice that $z' \in P_{z'} \cap R'_1$ is the root of $\blk_{\ca{P}'_1}(z'')$ and $x \in P \cap R_1$ is the root of $\blk_{\ca{P}_1}(z'')$. The unique path in $F_1$ from $x$ to $z''$ consists of the unique path in $F'_1$ from $x$ to $z$, followed by the edge $zz'$ and further followed by the unique path in $F'_1 $ from $z'$ to $z''$. We thus have $\dist_{F_1}(x, z'') = \dist_{F'_1}(x, z) + 1 + \dist_{F'_1}(z', z'') = (f'_1 (z) - f'_1 (x)) + 1 + (f'_1 (z'') - f'_1 (z')) = f'_1 (z'') + 1 - (f'_1 (z') - f'_1 (z)) - f'_1 (x) = f'_1 (z'') - f'_1 (x) = f_1 (z'') - f_1 (x)$. Thus condition~\ref{span}\ref{dist} holds. Now consider $z''' \in \pi(1) \setminus X_t = \pi(1) \setminus X_{t'}$. Since $\phi$ is $(t', \eta')$-valid, there exists $y \in R'_1$ such that $\dist_{F'_1}(y, z''') \leq \beta - f'_1 (y)$. If $y \neq z'$, then $y \in R_1$ and $\dist_{F'_1}(y, z''') = \dist_{F_1}(y, z''')$, and hence condition~\ref{span}\ref{beta} holds. So suppose $y = z'$. Then $\dist_{F'_1}(z', z''') \leq \beta - f'_1 (z')$. Again, the unique path in $F_1$ from $x$ to $z'''$ consists of the unique path in $F'_1$ from $x$ to $z$, followed by the edge $zz'$ and further followed by the unique path in $F'_1$ from $z'$ to $z'''$. Thus $\dist_{F_1}(x, z''') = \dist_{F'_1}(x,  z) + 1 + \dist_{F'_1}(z', z''') \leq (f'_1 (z) - f'_1 (x)) + 1 + (\beta - f'_1(z')) = \beta - f'_1 (x) + 1 - (f'_1 (z') - f'_1 (z)) = \beta - f_1 (x)$. Thus condition~\ref{span}\ref{beta} holds as well. 
Therefore, $F_1$ is a $(t, \eta, 1)$-witness for $\phi$. And we can conclude that $\phi$ is a $(t, \eta)$-valid allocation as conditions~\ref{inter} and \ref{value} trivially hold. Hence $A[t, \eta] = 1$, and thus \ref{eq:inted} correctly assigns the value $1$ to $A[t, \eta]$. 

\paragraph*{Join node.} Suppose that $t$ is a join node and let $t', t''$ the children of $t$. Then $X_t = X_{t'} = X_{t''}$, $V(G_{t}) = V(G_{t'}) \cup V(G_{t''})$, $V(G_{t'}) \cap V(G_{t''}) = X_t$, $E(G_{t}) = E(G_{t'}) \cup E(G_{t''})$ and $E(G_{t'}) \cap E(G_{t''}) = \emptyset$. 

Recall that for $\eta' = ((S'_i, f'_i, (\ca{P}'_i, R'_i))_{i \in [n]}, (w'_{ij})_{i, j \in [n]}) \in \ca{C}_{t'}$ and $\eta'' = ((S''_i, f''_i, (\ca{P}''_i, R''_i))_{i \in [n]}, (w''_{ij})_{i, j \in [n]}) \in \ca{C}_{t''}$, we say that $\eta$ is $(\eta', \eta'')$-consistent if (i) for each $i \in [n]$, $S'_i = S''_i = S_i$ and $f'_i = f''_i = f_i$, $(\ca{P}_i, R_i)$ is an acyclic join of $(\ca{P}'_i, R'_i)$ and $(\ca{P}''_i, R''_i)$ and (ii) for $i, j \in [n]$, $w_{ij} = w'_{ij} + w''_{ij} - v_i(S_j)$. Let $\ca{C}_{(t, t'')}(\eta) \subseteq \ca{C}_{t'} \times \ca{C}_{t''}$ be the set of all $(\eta', \eta'') \in  \ca{C}_{t'} \times \ca{C}_{t''}$ such that $\eta$ is $(\eta', \eta'')$-consistent. And we have set
 \[
A[t, \eta] = \bigvee_{(\eta', \eta'') \in \ca{C}_{(t, t'')}(\eta)} (A[t', \eta'] \land  A[t'', \eta'']). \tag{\textsl{Eq. Join}}\label{eq:join}
\]

\paragraph*{Join node-Forward Direction. } Assume that $A[t, \eta] = 1$. We will show that $A[t', \eta'] = 1$ and $A[t'', \eta''] = 1$ for some $(\eta', \eta'') \in C_{(t', t'')}(\eta)$. By the induction hypothesis, $A[t', \eta']$ and $A[t'', \eta'']$ are computed correctly, and thus we can conclude that the right hand side of \ref{eq:join} evaluates to $1$. 

For $i \in [n]$, let $H'_i$ be the subgraph of $H_i$ such that $H'_i$ is the intersection of $H_i$ with $G_{t'}$, i.e., $V(H'_i) = V(H_i) \cap V(G_{t'})$ and $E(H'_i) = E(H_i) \cap E(G_{t'})$. And we define $H''_i$ analogously to be the intersection of $H_i$ with $G_{t''}$. 
Based on $\eta$, $t$, $H'_i$ and $H''_i$, we define $\eta' \in \ca{C}_{t'}$ and $\eta'' \in \ca{C}_{t''}$ such that $A[t', \eta'] = $ and  $A[t'', \eta''] =1$, which will imply that \ref{eq:join} correctly assigns the value $1$ to $A[t, \eta]$. 

Informally, we define $\eta'$ and $\eta''$ based on the components of $H'_i$ and $H''_i$, respectively. Let us just consider $\eta'$. (The definition of $\eta''$ is analogous.) We define $\eta'$ in such a way that $H'_i$ is a $(t', \eta', i)$-witness for the allocation $\fn{\pi'}{[n]}{2^{V(G_{t'})}}$, where $\pi'(i) = \pi(i) \cap V(G_{t'})$ for each $i \in [n]$. 
Fix $i \in [n]$.  
Recall that $H'_i$ is the intersection of $H_i$ with the graph $G_{t'}$. That is, $H'_i$ is a rooted forest and each component of $H'_i$ is a rooted subtree of a connected component of $H_i$. So we define $\eta' = ((S'_i, f'_i, (\ca{P}'_i, R'_i))_{i \in [n]}, (w'_{ij})_{i, j \in [n]})$, where $S'_i = S_i$ and $f'_i = f_i$; $\ca{P}'_i$ is the partition of $S'_i$ into the connected components of $H'_i$; and $R'_i$ is the set of roots of $H'_i$; and $w'_{ij} = v_i(\pi(j) \cap V(G_{t'}))$. And we define $\eta'' = ((S''_i, f''_i, (\ca{P}''_i, R''_i))_{i \in [n]}, (w''_{ij})_{i, j \in [n]})$ analogously; we take $S''_i = S_i$ and $f''_i = f_i$; $\ca{P}''_i$ is the partition of $S''_i$ into the connected components of $H''_i$; and $R''_i$ is the set of roots of $H''_i$; and $w''_{ij} = v_i(\pi(j) \cap V(G_{t''}))$. To show that $\eta'$ and $\eta''$ are well-defined (i.e., $\eta' \in \ca{C}_{t'}$ and $\eta'' \in \ca{C}_{t''}$), we only have to argue that $R'_i \subseteq X_{t'}$ and $R''_i \subseteq X_{t''}$ for each $i \in [n]$ and (ii) $0 \leq w'_{ij}, w''_{ij} \leq W$ for every $i, j \in [n]$. 
Notice that to prove that $R'_i \subseteq X_{t'}$ (resp. $R''_i \subseteq X_{t''}$), it is enough to prove that the root of each connected component of $H'_i$ (resp. $H''_i$) is in $X_{t'}$ (resp. $X_{t''}$). For that, we prove the following claim.   

\begin{claim}
For every $i \in [n]$ and for every connected component $H'$of $H'_i$ (resp. $H''_i$), $x_{H'} \in X_{t'}$ (resp. $x_{H'} \in X_{t''}$), where $x_{H'}$ is the root of $H'$. 
\end{claim}
\begin{proof}
Fix $i \in [n]$. Let $H'$ be a connected component of $H'_i$ and $x_{H'}$ the root of $H'$. Then $H$ is contained in one of the connected components of $H_i$. Let $H$ be the connected component of $H_i$ that contains $H'$ and $x_H$ the root of $H$.  Notice that $x_H \in R_i \subseteq S_i \subseteq X_{t}$. 

Suppose for a contradiction that $x_{H'} \notin X_{t'}$. Then $x_{H'} \neq x_{H}$, as $x_H \in X_{t} = X_{t'}$. Let $P = y_1 y_2\cdots y_{k - 1} y_k$, where $y_1 = x_{H}$ and $y_k = x_{H'}$, be the unique path in $H$ from $x_H$ to $x_{H'}$.  Let $s \in [k]$ be the largest index such that $y_s \in X_{t'}$ and $y_{s + 1} \in V(G_{t'}) \setminus X_{t'}$. Such an index $s$ exists as $y_1 \in X_{t'}$ and $y_k \in V(G_{t'}) \setminus X_{t'}$. Then $y_s \neq x_{H'}$, and $x_{H'}$ is a descendant of $y_s$ in the rooted tree $H$. By the definition of $s$, the vertices $y_{s + 1}, y_{s + 2},\ldots, y_k \notin X_{t'}$; and since $y_k \in V(G_{t'}) \setminus X_{t'}$ and $y_{r + 1} y_{r + 2} \in E(G)$ for every $s  \leq r \leq k -2$, we can conclude that $y_{s + 1}, y_{s + 2},\ldots, y_k \in V(G_{t'}) \setminus X_{t'}$ and $y_{r + 1} y_{r + 2} \in E(G_{t'})$ for every $j \leq r \leq k - 2$. Also, since $y_{s + 1} \in V(G_{t'}) \setminus X_{t'}$ and $y_s y_{s + 1} \in E(G)$, we have $y_s \in V(G_{t'})$ and $y_s y_{s + 1} \in E(G_{t'})$ as well. Thus the sub-path of $P$ from $y_s$ to $y_k$ is completely contained in the graph $G_{t'}$ and hence in $H'$. Hence $x_{H'}$ is a descendant of $y_s$ in the tree $H'$ as well, which contradicts our assumption that $x_{H'}$ is the root of $H'$. 

The case when $H'$ is a connected component of $H''_i$ is identical. 
\end{proof}

Finally, consider $i, j \in [n]$. Notice that as $w'_{ij} = v_i(X)$ and $w''_{ij} = v_i(Y)$ for some $X, Y \subseteq V(G)$, we have $0 \leq w'_{ij}, w''_{ij} \leq W$. We have thus shown that $\eta' \in \ca{C}_{t'}$ and $\eta'' \in \ca{C}_{t''}$. 
We now argue that $(\ca{P}_i, R_i)$ is an acyclic join of $(\ca{P}'_i, R'_i)$ and $(\ca{P}''_i, R''_i)$. 

\begin{claim}\label{claim:acyclic}
For $i \in [n]$, $(\ca{P}_i, R_i)$ is an acyclic join of $(\ca{P}'_i, R'_i)$ and $(\ca{P}''_i, R''_i)$. 
\end{claim}

\begin{proof}
Fix $i \in [n]$. First, recall that that each connected component of $H'_i$ and $H''_i$ is a rooted sub-tree of a connected component of $H_i$. And we defined $R'_i$ and $R''_i$ to be the set of roots of $H'_i$ and $H''_i$, respectively. Recall also that $R_i$ is the set of roots of $H_i$. So if $z \in R_i$, then $z$ is the root of the connected component of $H_i$ that contains $z$. Since each connected component of $H'_i$ and each connected component of $H''_i$ is a rooted sub-tree of $H_i$,  we can conclude that $z$ is the root of the connected components of $H'_i$ and $H''_i$ that contain $z$. Thus $z \in R'_i  \cap R''_i$. Conversely, suppose that $z \in R'_i \cap R''_i$. Then $z$ is the root of a connected component of $H'_i$ and the root of a connected component of $H''_i$. But as $E(H_i) = E(H'_i) \cup E(H''_i)$, there does not exist a vertex $z'$ such that $z$ is a descendant of $z'$ in $H_i$. Hence $z$ is the root of a connected component of $H_i$ and thus $z \in R_i$.  

Now, to prove that $\ca{P}'$is an acyclic join of $\ca{P}'$ and $\ca{P}''_i$, consider  any graphical representations $\mathtt{G}(\ca{P}'_i)$ and $\mathtt{G}(\ca{P}''_i)$ of $\ca{P}'_i$ and $\ca{P}''_i$, respectively. Let $\mathtt{G}(\ca{P}'_i, \ca{P}''_i)$ denote the merge of $\mathtt{G}(\ca{P}''_i)$. We will show that $\mathtt{G}(\ca{P}'_i, \ca{P}''_i)$ is a graphical representation of $\ca{P}_i$, which will prove that $\ca{P}'$ is an acyclic join of $\ca{P}'$ and $\ca{P}''_i$. 

Let us first prove that $\ca{P}_i$ is the partition of $S_i$ into the connected components of $\mathtt{G}(\ca{P}'_i, \ca{P}''_i)$. Consider $z, z' \in S_i (= S'_i = S''_i)$. 

Suppose first that $z$ and $z'$ are in the same block of $\ca{P}_i$. Then $z$ and $z'$ are in the same connected component of $H_i$. Let $P = z_1 z_2 \cdots z_{r - 1} z_r$, where $z_1 = z$ and $z_r = z'$, be the unique path in $H_i$ from $z$ to $z'$. Notice that since $E(H_i)$ is the disjoint union of $E(H'_i)$ and $E(H''_i)$, one of the following three conditions must hold: (i) $E(P) \subseteq E(H'_i)$, in which case $P$ is a path in $H'_i$,  (ii) $E(P) \subseteq E(H''_i)$, in which case $P$ is a path in $H''_i$, or (iii) there exist $1 \leq r_1 < r_2 < \cdots < r_s \leq r$ such that for every $k \in [s - 1]$, the subpath of $P$ from $z_{r_k}$ to $z_{r_{k + 1}}$ is fully contained in $H'_i$ or fully contained in $H''_i$. 
If (i) holds, then $z$ and $z'$ are in the connected component of $H'_i$, which implies that $z$ and $z'$ are in the same block of $\ca{P}'_i$, which implies that $z$ and $z'$ are in the same connected component of $\mathtt{G}(\ca{P}'_i)$, which implies that $z$ and $z'$ are in the same connected component of $\mathtt{G}(\ca{P}'_i, \ca{P}''_i)$. Similar reasoning applies if (ii) holds. 
If (iii) holds, by similar reasoning, for each $k \in [s - 1]$, $z_k$ and $z_{k + 1}$ are in the same connected component of $\mathtt{G}(\ca{P}'_i, \ca{P}''_i)$, which implies that $z$ and $z'$ are in the same connected component of $\mathtt{G}(\ca{P}'_i, \ca{P}''_i)$. 

Suppose now that $z$ and $z'$ are in the same connected component of $\mathtt{G}(\ca{P}'_i, \ca{P}''_i)$. Let $\mathtt{P} = y_1 y_2 \cdots y_{p - 1} y_p$ be a path in $\mathtt{G}(\ca{P}'_i, \ca{P}''_i)$ from $z$ to $z'$, where $y_1 = z$ and $y_p = z'$. Again, since $E(\mathtt{G}(\ca{P}'_i, \ca{P}''_i)) = E(\mathtt{G}(\ca{P}'_i)) \cup E(\mathtt{G}(\ca{P}''_i))$, one of the three conditions must hold: (i) either $E(P) \subseteq E(\mathtt{G}(\ca{P}'_i))$ or (ii) $E(\mathtt{P}) \subseteq E(\mathtt{G}(\ca{P}''_i))$ or (iii) there exist $1 \leq p_1 < p_2 < \cdots < p_q \leq p$ such that for every $k \in [p - 1]$, the subpath of $\mathtt{P}$ from $y_{p_k}$ to $y_{p_{k + 1}}$ is fully contained in $\mathtt{G}(\ca{P}'_i)$ or fully contained in $\mathtt{G}(\ca{P}''_i)$. If (i) holds, then $z$ and $z'$ are in the same connected component of $E(\mathtt{G}(\ca{P}'_i))$, which implies that $z$ and $z'$ are in the same block of $\ca{P}'_i$, which implies that $z$ and $z'$ are in the same connected component of $H'_i$, which implies that $z$ and $z'$ are in the same connected component of $H_i$, which implies that $z$ and $z'$ are in the same block of $\ca{P}_i$. Similar reasoning applies if (ii) holds. If (iii) holds, then again, we can conclude that for every $k \in [p - 1]$, $y_{p_k}$ and $y_{p_{k + 1}}$ are in the same block of $\ca{P}_i$, which implies that $z$ and $z'$ are in the same block of $\ca{P}_i$. 

We have thus shown that $\ca{P}_i$ is indeed the partition of $S_i$ into the connected components of $\mathtt{G}(\ca{P}'_i, \ca{P}''_i)$. To complete the proof, we only need to prove that $\mathtt{G}(\ca{P}'_i, \ca{P}''_i)$ is acyclic. But observe that for every edge $xy \in E(\mathtt{G}(\ca{P}'_i))$, there exists a corresponding $x$-$y$ path in $H'_i$, which is an $x$-$y$ path in $H_i$. And for every edge $xy \in \mathtt{G}(\ca{P}''_i)$, there exists a corresponding $x$-$y$ path in $H''_i$, which is also a path in $H_i$. 
Recall now that $E(\mathtt{G}(\ca{P}'_i, \ca{P}''_i)) = E(\mathtt{G}(\ca{P}'_i)) \cup E(\mathtt{G}(\ca{P}'_i))$. Thus every edge in $\mathtt{G}(\ca{P}'_i, \ca{P}''_i)$ corresponds to a path in $H_i$. We can therefore conclude that for a cycle in $\mathtt{G}(\ca{P}'_i, \ca{P}''_i)$, there exists a corresponding cycle in $H_i$, which contradicts the fact that $H_i$ is acyclic. Hence $\mathtt{G}(\ca{P}'_i, \ca{P}''_i)$ is acyclic. 
\end{proof}

Finally, consider $i, j \in [n]$. Recall that $V(G_{t'}) \cap V(G_{t''}) = X_t$ and $\pi(j) \cap X_t = S_j$. Hence $v_i(\pi(j)) = v_i(\pi(j) \cap V(G_{t'})) + v_i(\pi(j) \cap V(G_{t''})) - v_i(\pi(j) \cap X_t) = w'_{ij} + w''_{ij} - v_i(S_j)$. This fact, along with Claim~\ref{claim:acyclic} and the definitions of $\eta'$ and $\eta''$, shows that $\eta$ is $(\eta', \eta'')$-consistent. That is, $(\eta', \eta'') \in \ca{C}_{(t', t'')}(\eta)$. 

To conclude the arguments, notice that the allocations $\fn{\pi'}{[n]}{2^{V(G_{t'})}}$ and $\fn{\pi'}{[n]}{2^{V(G_{t''})}}$ defined by $\pi'(i) = \pi(i) \cap V(G_{t'})$ and  $\pi''(i) = \pi(i) \cap V(G_{t''})$ for each $i \in [n]$ are $(t', \eta')$-valid and $(t'', \eta'')$-valid, respectively. This follows from the definitions of $\eta'$ and $\eta''$ and the fact that $H'_i$ and $H''_i$ are $(t', \eta', i)$-witness and $(t'', \eta'', i)$-witness for $\pi'$ and $\pi''$, respectively.   Hence, $A[t', \eta'] = 1$ and $A[t'', \eta''] = 1$, and by the induction hypothesis, we can thus conclude that \ref{eq:join} correctly assigns the value $1$ to $A[t, \eta]$.  

\paragraph*{Join node-Backward Direction. } Assume that the right hand side of \ref{eq:join} evaluates to $1$. Then $A[t', \eta'] = 1$ and $A[t'', \eta''] = 1$ for some $(\eta', \eta'') \in \ca{C}_{(t', t'')}(\eta)$. Thus there exist allocations $\fn{\phi'}{[n]}{2^{V(G_{t'})}}$ and $\fn{\phi'}{[n]}{2^{V(G_{t''})}}$ such that $\phi'$ is $(t', \eta')$-valid and $\phi''$ is $(t'', \eta'')$-valid.  Let $\eta' = ((S'_i, f'_i, (\ca{P}'_i, R'_i))_{i \in [n]}, (w'_{ij})_{i, j \in [n]})$ and $\eta'' = ((S''_i, f''_i, (\ca{P}'_i, R'_i))_{i \in [n]}, (w'_{ij})_{i, j \in [n]})$.

We now define an allocation $\fn{\phi}{[n]}{2^{V(G_t)}}$ as follows: $\phi(i) = \phi'(i) \cup \phi''(i)$ for each $i \in [n]$. 
We claim that $\phi$ is $(t, \eta)$-valid. 
Notice first that $\phi$ is indeed an allocation as $\phi(i) \cap \phi(j) = \emptyset$ for every $i, j \in [n]$. This follows from the facts that (i) $V(G_{t'}) \cap V(G_{t''}) = X_t = X_{t'} = X_{t''}$, (ii) $S_i = S'_i = S''_i$ and (iii) $\phi'$ and $\phi''$ are allocations. We can also verify that $\phi(i) \cap X_t = S_i$ for $i \in [n]$. Recall that since $(\eta', \eta'') \in \ca{C}_{(t, t'')}(\eta)$, by the definition of $\ca{C}_{(t', t'')}(\eta)$, we have $w_{ij} = w'_{ij} + w''_{ij} - v_i(S_j)$ for $i, j \in [n]$. And using this fact, we can also verify that $v_i(\phi(j)) = v_i(\phi'(j)) + v_i(\phi''(j)) - v_i(S_j) = w'_{ij} + w''_{ij} - v_i(S_j) = w_{ij}$. 

To prove that $\phi$ is $(t, \eta)$-valid, we now only have to prove that for each $i \in [n]$, there exists a $(t, \eta, i)$-witness for $\phi$. For $i \in [n]$, let $F'_i$ be a $(t', \eta', i)$-witness for $\phi'$ and and $F''_i$ be a $(t'', \eta'', i)$-witness for $\phi''$. And let  $F_i$ be the union of $F_i$ and $F''_i$. That is, $V(F_i) = V(F'_i) \cup V(F''_i)$ and $E(F_i) = E(F'_i) \cup E(F''_i)$. We claim that $F_i$ is a $(t, \eta, i)$-witness for $\phi$. Let us first see that $F_i$ is acyclic. 

\begin{claim}
\label{claim:Fi_acyclic}
For every $i \in [n]$, $F_i$ is acyclic. 
\end{claim}

\begin{proof}
Fix $i \in [n]$. Suppose for a contradiction that $F_i$ is not acyclic. Let $C$ be a cycle in $F_i$. We will argue that corresponding to $C$, there exists a cycle in the merge of $\mathtt{G}(\ca{P}'_i)$ and $\mathtt{G}(\ca{P}''_i)$, which will contradict the fact that $\ca{P}'_i$ is an acyclic join of $\ca{P}'_i$ and $\ca{P}''_i$. 

First, observe that $C$ must contain edges from both $F'_i$ and $F''_i$, i.e., $E(C) \cap E(F'_i) \neq \emptyset$ and $E(C) \cap E(F''_i) \neq \emptyset$. Otherwise, $E(C) \subseteq E(F'_i)$, in which case $C$ is a cycle in $F'_i$, or $E(C) \subseteq E(F''_i)$, in which case $C$ is a cycle in $F''_i$. In either case, we get a contradiction, as $F'_i$ and $F''_i$ are acyclic. Therefore, there exist vertices $z, z', z'' \in V(C)$ such that $zz' \in E(C) \cap E(F'_i)$ and $zz'' \in E(C) \cap E(F''_i)$; we call such a vertex $z$ a flip vertex. Notice now that $C$ contains at least two distinct flip vertices. Otherwise, let $z$ be the unique flip vertex in $C$, and let $N_C(z) = \set{z', z''}$ with $zz' \in E(F'_i)$ and $zz'' \in E(F''_i)$.  Notice that the graph $C - z$ is a $z'$-$z''$ path. Since $z$ is the unique flip vertex in $C$, either $E(C - z) \subseteq E(F'_i)$ or $E(C - z) \subseteq E(F''_i)$. Assume without loss of generality that $E(C - z) \subseteq E(F'_i)$. Let $y \in V(C)$ be the unique vertex such that $y \neq z$ and $yz'' \in E(C)$. Then, by our assumption that $E(C - z) \subseteq E(F'_i)$, we have $yz'' \in E(F'_i)$. Thus we have $y, z'', z \in V(C)$ such that $yz'' \in E(C) \cap E(F'_i)$ and $z''z \in E(C) \cap E(F')$, which implies that $z''$ is a flip vertex, a contradiction to our assumption that $z$ is the unique flip vertex in $C$.  Also, notice that for every flip vertex $z$ in $C$, we have $z \in X_t$, as both $G_{t'}$ and $G_{t''}$ contain edges incident with $z$. 

Fix a flip vertex $z$ and fix a traversal of $C$ starting at $z$. Let $(z_1, z_2, \ldots, z_{r-1}, z_r)$ be the ordered sequence of flip vertices we encounter in this traversal, where $z_1= z_r = z$. Then $z_j \in X_t$ for every $j \in [r-1]$ and since $C$ contains at least two flip vertices, we have $r - 1 \geq 2$. 
For $j \in [r - 1]$, let $P_j$ be the $z_j$-$z_{j + 1}$ path in $C$ such that $V(P) \setminus \set{z_j, z_{j + 1}}$ does not contain any flip vertex. 
Notice that for each $j \in [r - 1]$, either $P_j$ is a path in $F'_i$ or $P_j$ is a path in $F''_i$. Hence either $\mathtt{G}(\ca{P}'_i)$ contains a $z_j$-$z_{j + 1}$ path or $\mathtt{G}(\ca{P}''_i)$ contains a $z_j$-$z_{j + 1}$ path. In either case, the merge of  $\mathtt{G}(\ca{P}'_i)$ and $\mathtt{G}(\ca{P}''_i)$ contains a $z_j$-$z_{j + 1}$ path for every $j \in [r - 1]$. And these paths together constitute a cycle in the merge of  $\mathtt{G}(\ca{P}'_i)$ and $\mathtt{G}(\ca{P}''_i)$.  
\end{proof}

We have thus shown that $F_i$ is indeed a spanning forest for $G_{t}[\pi(i)]$. To prove that $F_i$ is a $(t, \eta, i)$-witness for $\phi$, we have to show that conditions \ref{span}\ref{comp}--\ref{span}\ref{beta} hold. And for that, we now prove the following claims. 

\begin{claim}
\label{claim:Fi_comp}
For every $i \in [n]$, $\ca{P}_i$ is the partition of $S_i$ into the connected components of $F_i$. 
\end{claim}
\begin{proof}
Fix $i \in [n]$. Let $z, z' \in S_i$. Assume first that $\blk_{\ca{P}_i}(z) = \blk_{\ca{P}_i}(z')$. Since $\ca{P}_i$ is the partition of $S_i$ into the connected components of $\merge(\mathtt{G}(\ca{P}'_i), \mathtt{G}(\ca{P}''_i))$, $z$ and $z'$ are in the same connected component of $\merge(\mathtt{G}(\ca{P}'_i), \mathtt{G}(\ca{P}''_i))$. Let $\mathtt{P} = z_1 z_2 \cdots z_r$, where $z_1 = z$ and $z_r = z'$, be the unique $z$-$z'$ path in $\merge(\mathtt{G}(\ca{P}'_i), \mathtt{G}(\ca{P}''_i))$. We can break $\mathtt{P}$ into sub-paths such that each subpath is either fully contained in $\mathtt{G}(\ca{P}'_i)$ or fully contained in $\mathtt{G}(\ca{P}''_i)$. That is, let $1 = j_1 < j_2 < \cdots j_s = r$ be such that for every $p \in [s - 1]$, the $z_{j_p}$-$z_{j_{p + 1}}$ subpath of $\mathtt{P}$ is fully contained in $\mathtt{G}(\ca{P}'_i)$ or fully contained in $\mathtt{G}(\ca{P}''_i)$. But $\ca{P}'_i$ and $\ca{P}''_i$ are the partitions of $S_i = S'_i = S''_i$ into the connected components of $F'_i$ and $F''_i$ respectively. Thus for each $p \in [s - 1]$, there is an $z_{j_p}$-$z_{j_{p + 1}}$ path in either $F'_i$ or $F''_i$. In either case, $F_i$ thus contains a $z_{j_p}$-$z_{j_{p + 1}}$ path for every $p \in [s - 1]$, which implies that $F_i$ contains a path between $z_{j_1} = z$ and $z_{j_s} = z'$. 

Conversely, assume that $z$ and $z'$ are in the same connected component of $F_i$. By Claim~\ref{claim:Fi_acyclic}, $F_i$ is acyclic. And hence $F_i$ contains a unique path from $z$ to $z'$, say $P$. We can again break $P$ into subpaths such that each subpath is fully contained in $F'_i$ or fully contained in $F''_i$, which will imply that $\merge{\mathtt{G}(\ca{P}'), \mathtt{G}(\ca{P}''_i})$ contains a $z$-$z'$ path. Since $\ca{P}_i$ is an acyclic join of $\ca{P}'_i$ and $\ca{P}''_i$, $\merge{\mathtt{G}(\ca{P}'), \mathtt{G}(\ca{P}''_i})$ is a graphical representation of $\ca{P}_i$. And hence we can conclude that $\blk_{\ca{P}_i}(z) = \blk_{\ca{P}_i}(z')$.  
\end{proof}

\begin{claim}
\label{claim:dist}
For every $i \in [n]$ and $z \in S_i$, $f_i(z) = f_i(x) + \dist_{F_i}(x, z)$, where $x \in R_i$ is the root of $\blk_{\ca{P}_i}(z)$.  
\end{claim}

\begin{proof}
Fix $i \in [n]$ and $z \in S_i$. Let $x \in R_i$ be the root of $\blk_{\ca{P}_i}(z)$.  If $x = z$, then $\dist_{F_i}(x, z) = 0$ and the claim trivially holds. So assume that $x \neq z$. Let $P = z_1 z_2\cdots z_r$ be the unique path in $F_i$ from $x$ to $z$, where $z_1 = x$ and $z_r = z$. For $p, q \in [r]$, let $P_{p, q}$ denote the subpath of $P$ from $z_p$ to $z_q$. Recall that as $E(F_i)$ is the disjoint union of $E(F'_i)$ and $E(F''_i)$, we can ''break'' $P$ into subpaths such that each subpath is fully contained in $F'_i$ or fully contained in $F''_i$. Assume without loss of generality that $z_1 z_2 \in E(F'_i)$. Let $1 = j_1 < j_2 < \cdots j_{s - 1} < j_s = r$ be such that $P_{j_p}, j_{p + 1}$ is fully contained in $E(F'_i)$ for every $p \in \tsub{[s - 1]}{od}$ and $P_{j_p, j_{p + 1}}$ is fully contained in $E(F''_i)$ for every $p \in \tsub{[s - 1]}{ev}$. (Recall that $\tsub{[s - 1]}{od}$ and $\tsub{[s - 1]}{ev}$ are the sets of odd numbers and even numbers in $[s - 1]$, respectively.) Observe first that $z_{j_1}, z_{j_2}, \ldots, z_{j_s} \in X_t = X_{t'} = X_{t''}$. 

Let us first assume that$z_{j_p} \in R'_i$ for every $p \in \tsub{[s - 1]}{od}$ and $z_{j_p} \in R''_i$ for every $p \in \tsub{[s - 1]}{ev}$,  and complete the proof of the claim. Consider $p \in \tsub{[s - 1]}{od}$. Since $P_{j_p, j_{p + 1}}$ is a $z_{j_p}$-$z_{j_{p + 1}}$ path in $F'_i$, $z_{j_p}$ and $z_{j_{p + 1}}$ are in the same connected component of $F'_i$. And since $\ca{P}'_i$ is the partition of $S'_i = S_i$ into the connected components of $F'_i$ and since $z_{j_p} \in R'_i$, we can conldue that $z_{j_p}$ is the root of $\blk_{\ca{P}'_i}(z_{j_{p + 1}})$. Hence by condition~\ref{span}\ref{dist}, we have $f_i(z_{j_{p + 1}}) = f_i(z_{j_p}) + \dist_{F'_i} (z_{j_p}, z_{j_{p + 1}})$. Since $P_{j_{p}, j_{p + 1}}$ is the unique $z_{j_p}$-$z_{j_{p + 1}}$ path in $F'_i$ and hence in $F_i$, we have $\dist_{F'_i} (z_{j_p}, z_{j_{p + 1}}) = \dist_{F_i} (z_{j_p}, z_{j_{p + 1}})$. Thus $f_i(z_{j_{p + 1}}) = f_i(z_{j_p}) + \dist_{F_i} (z_{j_p}, z_{j_{p + 1}})$. Symmetric arguments hold for $p \tsub{[s - 1]}{ev}$. We thus have $\dist_{F_i} (z_{j_p}, z_{j_{p + 1}}) = f_i(z_{j_{p + 1}}) - f_i(z_{j_p})$ for every $p \in [s - 1]$. 

Now, notice that $\dist_{F_i}(x, z) = \dist_{F_i}(z_{j _1}, z_{j_s}) = \sum_{p = 1}^{s -1}\dist_{F_i}(z_{j_p}, z_{j_{p + 1}}) = \sum_{p = 1}^{s -1} (f_i(z_{j_{p + 1}}) - f_i(z_{j_p})) = f_i(z_{j_s}) - f_i(z_{j_1}) = f_i(z) - f_i(x)$. That is, $f_i(z) = f_i(x) + \dist_{F_i}(x, z)$. 

So, to complete the proof of the claim, we only have to prove that $z_{j_p} \in R'_i$ for every $p \in \tsub{[s - 1]}{od}$ and $z_{j_p} \in R''_i$ for every $p \in \tsub{[s - 1]}{ev}$.  First, for $p = 1$, we have $z_{j_1} = z_1 = x \in R_i = R'_i \cap R''_i$; thus, in particular, $z_1 \in R'_i$. Consider $p = 2$.  
We have to prove that $z_{j_2} \in R''_i$. Suppose for a contradiction that $z_{j_2} \notin R''_i$. Then there exists a vertex $x_1 \in R''_i$ such that $x_1$ is the root of $\blk_{\ca{P}''_i}(z_{j_2})$. Since $z_{j_2} \neq x_1$, we have $\dist_{F''_i}(x_1, z_{j_2}) > 0$; and since $F''_i$ is a $(t'', \eta'', i)$-witness for $\phi''$, by condition~\ref{span}\ref{dist}, we have $f_i(z_{j_2}) = f_i(x_1) + \dist_{F''_i}(x_1, z_{j_2}) > f_i(x_1)$.  
Since $\ca{P}''_i$ is the partition of $S''_i = S_i$ into the connected components of $F''_i$,  $x_1$ and $z_{j_2}$ are in the same connected component of $F''_i$. Let $Q_{z_{j_2}, x_1}$ be the unique path in $F''_i$ from $z_{j_2}$ to $x_1$. Notice then that the paths $P_{j_1, j_2}$ and $Q_{z_{j_2}, x_1}$ intersect only in $\set{z_{j_2}}$, for otherwise $F_i$ would contain a cycle. In particular, $x_1 \neq x (= z_{j_1})$. Let $Q_1$ denote the $x$-$x_1$ path in $F_i$ formed by $P_{j_1, j_2}$ followed by $Q_{z_{j_2}, x_1}$. Notice that $Q_1$ is the unique path in $F_i$ from $x = z_{j_1}$ to $x_1$. In particular, $x$ and $x_1$ are in the same connected component of $F_i$. And by Claim~\ref{claim:Fi_comp}, we have $\blk_{\ca{P}_i}(x) = \blk_{\ca{P}_i}(x_1)$. Since $x$ is the unique element of  $\blk_{\ca{P}_i}(x)$ such that $x \in R_i$, we can conclude that $x_1 \notin R_i$. Thus, $x_1 \notin R_i$ and $x_1 \in R''_i$; this, along with the fact that $R_i = R'_i \cap R''_i$, implies that $x_1 \notin R'_i$. But then there exists $x_2 \in R'_i$ such that $x_2$ is the root of $\blk_{\ca{P}'_i}(x_1)$.  Since $x_2 \neq x_1$, we have $\dist_{F'_i}(x_2, x_1) > 0$; and since $F'_i$ is a $(t', \eta', i)$-witness for $\phi''$, by condition~\ref{span}\ref{dist}, we have $f_i(x_1) = f_i(x_2) + \dist_{F'_i}(x_2, x_1) > f_i(x_2)$.  Since $\ca{P}'_i$ is the partition of $S'_i = S_i$ into the connected components of $F'_i$,  $x_2$ and $x_1$ are in the same connected component of $F'_i$. Let $Q_{x_1, x_2}$ be the unique path in $F'_i$ from $z_{j_2}$ to $x_1$ . Notice that  the paths $Q_1$ and $Q_{x_1, x_2}$ intersect only in $\set{x_1}$, for otherwise $F_i$ would contain a cycle. In particular, $x \neq x_2$. Let $Q_2$ denote the $x$-$x_2$ path in $F_i$ formed by $Q_1$ followed by $Q_{x_1 x_2}$. Thus $x$ and $x_2$ are in the same connected component of $F_i$. Since $x \in R_i$, we can conclude that $x_2 \notin R_i$. Again, as $x_2 \in R'_i$ and $x_2 \notin R_i = R'_i \cap R''_i$, we can conclude that $x_2 \notin R''_i$. Then, as before, there exists $x_3 \in R''_i$ such that  $x_3$ is the root of $\blk_{\ca{P}''_i}(x_2)$ and so on. We thus get an infinite sequence of vertices $x_1, x_2, x_3,\ldots$ such that $x_q \in S_i$ for every $q$ and $f_i(x_1) > f_i(x_2) > f_i(x_3) > \cdots$. But this is not possible as $0 \leq f_i(x_q) \leq \beta$ for every $q$. We thus have a contradiction to our assumption that $z_{j_2} \notin R''_i$.  
Symmetric arguments hold for every $p > 2$, and we can show that $z_{j_p} \in R'_i$ for every $p \in \tsub{[s -1]}{od}$ and $z_{j_p} \in R''_i$ for every $p \in \tsub{[s -1]}{ev}$. 
\end{proof}

\begin{claim}
\label{claim:beta}
For every $i \in [n]$ and $z \in \phi(i) \setminus X_t$, there exists $y \in R_i$ such that $\dist_{F_i}(y, z) \leq \beta - f_i(y)$. 
\end{claim}
\begin{proof}
Fix $i \in [n]$ and $z \in \phi(i) \setminus X_t$. Since $\phi(i) = \phi'(i) \cup \phi''(i)$, either $z \in \phi'(i)$ or $z \in \phi''(i)$. Assume without loss of generality that $z \in \phi'(i)$. Then, since $F'_i$ is a $(t', \eta', i)$-witness for $\phi'$, ther exists $y' \in R'_i$ such that $\dist_{F'_i}(y', z) \leq \beta - f_i(y')$. If $y' \in R_i$, then the claim holds. So suppose $y' \notin R_i$. Let $y  \in R_i$ be the root of $\blk_{\ca{P}_i}(y)$. By Claim~\ref{claim:dist}, we have $f_i(y') = f_i(y) + \dist_{F_i}(y, y')$. Let $P$ be the unique path in $F'_i$ from $y'$ to $z$ and $Q$ the unique path in $F_i$ from $y$ to $y'$. Notice that the $y$-$z$ path in $F_i$ formed by $Q$ and $P$ is the unique $y$-$z$ path in $F_i$. We thus have $\dist_{F_i}(y, z) = \dist_{F_i}(y, y') + \dist_{F_i}(y', z) = (f_i(y') - f_i(y)) + \dist_{F'_i} (y', z) \leq f_i(y') - f_i(y) + \beta - f_i(y') = \beta - f_i(y)$. 
\end{proof}

Claims~\ref{claim:Fi_acyclic}-\ref{claim:beta} show that $F_i$ is indeed a $(t, \eta, i)$-witness for $\phi$. We can thus conclude  that $\phi$ is a $(t, \eta)$-valid allocation and hence $A[t, \eta] = 1$. And \ref{eq:join} correctly assigns the value $1$ to $A[t, \eta]$. 
\end{sloppypar}
\end{proof}
\end{toappendix}

\section{Discussion}
We proposed an alternative for connectivity in the fair division of graphs. Our results demonstrate that we can achieve tractability results under such alternative constraints. Our results also demonstrate that a number of tools and deep results from graph theory could find applications in fair division; this is still an under-explored direction. Compact allocations, we believe, can be a compelling alternative for connected allocations, and warrant further study. As we noted in Section~\ref{sec:intro}, compactness generalises the connectivity constraint as well. 

This work leaves several questions open. First, it would be interesting to see if our algorithmic results translate to the strongly compact setting. Second, all our strong \NPH ness results rely on a arbitrary number of agents and non-identical valuations. It would be interesting to see if more restricted hardness results can be proved. 

Apart from connectivity and compactness, there may be other structured bundles that are worth investigating. In situations where the graph models a spatial relationship among the items---rooms in a building, for example---connectivity or compactness may be the most appropriate choices for the bundles. But as we noted in Section~\ref{sec:intro}, graphs can model a variety of relationships among the items. It is pertinent to consider what other graph theoretic restrictions on the bundles would be useful when the graph models non-spatial relationships among the items, say similarity of the items based on rather abstract features. 
We hope this work will trigger such questions.

\paragraph*{Acknowledgement.} This work was supported by the Engineering and Physical Sciences Research Council (EPSRC) [EP/V032305/1]. The author thanks the anonymous referees and Jessica Enright for their comments on the manuscript. 
\addcontentsline{toc}{section}{References}
\bibliographystyle{plainnat} 
\bibliography{arxiv-fd_compact}


\end{document}